\documentclass[journal]{IEEEtran}
\usepackage{array,url}
\usepackage[geometry]{ifsym}
\usepackage{amssymb,amsfonts,amsmath,amsthm,amscd,mathrsfs, mathtools}
\usepackage[textsize=footnotesize]{todonotes}
\usepackage{graphicx, float, color, comment, enumerate, verbatim, fixmath, bbm, latexsym, scrextend}
\usepackage[nice]{nicefrac}
\usepackage[hidelinks]{hyperref}

\makeatletter
\newcommand\footnoteref[1]{\protected@xdef\@thefnmark{\ref{#1}}\@footnotemark}
\makeatother

\usepackage{cleveref}
\crefname{apdx}{Appendix}{Appendices}

\usepackage{pgf,tikz}
\usepackage{pgfplots}

    \newcommand{\rmd}{\mathrm{d}}
\newcommand{\bbE}{\mathbb{E}}    \newcommand{\rme}{\mathrm{e}}

\newcommand{\bbN}{\mathbb{N}}    
  \newcommand{\rmO}{\mathrm{O}}  
\newcommand{\bbP}{\mathbb{P}}    
    
\newcommand{\bbR}{\mathbb{R}}

      \newcommand{\sfc}{\mathsf{c}}

      \newcommand{\sfs}{\mathsf{s}}

\newcommand{\cA}{\mathcal{A}}   
\newcommand{\cB}{\mathcal{B}}   
\newcommand{\cC}{\mathcal{C}} \newcommand{\scrC}{\mathscr{C}}  
 \newcommand{\scrD}{\mathscr{D}}

\newcommand{\cP}{\mathcal{P}}

\newcommand{\cT}{\mathcal{T}}

\newcommand{\cX}{\mathcal{X}}   
\newcommand{\cY}{\mathcal{Y}}  \newcommand{\mfrY}{\mathfrak{Y}} 
   
\newtheoremstyle{mystyle}% name
{}% space above
{}% space below
{\itshape}% body font
{}% indent amount
{\bfseries}% theorem head font
{}% punctuation after theorem head
{.5em}% space after theorem head
{}% theorem head spec

\newtheoremstyle{remark}% name
{}% space above
{}% space below
{}% body font
{}% indent amount
{\itshape}% theorem head font
{}% punctuation after theorem head
{.5em}% space after theorem head
{}% theorem head spec

\makeatletter
\def\thmhead@plain#1#2#3{%
  \thmname{#1}\thmnumber{\@ifnotempty{#1}{ }\@upn{#2.}}%
  \thmnote{ \small\textsf{\the\thm@notefont\textit{#3}.}}}
\let\thmhead\thmhead@plain
\makeatother

\theoremstyle{mystyle}
\newtheorem{theorem}{Theorem}%[section]
\theoremstyle{mystyle}
\newtheorem{lemma}{Lemma}%[section]
\theoremstyle{mystyle}
\newtheorem{prop}{Proposition}%[section]
\theoremstyle{mystyle}
\newtheorem{corollary}{Corollary}%[thm]
\theoremstyle{mystyle}
\newtheorem{definition}{Definition}%[section]
\theoremstyle{mystyle}
\theoremstyle{mystyle}
\newtheorem{exa}{Example}%[section]
\theoremstyle{mystyle}
\theoremstyle{mystyle}
\theoremstyle{mystyle}
\theoremstyle{mystyle}
\theoremstyle{remark}
\newtheorem{rem}{Remark}%[section]

\usepackage{subcaption}
\newcommand\independent{\protect\mathpalette{\protect\independent}{\perp}}
\def\independent#1#2{\mathrel{\rlap{$#1#2$}\mkern2mu{#1#2}}}

\def\squarebox#1{\hbox to #1{\hfill\vbox to #1{\vfill}}}

\newcommand{\ds}{\displaystyle}

\newcommand{\induced}{P_{Y^n|\scrC_M^n}}
\newcommand{\ccinduced}{P_{\breve Y^n|\scrD_M^n}}
\newcommand{\condisset}{\cP_n(\cX|Q_{\bar Y}; P_{\bar X})}

\DeclareMathOperator{\Var}{Var}

\makeatletter
\newcommand{\bigger}{\bBigg@{3}}
\newcommand{\vast}{\bBigg@{4}}
\newcommand{\Vast}{\bBigg@{5}}
\newcommand{\Gigantic}{\bBigg@{8}}

\makeatother

\allowdisplaybreaks 

\renewcommand{\qedsymbol}{$\blacksquare$}
\begin{document}
\title{Exact Exponent for Soft Covering}
\date{}
\author{Semih Yagli,~\IEEEmembership{Student Member, IEEE} and Paul Cuff,~\IEEEmembership{Member, IEEE}%
\thanks{Semih Yagli is with the Electrical Engineering Department, Princeton University, Princeton, NJ 08544. Paul Cuff is with Renaissance Technologies LLC, Long Island, NY 11733. E-mail: {\em syagli@princeton.edu}, {\em pcuff@rentec.com}. This work has been supported by the National Science Foundation under grant CCF-1350595 and the Air Force Office of Scientific Research under grant FA9550-15-1-0180. Part of this work was presented at ISIT 2018, Vail, CO.}
}
\maketitle

\begin{abstract}
This work establishes the exact exponents for the soft-covering phenomenon of a memoryless channel under the total variation metric when random (i.i.d. and constant-composition) channel codes are used.  The exponents, established herein, are strict improvements in both directions on bounds found in the literature.  This complements the recent literature establishing the exact exponents under the relative entropy metric; however, the proof techniques have significant differences, and thus, neither result trivially implies the other.

The found exponents imply new and improved bounds for various problems that use soft-covering as their achievability argument, including new lower bounds for the resolvability exponent and the secrecy exponent in the wiretap channel.


\end{abstract}

%%%%%%%%%

\section{Introduction}

\IEEEPARstart{T}{he} soft-covering lemma is a strong and useful tool commonly used for proving achievability results for information theoretic security, resolvability, channel synthesis and lossy source coding. The roots of the soft-covering concept originate back to Wyner \cite[Theorem 6.3]{wyner1975common} where he developed this tool with the aim of proving achievability in his work on the common information of two random variables. Coincidentally, the most widespread current application of soft-covering is security proofs in wiretap channels, e.g., \cite{bloch2013strong}, which Wyner also introduced in that same year in \cite{wyner1975wire} but apparently did not see how soft-covering applied. 

The soft-covering lemma states that given a stationary memoryless channel $P_{Y^n|X^n}$ with stationary memoryless input distribution $P_{X^n}$ yielding an output distribution $P_{Y^n}$, the distribution $P_{Y^n|\scrC^n_M}$ induced by instead selecting a sequence $X^n$ at random from a codebook $\scrC^n_M$ and passing it through the channel, see \Cref{def:induced output distribution}, will be a good approximation\footnote{When the relative entropy or the total variation distance is used as the distinction measure, cf. \cite{wyner1975common} and \cite{han1993approximation}, respectively.} of the output distribution $P_{Y^n}$ in the limit as $n$ goes to infinity so long as the codebook is of size $M$ greater than $\exp(nR)$ where $R$ is greater than the single-shot mutual information between the input and output, i.e., $R>I(P_X, P_{Y|X})$. In fact, the aforementioned codebook $\scrC^n_M$ can be chosen quite carelessly, e.g., by drawing each codeword independently from $P_{X^n}$ or by drawing each codeword uniformly at random from the type class $\cT^n_{P_X}$. 

The concept of soft-covering is fundamentally related to that of channel resolvability \cite{han1993approximation}, in that the former is a property of random codebooks while the latter is the fundamental limit of optimal codebooks. As a matter of fact, soft-covering establishes the direct proof (also known as ``achievability") for resolvability. Furthermore, given the chronology of the literature, the resolvability problem can be viewed as a question about soft-covering---how much better can an optimized codebook match an output distribution than a random codebook?  To the first order, the answer is that it does no better.

In the literature, various versions of the soft-covering lemma use various distinctness measures on distributions (commonly relative entropy or total variation distance, see \Cref{def:relative entropy,def:total variation distance}) and claim that the distance between the induced distribution $\induced$ and the desired distribution $P_{Y^n}$ vanishes in expectation over the random selection of the codebook $\scrC^n_M$. Regarding the most notable contributions, \cite{han1993approximation} studies the fundamental limits of soft-covering under the name of ``resolvability", \cite{cuff2009communication} develops the lemma calling it a ``cloud mixing" lemma, \cite{hayashi2006general} provides achievable rates of exponential convergence, \cite{cuff2013distributed} improves the exponent and extends the framework, \cite{ahlswede2002strong} and \cite[Chapter 16]{wilde2013quantum} refer to soft-covering simply as ``covering" in the quantum context, \cite{winter2005secret} refers to it as a ``sampling lemma" and points out that it holds for the stronger metric of relative entropy, \cite{hou2014effective} gives a direct proof of the relative entropy result, and \cite{cuff2015stronger} and \cite{cuff2016soft} move away from expected value analysis and show that a random codebook achieves soft-covering phenomenon with a doubly exponentially high probability under the relative entropy measure and total variation distance, respectively.

The motivation of this work is to complement the results of Parizi \emph{et al.} \cite[Theorem 4]{parizi2017exact}, and Yu and Tan \cite[Theorem 3]{yutanrenyi}, where they pin down the exact soft-covering exponents in the expected value analysis of the relative entropy, and of the R\'enyi divergence of order $\alpha\in (0,1)\cup(1,2)$%
%between the random i.i.d. codebook $\scrC_M^n $ induced distribution $\induced$ and the desired output distribution $P_{Y^n}$
 , respectively. In this paper, we first highlight that the \emph{total variation distance} between the i.i.d. codebook induced distribution $\induced$ and the desired output distribution $P_{Y^n}$ concentrates to its expected value with doubly exponential certainty \cite[Theorem 31]{lcv2017egamma}.
 %, a property which is neither satisfied by relative entropy nor by the R\'enyi divergence
 The first main result of this paper, stated in \Cref{thm:main}, provides the exact soft-covering exponent for the expected value of the \emph{total variation distance} between $\induced $ and $P_{Y^n} $. Next, we consider the setting when the random codebook is restricted to contain codewords of the same empirical distributions. Calling this the \emph{random constant-composition codebook} and denoting it by $\scrD_M^n$, in \Cref{lem:cc:concentration of TV distance btw ccinduced and ccoutput}, we show the counterpart of \cite[Theorem 31]{lcv2017egamma}. In other words, we prove the fact that the total variation distance between the constant-composition induced distribution $\ccinduced $ and the desired output distribution $R_{\breve Y^n} $ concentrates to its expected value in a doubly exponential fashion as well. Finally, we present our second main result in \Cref{thm:cc:main}, providing the exact soft-covering exponent for the expected value of the total variation distance between $\ccinduced $ and $R_{\breve Y^n} $. The exponents for soft-covering, established in this work, provide improved lower bounds on the exponents for resolvability. It may be the case that use of an optimized codebook provides better exponents, even though this work provides the exact exponents (both upper and lower bounds) for the random codebooks.

In the remainder of this paper, \Cref{sec:notation} establishes the basic notation and definitions adopted throughout, and \Cref{sec:main} highlights \cite[Theorem 31]{lcv2017egamma}, shows its counterpart in the constant-composition setting, and states the main results of this paper, namely, the exact soft-covering exponents for the cases of random i.i.d. codebooks and random constant-composition codebooks, along with a number of remarks. \Cref{sec:lower bound,sec:upper bound} prove the lower and upper bound directions of the main result in \Cref{thm:main} together with the remarks of how one would recover the proof of \Cref{thm:cc:main} based on the proof provided. As \Cref{sec:alternative representations} proves alternative expressions for the exact soft-covering exponents, \Cref{sec:comparisons} compares the exact exponents to their previously discovered lower bounds, and finally, \Cref{apdx:proof of TV concentration McDiarmid,apdx:main,apdx:asymptotic exponents,apdx:optimizations over types in the limit,%
%apdx:Exponential Vanishing Argument,% EXPONENTIAL VANISHING ARGUMENT APPENDIX REMOVED AFTER REVIEWER COMMENT 
apdx:alternative representation and comparisons} provide the lemmas and corollaries that are invoked in the main proofs while \Cref{%apdx:cc:full proof,
apdx:finite block-length results} provides
% the full proof of \Cref{thm:cc:main} and 
 the finite block-length results that appear as a byproduct of our proof technique%
%, respectively
 .
\section{Notation and Definitions}\label{sec:notation}
This section introduces the basic notation and fundamental concepts as well as several definitions and properties to be used in the sequel.

Given a finite alphabet $\cX$, let $\cP(\cX)$ denote the set of all distributions defined on it. For a random variable $X$ on $\cX$, a central measure in information theory, namely the amount of \emph{information} provided by $X=x\in \cX$, is defined as follows.

\begin{definition}[Information]
Suppose $X\sim P_X\in\cP(\cX)$, the \emph{information} in $X=x\in \cX$ is\footnote{Unless otherwise stated, logarithms and exponentials are of arbitrary (but matching) bases throughout this paper.}
\begin{align}
\imath^{}_{P_X}(x)= \log\frac{1}{P_X(x)} \text{.}
\end{align}
\end{definition}

When we investigate the interplay between two random variables $(X,Y) \in \cX\times \cY$, the amount of information provided by $Y=y$ after observing $X=x$ is measured by \emph{conditional information}.

\begin{definition}[Conditional Information] Suppose that given $X=x$, $Y\sim P_{Y|X=x}$. The conditional information provided by $Y=y$, given $X=x$, is
\begin{align}
  \imath^{}_{P_{Y|X}}(y|x) = \log \frac{1}{P_{Y|X}(y|x)} \text{.}
\end{align}
  
\end{definition}

Notice that information $\imath^{}_{P_X}(x)$ is a deterministic function depending on the random variable $X\sim P_X$ only through its probability mass function. If one considers the average of $ \imath^{}_{P_X}(X)$, the random information provided by $X$, this gives rise to the definition of the most famous information theoretic quantity, \emph{entropy}, which is defined next. 

\begin{definition}[Entropy]
The \emph{entropy} of a discrete random variable $X\sim P_X \in \cP(\cX) $ is the average information provided by $X$, that is
\begin{align}
H(P_X) = \bbE[\imath^{}_{P_X}(X)] \text{.}
\end{align}
\end{definition}

When the distribution of the discrete random variable $X$ is clear from the context, it is customary to denote its entropy by $H(X)$. Given $(X, Y) \sim P_{X|Y} P_Y $ the average entropy remaining in $X$ when given $Y$ is measured by \emph{conditional entropy} which is defined as follows.  
\begin{definition}[Conditional Entropy]
	Suppose that $(X,Y)\sim P_{X|Y} P_Y \in \cP(\cX\times \cY)$. The \emph{conditional entropy} of a discrete random variable $X$ given $Y$ is 
	\begin{align}
		H(X|Y) &= \bbE[\imath^{}_{P_{X|Y}}(X|Y) ] \\
		 &= \sum_{b\in \cY} H(P_{X|Y=b}) P_Y(b) \text{.}
	\end{align}
\end{definition} 

Given two random variables $X$ and $\widetilde X$ on the same alphabet $\cX$, the information provided by the event $X=x$ relative to the information provided by $\widetilde X = x$ is captured by \emph{relative information}, whose definition is given below.
\begin{definition}[Relative Information]
	Let $P_X$ and $Q_X$ be two distributions in $\cP(\cX)$, the \emph{relative information} in $x\in \cX$ according to $(P_X, Q_X)$ is 
	\begin{align}
		\imath^{}_{P_X\|Q_X}(x) = \log \frac{P_X(x)}{Q_X(x)} \text{.}
	\end{align}
\end{definition}

Although it neither satisfies symmetry nor the triangular inequality, widely used in probability theory, statistical inference, and physics, the expectation of the random variable $\imath^{}_{P_X\| Q_X}(X)$ when $X\sim P_X$ is a non-negative measure of distinctness between $P_X$ and $Q_X$. This expectation is \emph{relative entropy}, defined as follows. 

\begin{definition}[Relative Entropy]\label{def:relative entropy}
Suppose $P_X$ and $Q_X$ are two distributions in $\cP(\cX)$ such that $P_X$ is absolutely continuous with respect to $Q_X$, i.e., $P_X \ll Q_X $. The \emph{relative entropy} between $P_X$ and $Q_X$ is
\begin{align}
	D(P_X\|Q_X) = \bbE[\imath^{}_{P_X \| Q_X}(X)]\text{,} 
\end{align}
where $X\sim P_X$. If $P_X \not\ll Q_X$, then $D(P_X\| Q_X) = +\infty $. 
\end{definition}
Several key properties of the relative entropy, including but not limited to its non-negativity and convexity, can be found in standard information theory books such as \cite{cover2012elements, csiszar2011information}. 

We define a conditional version of the relative entropy as below. 
\begin{definition}[Conditional Relative Entropy]
	Let $P_Y\in \cP(\cY)$ and suppose that $P_{X|Y}\colon \cY \to \cX $ and $Q_{X|Y}\colon \cY \to \cX $ are two conditional distributions on the finite alphabet $\cX$. The conditional relative entropy between $P_{X|Y}$ and $Q_{X|Y} $ given $Y\sim P_Y$ is defined as  
	\begin{align}
		&D( P_{X|Y} \| Q_{X|Y} |P_Y ) \nonumber \\ 
		 &\quad = D(P_{X|Y} P_Y \| Q_{X|Y} P_Y ) \\
		&\quad =\sum_{b\in \cY} P_Y(b) D(P_{X|Y=b}\|Q_{X|Y=b}) \text{.}
	\end{align}
\end{definition}

As mentioned above, since $D(P_X\|Q_X)$ does not satisfy all of the metric axioms, it is not a proper measure of distance between $P_X$ and $Q_X$ in the topological sense. One such metric that measures topological distance between two distributions $P_X$ and $Q_X$ is \emph{total variation distance} which is defined next.
\begin{definition}[Total Variation Distance]\label{def:total variation distance}
	Suppose $P_X$ and $Q_X$ are two distributions in $\cP(\cX)$, the \emph{total variation distance}\footnote{Also known as variational distance. Notice that our definition in \eqref{eqn:def:total variation distance} does not have the normalization factor of $1/2$, and for this reason, given $P_X$, $Q_X\in \cP(\cX)$, we have $0 \le \|P_X - Q_X\|^{}_1 \le 2 $.  The main results of this work do not change if the normalization factor is included.} (or $\ell_1$-distance) between $P_X$ and $Q_X$ is 
	\begin{align}
		\left\|P_X - Q_X\right\|^{}_1  &= \sum_{x\in \cX} \left|P_X(x) - Q_X(x)\right| \label{eqn:def:total variation distance}  \\
		&= 2\sup_{\cA\subset \cX} |P_X(\cA) - Q_X(\cA)| \text{.}
 	\end{align}
\end{definition}

Letting $\cX $ and $\cY$ denote finite input and output alphabets, respectively, and using the standard notation $a^n = (a_1, \ldots, a_n)$ to denote an $n$-dimensional array, a \emph{stationary discrete memoryless channel} is defined through the sequence of random transformations as follows.
\begin{definition}[Discrete Memoryless Channel]
Suppose that $P_{Y|X}\colon \cX \to \cY$ is a random transformation between the finite alphabets $\cX$ and $\cY$. A \emph{stationary discrete memoryless channel} with input and output alphabets, $\cX$ and $\cY$, respectively, is a sequence of random transformations $\{P_{Y^n|X^n}\colon \cX^n \to \cY^n \}_{n=1}^{\infty}$ such that
\begin{align}
	P_{Y^n|X^n}(y^n|x^n) = \prod_{i=1}^n P_{Y_i|X_i}(y_i|x_i) \text{,}
\end{align}
where for each $i$, $P_{Y_i|X_i} = P_{Y|X}$. 
\end{definition}
If the input and the output of the stationary discrete memoryless channel are independent from each other, i.e., $P_{Y^n|X^n} = P_{Y^n}$, then we call this channel a \emph{degenerate channel} as it is impossible to communicate a meaningful message through it. 

Assume that $P_{X}\in \cP(\cX)$, unless otherwise stated, the product distribution $P_{X^n} \in \cP(\cX^n) $ denotes its independent identically distributed (i.i.d.) extension, i.e.,
\begin{align}
	P_{X^n}(x^n)  = \prod_{i=1}^{n} P_{X_i}(x_i) \text{,}
\end{align}
where $X_i$ are i.i.d. according to $P_{X}$. If we input an $n$-shot stationary discrete memoryless channel $P_{Y^n|X^n}$ with $X^n\sim P_{X^n}$, then at the output we get $Y^n\sim P_{Y^n}$ where
\begin{align}
	P_{Y^n}(y^n) = \sum_{x^n\in \cX^n} P_{X^n}(x^n) P_{Y^n|X^n}(y^n|x^n) \text{.}
\end{align}

\begin{rem}
  Throughout this paper, $P_{Y^n|X^n}$ denotes a stationary memoryless extension of the single-shot discrete channel $P_{Y|X}$. Similarly, $P_{X^n}$ and $P_{Y^n}$ always denote the product distributions of $P_X \in \cP(\cX) $ and $P_Y \in \cP(\cY) $, respectively, with former denoting the input distribution and the latter denoting the output distribution.
\end{rem} 

In what follows, we occasionally make use of the notation 
\begin{align}
	P_{X^n} \to P_{Y^n|X^n} \to P_{Y^n} \nonumber
\end{align}
to indicate that the $n$-shot channel $P_{Y^n|X^n} \colon \cX^n \to \cY^n $ is inputted with a random variable $X^n$ whose distribution is $P_{X^n} $, and the resulting random variable $Y^n$ at the output of the channel has distribution $P_{Y^n} = \sum_{x^n\in \cX^n} P_{X^n}(x^n)P_{Y^n|X^n}(\cdot|x^n)$.  Indeed, $P_{X^n}\to P_{Y^n|X^n} \to P_{Y^n}$ also defines a joint distribution $P_{X^nY^n} = P_{X^n}P_{Y^n|X^n}$, and furthermore, it allows us to define a key quantity in information theory, namely the \emph{information density}. 
\begin{definition}[Information Density]
	Given $P_X \to P_{Y|X} \to P_Y$, the \emph{information density} of $(x,y)\in \cX\times \cY $ is 
	\begin{align}
		\imath^{}_{X;Y}(x,y) &= \imath^{}_{P_{XY}\| P_X P_Y} (x,y) \\
		&= \log \frac{P_{Y|X}(y|x)}{P_Y(y)} \text{.}
	\end{align}
\end{definition}
Granted that the correlation between $X\sim P_X$ and $Y\sim P_Y$ is through $P_X \to P_{Y|X} \to P_Y$, the expected value of the random variable $\imath^{}_{X;Y}(X;Y)$ is a measure of dependency between $X$ and $Y$, which gives rise to the definition of mutual information. 
\begin{definition}[Mutual Information]
	Given $P_X \to P_{Y|X} \to P_Y$, the \emph{mutual information} of $(X, Y)\sim P_X P_{Y|X}$ is
	\begin{align}
		I(P_X, P_{Y|X}) &= \bbE[\imath^{}_{X;Y}(X;Y)]\\
		&= D(P_{XY}\| P_X P_Y)  \\
		&= D(P_{Y|X}\| P_Y | P_X) \text{.} 
	\end{align}
\end{definition}
The heart of the proof in channel coding theorem, random i.i.d. coding ensemble can be defined as follows. 
\begin{definition}[Random (i.i.d.) Codebook]\label{def:random codebook}
	Given $P_X \in \cP(\cX)$, let $P_{X^n}\in \cP(\cX^n)$ be its i.i.d. extension. A \emph{random (i.i.d.) codebook} $\scrC_M^n$ of size $M$ and block-length $n$ satisfies 
\begin{align}
	\scrC_M^n =\{X_1^n, \ldots , X_M^n \}, \label{eqn:def:random codebook}
\end{align} 
where $X_j^n$ are independently drawn from $P_{X^n}$ for each $j\in \{1, \ldots, M\}$. 
\end{definition}
Given a random codebook $\scrC_M^n$, the distribution at the output of the channel induced by $\scrC_M^n$ is defined next.
\begin{definition}[Induced Output Distribution]\label{def:induced output distribution}
Given an $n$-shot stationary discrete memoryless channel $P_{Y^n|X^n} \colon \cX^n \to \cY^n$, let $\scrC^n_M$ be the random codebook defined as in \eqref{eqn:def:random codebook}. Then, $\induced$ denotes the \emph{induced output distribution} when a uniformly chosen codeword from $\scrC_M^n$ is transmitted through $P_{Y^n|X^n}$. In other words, for any $y^n\in \cY^n$,
\begin{align}
	\induced(y^n) = \frac{1}{M} \sum_{j=1}^{M} P_{Y^n|X^n}(y^n|X_j^n) \text{,} \label{eqn:def:induced channel output distribution}
\end{align}
where $X_j^n \sim  P_{X^n}$ for each $j\in \{1, \ldots, M\}$. 
\end{definition} 
\begin{rem}
  Due to its dependence on the random codebook $\scrC_M^n$, $\induced$ is, in fact, a random variable. 
\end{rem}
Oftentimes, it is combinatorially convenient to treat the sequences with identical empirical distributions on an equal footing. Given a sequence $x^n\in \cX^n$, its empirical distribution is called an \emph{$n$-type} which we define as follows.
\begin{definition}[$n$-Type]
	For any positive integer $n$, a probability distribution $Q_{\bar X}\in \cP(\cX)$ is called an $n$-type if for any $x\in \cX$
	\begin{align}
		Q_{\bar X}(x) \in \left\{0, \frac{1}{n}, \frac{2}{n}, \ldots, 1\right\} \text{,}
	\end{align}
and the set of all $n$-types is denoted by $\cP_n(\cX) \subset \cP(\cX)$. 
\end{definition}
\begin{rem}\label{rem:m-type and km-type}
  For $m, k\in \bbN$, if $Q_{\bar X}$ is an $m$-type, it is also an $km$-type.
\end{rem}
Note that, see, e.g., \cite[Problem 2.1]{csiszar2011information}, the exact number of $n$-types in $\cX^n$ is $|\cP_n(\cX)| = \binom{n+|\cX|-1}{|\cX|-1}$ which grows polynomially with $n$. Since $n$-types play a significant role in our proofs, from this point onward, we reserve the overbar random variable notation for $n$-types. That is, for example, $\bar X \sim Q_{\bar X} $ denotes a random variable whose distribution  is an $n$-type $Q_{\bar X} \in \cP_n(\cX)$. Similarly, $(\bar X,\bar Y) \sim Q_{\bar X \bar Y} $ denotes a random variable whose distribution is a joint $n$-type $Q_{\bar X \bar Y} \in \cP_n(\cX \times \cY)$.

It is easy to see that given a sequence $x^n = (x_1, \ldots, x_n) \in \cX^n$ of block-length $n$, its empirical distribution defines an $n$-type $Q_{\bar X} \in \cP_n(\cX) $ as
	\begin{align}
		Q_{\bar X}(a) = \frac 1n \sum_{i=1}^n 1\{a = x_i\} \text{.}
	\end{align}
Conversely, given an $n$-type $Q_{\bar X} \in \cP_n(\cX)$, one can find a sequence $x^n\in \cX$ whose empirical distribution is $Q_{\bar X}$. This gives rise to the following definition. 
\begin{definition}[Type Class]
Given an $n$-type $Q_{\bar X}\in \cP_n(\cX)$, the subset $\cT^n_{Q_{\bar X}} \subset \cX^n$ is called the \emph{the type class of} $Q_{\bar X}$, and it denotes the set of all $x^n\in \cX^n$ whose empirical distribution is $Q_{\bar X}$. 
\end{definition}
To better understand the interplay of the joint sequences, the concept of \emph{conditional $n$-type} will be required. Let 
\begin{align}
  \cP(\cX |\cY ) = \{P_{X|Y} \colon \cY \to \cX  \}
\end{align}  
denote the set of all random transformations\footnote{Since both $\cX$ and $\cY$ are finite alphabets, under the convention that probability distributions are column vectors, $\cP(\cX|\cY)$ denotes the set of size $|\cX|\times |\cY|$ stochastic matrices.\label{ftnt:convention}} from $\cY$ to $\cX$. 
\begin{definition}[Conditional Type]\label{def:conditional type} 
Given an $n$-type $Q_{\bar Y}$, fix $y^n \in \cT^n_{Q_{\bar Y}}$. A random transformation\footnote{Under the convention of \Cref{ftnt:convention}, a stochastic matrix of dimension $|\cX| \times |\cY|$.} $Q_{\bar X|\bar Y} \colon \cY \to \cX \in \cP(\cX|\cY)$ is called the conditional type of $x^n \in \cX^n$ given $y^n$ if for any $(a,b)\in \cX \times \cY $
\begin{align}
 Q_{\bar X \bar Y}(a,b) = Q_{\bar X|\bar Y}(a|b) Q_{\bar Y}(b)   \text{,}
\end{align}
where $Q_{\bar X\bar Y}$ denotes the joint $n$-type of $(x^n, y^n)$. 
\end{definition}
\begin{rem}\label{rem:conditional types}
  Note that if $Q_{\bar Y}(b) = 0$ for some $b\in\cY$, then $Q_{\bar X\bar Y}(a,b)=0$ for any $a\in\cX$ and $Q_{\bar X|\bar Y} (\cdot|b)$ is not defined. If $Q_{\bar Y}(b) > 0$, then $Q_{\bar X|\bar Y} (\cdot|b)$ is a $t$-type where $t=n Q_{\bar Y}(b) $ is the number of times $b$ appears in $y^n$.
\end{rem}
Given a fixed $y^n\in \cT^n_{Q_{\bar Y}}$, the joint type $Q_{\bar X\bar Y}$ of the sequence $(x^n, y^n)$ can be determined by the conditional type $Q_{\bar X|\bar Y}$ of $x^n$ given $y^n$, in which case $Q_{\bar X\bar Y} = Q_{\bar X|\bar Y} Q_{\bar Y}  $. As this concept is utilized throughout this paper, a notation for the set of all conditional types is in order. 
\begin{definition}[Set of Conditional Types]\label{def:set of conditional types}
Given an $n$-type $Q_{\bar Y} \in \cP_n(\cY)$, $\cP_n(\cX|Q_{\bar Y})$ denotes the set of all conditional types given $y^n\in\cT^n_{Q_{\bar Y}}$. 
\end{definition}
\begin{rem}\label{rem:cP_n(cX|Q_{bar Y} depends only on the type Q_{bar Y}}
 As suggested by our careful choice of notation, it is easy to see that $\cP_n(\cX|Q_{\bar Y})$ depends on $y^n\in \cT^n_{Q_{\bar Y}} $ only through its type $Q_{\bar Y}$. Note that the subscript $n$ in the notation $\cP_n(\cX|Q_{\bar Y})$ is to denote that $Q_{\bar Y}$ is an $n$-type. Elements of $\cP_n(\cX|Q_{\bar Y})$ are conditional types, which are not necessarily $n$-types, see \Cref{rem:conditional types}.
\end{rem}
\begin{rem}\label{rem:joint types vs right coset of conditional types}
With \Cref{def:set of conditional types} at hand, the set of the joint $n$-types on $\cX\times\cY$ can be written as the disjoint union over $n$-types $\cP_n(\cY)$ of the right $Q_{\bar Y}$ coset\footnote{Abuse of terminology. $\cP_n(\cX|Q_{\bar Y})$ does not have a group structure.} of the set of conditional types $\cP_n(\cX|Q_{\bar Y})$. That is, borrowing the coset notation from algebra, 
\begin{align}
  \cP_n(\cX\times\cY) = \bigsqcup_{Q_{\bar Y}\in \cP_n(\cY)} \cP_n(\cX|Q_{\bar Y}) \times Q_{\bar Y} \text{,}
\end{align}
where the notation $\sqcup$ emphasizes that the unionization is disjoint.
\end{rem}
It is straightforward that given $y^n\in \cT^n_{Q_{\bar Y}}$, the empirical distribution of $x^n\in\cX^n $ in comparison with $y^n$ defines a conditional type $Q_{\bar X|\bar Y}\in \cP_n(\cX|Q_{\bar Y})$ as
\begin{align}
  Q_{\bar X|\bar Y}(a|b) = \frac{1}{n Q_{\bar Y}(b)} \sum_{i=1}^{n Q_{\bar Y}(b)} 1\left\{(x_i, y_i) = (a,b)\right\} \text{.}
\end{align}
Conversely, suppose we have a conditional type $Q_{\bar X|\bar Y} \in \cP_n(\cX|Q_{\bar Y})$ given $y^n \in \cT^n_{Q_{\bar Y}}$, we can construct a sequence $x^n\in\cX^n$ whose empirical distribution in comparison with $y^n$ is $Q_{\bar X|\bar Y}$. This gives rise to the definition of \emph{conditional type class}. 
\begin{definition}[Conditional Type Class]\label{def:conditional type class}
  Let $Q_{\bar X|\bar Y} \in \cP_n(\cX|Q_{\bar Y})$ be a conditional type given $y^n \in \cT^n_{Q_{\bar Y}}$, the subset $\cT^n_{Q_{\bar X|\bar Y}}(y^n)$ is called the \emph{conditional type class} of $Q_{\bar X|\bar Y}$ given $y^n$, and it denotes the set of all $x^n\in\cX^n $ whose empirical distribution in comparison with $y^n$ is $Q_{\bar X|\bar Y}$.
\end{definition}
\begin{rem}
The size of the conditional type class, namely $|\cT^n_{Q_{\bar X|\bar Y}}(y^n)|$, depends on $y^n$ only through its type. This is because shuffling the order of terms in which they appear in $y^n$, one can always shuffle $x^n$ in the same manner preserving the conditional type of $x^n$ given $y^n$. 
\end{rem}
Using the established familiarity with types, a \emph{random constant-composition codebook} can be defined as follows. 
\begin{definition}[Random Constant-Composition Codebook] \label{def:random cc codebook}
For a fixed integer $m$, suppose we are given an $m$-type $P_{\bar X}\in \cP_m(\cX)$. Let $n$ be a multiple of $m$ (i.e., $n\in m\bbN $) and define a \emph{constant-composition distribution} on $\cX^n$ based on $P_{\bar X}$ as
\begin{align}
  R_{\breve X^n}(x^n) = \frac{1}{|\cT^n_{P_{\bar X}}|} 1\left\{x^n \in \cT^n_{P_{\bar X}}\right\} \text{.} \label{eqn:def:cc distribution}
\end{align}
 Then, a \emph{random constant-composition codebook} of size $M$, and block-length $n$, that is based on $P_{\bar X}$ is defined as
 \begin{align}
   \scrD_M^n = \left\{\breve X^n_1, \ldots, \breve X^n_M\right\}\text{,} \label{eqn:def:cc random codebook}
 \end{align}
  where $\breve X^n_j $ are pairwise independent and identically distributed with $R_{\breve X^n}$ for each $j\in \{1, \ldots, M\}$. 
\end{definition}
\begin{rem}
  Each codeword in $\scrD_M^n$ has the same $m$-type $P_{\bar X}$ as they are taken uniformly at random from the type class $\cT^n_{P_{\bar X}}$, hence the name \emph{constant-composition}.
\end{rem}
\begin{rem}
  In the constant-composition case, $m$ is always fixed and $n$ is always a multiple of $m$. This ensures that the type class $\cT^n_{P_{\bar X}}$ is a well-defined non-empty set as the $m$-type $P_{\bar X}$ is also an $n$-type, see \Cref{rem:m-type and km-type}.
\end{rem}
\begin{rem}
  Throughout this paper, the distributions with breve accent `` $\breve{} $ " either denote constant-composition distributions or denote output distributions that are induced by constant-composition distributions. That is, unlike $P_{X^n} \in \cP(\cX^n)$, or $P_{Y^n} \in \cP(\cY^n) $;  $R_{\breve X^n} \in \cP(\cX^n)$, nor $R_{\breve Y^n} \in \cP(\cY^n)$, is \emph{not} a product distribution. 
\end{rem}
Given a random constant-composition codebook $\scrD_M^n$, the \emph{constant-composition induced output distribution} $\ccinduced$, in other words, the distribution induced by $\scrD_M^n$ at the channel output, is defined as follows.
\begin{definition}[Constant-Composition Induced Output Distribution]\label{def:cc induced output distribution}
  Given an $n$-shot stationary discrete memoryless channel $P_{Y^n|X^n}\colon \cX^n \to \cY^n$, let $\scrD_M^n $ be a random constant-composition codebook defined as in \eqref{eqn:def:cc random codebook}. Then, $\ccinduced $ denotes the \emph{constant-composition induced output distribution} when a uniformly chosen codeword from $\scrD_M^n$ is transmitted through $P_{Y^n|X^n}$. In other words, for any $y^n \in \cY^n$,
  \begin{align}
    \ccinduced(y^n) = \frac{1}{M} \sum_{j=1}^M P_{Y^n|X^n}(y^n| \breve{X}_j^n ) \text{,}
  \end{align}
   where for each $j\in \{1, \ldots, M \} $ the random variable $\breve X_j^n$ is distributed according to a constant-composition distribution $R^{}_{\breve X^n}$ that is based on an $m$-type $P_{\bar X} \in \cP_m(\cX) $, namely $\breve X_j^n \sim R_{\breve X^n}^{}$ as in \eqref{eqn:def:cc distribution}. 
\end{definition}
\begin{rem}
  Similar to $\induced$, due to its dependence on the random (constant-composition) codebook $\scrD_M^n$, $\ccinduced$ is, indeed, a random variable. 
\end{rem}
In proving the second main result of this paper, namely \Cref{thm:cc:main}, some additional notions, such as the set of joint types with fixed $\cX$- and $\cY$-marginals and the set of conditional types with fixed marginals, will be of use. The following definitions set the notation.
\begin{definition}[Set of Joint Types with Fixed Marginals]\label{def:set of joint types with fixed marginals}
Consider the set of joint $n$-types $\cP_n(\cX \times \cY)$. The subset $\cP_n(\cX\times \cY; Q_{\bar X}\times Q_{\bar Y}) \subset \cP_n(\cX \times \cY)$ denotes the set of all joint $n$-types whose $\cX$-marginal is fixed to be $Q_{\bar X}$ and $\cY$-marginal is fixed to be $Q_{\bar Y}$. That is
\begin{align}
  & \cP_n(\cX\times \cY; Q_{\bar X}\times Q_{\bar Y})= \Bigg\{Q_{\bar X\bar Y}\colon \nonumber \\  
  &\quad \sum_{b \in \cY} Q_{\bar X\bar Y}(\cdot, b) = Q_{\bar X} \text{, } \sum_{a \in \cX} Q_{\bar X\bar Y}(a, \cdot) = Q_{\bar Y} \Bigg\}  \text{.}
\end{align}
Similarly, the subset $\cP_n(\cX\times \cY; Q_{\bar X}\times \cdot)$ (respectively, $\cP_n(\cX\times \cY;\cdot\times Q_{\bar Y})$) denotes the set of joint $n$-types on $\cX\times\cY$ whose $\cX$-marginal is fixed to be $Q_{\bar X}$ (respectively, $\cY$-marginal is fixed to be $Q_{\bar Y}$). That is,
\begin{align}
  &\cP_n(\cX\times \cY; Q_{\bar X}\times \cdot) \nonumber \\ 
  &\qquad = \left\{Q_{\bar X\bar Y}\colon \sum_{b\in\cY} Q_{\bar X\bar Y}(\cdot, b) = Q_{\bar X} \right\} \text{,}  \\
  &\cP_n(\cX\times \cY; \cdot\times Q_{\bar Y}) \nonumber \\ 
  &\qquad = \left\{Q_{\bar X\bar Y}\colon \sum_{a\in\cX} Q_{\bar X\bar Y}(a, \cdot) = Q_{\bar Y} \right\} \text{.}
\end{align}  
\end{definition}
\begin{definition}[Set of Conditional Types with Fixed Marginals]\label{def:set of conditional types with fixed marginal}
Consider $\cP_n(\cX|Q_{\bar Y})$, the set of all conditional types given $y^n \in \cT^n_{Q_{\bar Y}}$. The subset $\cP_n(\cX|Q_{\bar Y};P_{\bar X}) \subset \cP_n(\cX|Q_{\bar Y})$ denotes the set of conditional types given $y^n\in \cT^n_{Q_{\bar Y}}$ with a fixed $\cX$-marginal $P_{\bar X}$. That is, 
\begin{align}
  &\cP_n(\cX|Q_{\bar Y};P_{\bar X}) \nonumber \\ 
  &\qquad = \{Q_{\bar X|\bar Y} \colon Q_{\bar Y} \to Q_{\bar X|\bar Y} \to P_{\bar X} \}  \\
  &\qquad = \left\{Q_{\bar X|\bar Y} \colon \sum_{b\in \cY}Q_{\bar X|\bar Y}(\cdot| b) Q_{\bar Y}(b) = P_{\bar X}\right\}  \text{.}
\end{align} 
\end{definition}
\begin{rem}\label{rem:cc:cP_n(cX|Q_{bar Y};P_{bar X}) depends on y^n through its type}
$ \cP_n(\cX|Q_{\bar Y};P_{\bar X}) $ depends on $y^n \in \cT^n_{Q_{\bar Y}}$ only through its type $Q_{\bar Y}$. The subscript $n$ in its notation is to denote that $P_{\bar X} $ and $Q_{\bar Y} $ are $n$-types. The elements of $ \cP_n(\cX|Q_{\bar Y};P_{\bar X}) $ are conditional types, which are not necessarily $n$-types, see \Cref{rem:conditional types}.
\end{rem}
\begin{rem}\label{rem:cc:coset notation for set of conditional types}
Using the coset\footnote{Abuse of terminology. $\cP_n(\cX|Q_{\bar Y};P_{\bar X}) $ does not have a group structure.} notation and the definitions above, the following identities are immediate: 
\begin{align}
  & \cP_n(\cX\times \cY; P_{\bar X}\times \cdot) \nonumber \\ 
  &\qquad     = \bigsqcup_{Q_{\bar Y}\in\cP_n(\cY)} \cP_n(\cX|Q_{\bar Y};P_{\bar X}) \times Q_{\bar Y} \label{eqn:joint with fixed vs conditional-1}  \\
  &\qquad =   P_{\bar X} \times \cP_n(\cY|P_{\bar X}) \text{,} \label{eqn:joint with fixed vs conditional-2} 
\end{align}
where the notation $\sqcup$ emphasizes that the unionization is disjoint.
\end{rem}
\section{Exact Soft-Covering Exponent}\label{sec:main}

We begin by citing \cite[Theorem 31]{lcv2017egamma} which establishes that the total variation distance between the induced output distribution $\induced$ and the desired output distribution $P_{Y^n}$ has a concentration property. As the block-length $n$ increases, the total variation distance between these two distributions (a random quantity, due to the randomness of the codebook) concentrates tightly to its exponentially vanishing expected value with double-exponential certainty.\footnote{The result of \Cref{lem:TV_concentration-McDiarmid-deterministic} holds regardless of the value of the rate $R$ whereas $\bbE[\|\induced - P_{Y^n}\|^{}_1]$ vanishes exponentially fast only if $R > I(P_X, P_{Y|X})$. In the right side of \eqref{eqn:thm:TV mcdiarmd}, double-exponential decay is guaranteed when $M = \exp(nR) $.}

\begin{lemma}\label{lem:TV_concentration-McDiarmid-deterministic}
Suppose $P_{X^n} \to P_{Y^n|X^n} \to P_{Y^n} $ and denote by $\induced $ the induced output distribution when a uniformly chosen codeword from the random (i.i.d.) codebook $\scrC_M^n$ is transmitted through the channel $P_{Y^n|X^n}$, see \Cref{def:random codebook,def:induced output distribution}. Then, for any $t>0$,
\begin{align}
 	&\bbP \Big[ \Big| \Big\|\induced - P_{Y^n}\Big\|^{}_1  - \bbE\Big[\Big\|\induced - P_{Y^n}\Big\|^{}_1\Big]\Big|\ge t \Big]  \nonumber \\  
 	&\qquad \le 2 \exp_\rme\left(-\frac{Mt^2}2\right) \text{.} \label{eqn:thm:TV mcdiarmd}
\end{align}
\end{lemma}

Predictably, replacing the random (i.i.d.) codebook $\scrC_M^n$ with a random constant-composition codebook $\scrD_M^n$ in \Cref{lem:TV_concentration-McDiarmid-deterministic} and looking at the total variation distance between the constant-composition induced output distribution $\ccinduced$ and the desired output distribution $R_{\breve Y^n}$, we see that the same concentration property holds:

\begin{lemma}\label{lem:cc:concentration of TV distance btw ccinduced and ccoutput}
  Suppose $R_{\breve X^n} \to P_{Y^n|X^n} \to R_{\breve Y^n} $ and denote by $\ccinduced$ the induced output distribution when a uniformly chosen codeword from the random (constant-composition) codebook $\scrD_M^n $ is transmitted through the channel $P_{Y^n|X^n}$, see \Cref{def:random cc codebook,def:cc induced output distribution}. Then, for any $t>0$, 
  \begin{align}
    &\bbP \Big[ \Big| \Big\| \ccinduced - R_{\breve Y^n} \Big\|^{}_1  - \bbE\Big[ \Big\| \ccinduced - R_{\breve Y^n} \Big\|^{}_1\Big]\Big| \ge t \Big] \nonumber \\ 
    &\qquad \le 2\exp_\rme \left( -\frac{Mt^2}{2} \right)  \text{.}
  \end{align}
\end{lemma}

The main results of this paper, stated in \Cref{thm:main,thm:cc:main}, give the exact asymptotic exponential decay rate of the expected total variation distance between the induced distribution $\induced$ (respectively, $\ccinduced$) and the desired output distribution $P_{Y^n}$ (respectively, $R_{\breve Y^n}$). 
\begin{theorem}[Exact Soft-Covering Exponent (i.i.d.)]\label{thm:main}
Suppose $P_{X^n}\to P_{Y^n|X^n} \to P_{Y^n}$, where the $n$-shot stationary memoryless channel $P_{Y^n|X^n}$ is non-degenerate, i.e., $P_{Y^n|X^n} \neq P_{Y^n}$. For any $R> I(P_X, P_{Y|X})$, let $M=\lceil \exp(nR) \rceil $, and denote by $\induced$ the induced output distribution when a uniformly chosen codeword from the random codebook $\scrC_M^n$ is transmitted through the channel, see \Cref{def:random codebook,def:induced output distribution}. Then,
	\begin{align}
		& \lim_{n\to \infty} -\frac{1}{n} \log \bbE \left[\left\| \induced - P_{Y^n} \right\|^{}_1\right] \nonumber \\   
		&\qquad = \min_{Q_{XY}\in \cP(\cX\times\cY)} \bigg\{ D(Q_{XY} \| P_{XY})  \nonumber \\ 
		&\qquad \qquad \quad \ \,  + \frac12 \left[R-D(Q_{XY}\|P_XQ_Y) \right]_{+}\bigg\} \label{eqn:thm:exact soft-cov exponent}\\
		&\qquad = \max_{\lambda\in [1,2] } \left\{ \frac{\lambda-1}{\lambda} \left(R - I^\sfs_\lambda (P_X, P_{Y|X})\right)\right\} \text{,} \label{eqn:thm:exact soft-cov exponent_ALTERNATIVE}
	\end{align}
	where in \eqref{eqn:thm:exact soft-cov exponent} $[f]_+ = \max \{0,f \} $; 
	%and the optimization is carried over all distributions $Q_{XY}$ that are absolutely continuous with respect to $P_{XY}=P_X P_{Y|X}$, i.e., $Q_{XY} \ll P_{XY}$;  
	and in \eqref{eqn:thm:exact soft-cov exponent_ALTERNATIVE} $I^\sfs_\lambda(P_X, P_{Y|X})$ denotes Sibson's proposal of the $\alpha$-mutual information\footnote{See \Cref{rem:alpha-mutual information}.} of order $\lambda$.
\end{theorem}

\begin{theorem}[Exact Soft-Covering Exponent (constant-composition)]\label{thm:cc:main}
Let $m$ be a fixed integer and $P_{\bar X} \in \cP_m(\cX)$ be a fixed $m$-type. For  $n\in m\bbN $, suppose that $R_{\breve X^n}$ is a constant-composition distribution based on $P_{\bar X}$ defined as in \eqref{eqn:def:cc distribution}, and let $R_{\breve X^n}\to P^{}_{Y^n|X^n} \to R_{\breve Y^n} $, where the $n$-shot stationary discrete memoryless channel $P_{Y^n|X^n}$ is non-degenerate, i.e., $P^{}_{Y^n|X^n} \neq R_{\breve Y^n}$. For any $R>I(P_{\bar X}, P_{Y|X}) $, let $ M = \lceil \exp(nR) \rceil $, and denote by $\ccinduced $ the constant-composition induced output distribution when a uniformly chosen codeword from the random constant-composition codebook $\scrD^n_M$ is transmitted through the channel, see \Cref{def:random cc codebook,def:cc induced output distribution}. Then,
  \begin{align}
    & \lim_{n\to \infty} -\frac{1}{n}\log \bbE\left[\left\| \ccinduced - R_{\breve Y^n} \right\|^{}_1\right] \nonumber \\  
    &\qquad =\min_{Q_{Y|X}\in \cP(\cY|\cX)}  \bigg\{ D(P_{\bar X}Q_{Y|X}\|P_{\bar XY}) \nonumber \\ 
    &\qquad \quad \ \ \ \, + \frac12[R-D(P_{\bar X}Q_{Y|X}\|P_{\bar X}Q_Y)]_+  \bigg\}  \label{eqn:thm:cc:exponent}  \\ 
    & \qquad  = \max_{\lambda\in [1,2]} \left\{ \frac{\lambda-1}{\lambda} \left(R - I^\sfc_\lambda (P_{\bar X}, P_{Y|X})\right)\right\} \text{,}   \label{eqn:thm:cc:exponent_ALTERNATIVE} 
  \end{align}
  where in \eqref{eqn:thm:cc:exponent} $[f]_+ = \max\{0,f\} $, and $P_{\bar X}  \to Q_{Y|X} \to Q_Y  $; and in \eqref{eqn:thm:cc:exponent_ALTERNATIVE} $I^\sfc_\lambda (P_{\bar X}, P_{Y|X})$ denotes Csisz\'ar's proposal of the $\alpha$-mutual information\footnote{See Remark~\ref{rem:cc:csiszar-alpha-mutual information}.} of order $\lambda$. 
 % and the optimization is carried over all random transformations $Q_{Y|X}$ such that $P_{\bar X}Q_{Y|X} \ll P_{\bar X}P_{Y|X} = P_{\bar XY} $; and in \eqref{eqn:thm:cc:exponent_ALTERNATIVE} $(\bar X, Y) \sim P_{\bar X}P_{Y|X} = P_{\bar XY}$. 
\end{theorem}

Some remarks are in order. 

\begin{rem}
	To make it easier to refer, assuming $R>I(P_X,P_{Y|X})>0$, define 
	\begin{align}
		&\alpha(R,P_X,P_{Y|X}) \nonumber \\ 
		&\qquad = \min_{Q_{XY}\in \cP(\cX\times\cY)} \bigg\{  D(Q_{XY} \| P_{XY} ) \nonumber \\ 
		&\qquad \qquad \quad \ \,  + \frac12 \left[R-D(Q_{XY}\|P_XQ_Y) \right]_{+} \bigg\} \label{eqn:def:alpha(R,P_X,P_{Y|X})}  \\
		&\qquad = \max_{\lambda\in [1,2] } \left\{ \frac{\lambda-1}{\lambda} \left(R - I^\sfs_\lambda (P_X, P_{Y|X}) \right)  \right\}  \text{,} \label{eqn:def:alpha(R,P_X,P_{Y|X})_the dual}
	\end{align}
where the minimization in \eqref{eqn:def:alpha(R,P_X,P_{Y|X})} is over all joint distributions on $\cX\times \cY $. 

Similarly, assuming $R>I(P_{\bar X},P_{Y|X})>0$, define 
  \begin{align}
    &\aleph(R, P_{\bar X}, P_{Y|X}) \nonumber \\ 
    &\qquad = \min_{Q_{Y|X} \in \cP(\cY|\cX)} \bigg\{D(P_{\bar X}Q_{Y|X}\|P_{\bar XY}) \nonumber \\  
    &\qquad \quad \ \ \ \, + \frac12[R-D(P_{\bar X}Q_{Y|X}\|P_{\bar X}Q_Y)]_+ \bigg\}  \label{eqn:cc:def:aleph} \\
    &\qquad = \max_{\lambda\in [1,2]} \left\{ \frac{\lambda-1}{\lambda} \left(R - I^\sfc_\lambda (P_{\bar X}, P_{Y|X})\right)\right\} \text{,} \label{eqn:cc:def:aleph_the dual}
  \end{align} 
  where the minimization in \eqref{eqn:cc:def:aleph} is over all random transformations from $\cX$ to $\cY$. 
\end{rem}
\begin{rem}
	Perhaps surprisingly, the proof of \Cref{lem:TV_concentration-McDiarmid-deterministic}, which can be found in \cite[Theorem 31]{lcv2017egamma}, easily follows from McDiarmid's inequality \cite[Theorem 2.2.3]{raginsky2013concentration}. As \Cref{lem:TV_concentration-McDiarmid-deterministic} is an integral part of the spirit of this paper, Appendix~\ref{apdx:proof of TV concentration McDiarmid} repeats its simple proof. Also contained in Appendix~\ref{apdx:proof of TV concentration McDiarmid}, the proof of \Cref{lem:cc:concentration of TV distance btw ccinduced and ccoutput} follows the footsteps of that of \Cref{lem:TV_concentration-McDiarmid-deterministic}.
\end{rem}
\begin{rem}\label{rem:foreshadowing the upper bound proof}
	By further assuming that the codebooks $\scrC_M^n$ and $\scrD_M^n$ contain a random number of codewords $M$, thanks to the total probability law, it is possible to get corollaries to the results of \Cref{lem:TV_concentration-McDiarmid-deterministic,lem:cc:concentration of TV distance btw ccinduced and ccoutput}. Indeed, an example, in which we assume that $M$ is Poisson distributed, is useful in the proof of the upper bound in \Cref{thm:main}, cf. \Cref{lem:TV-concentration_random-poisson} in \Cref{apdx:main}. 
\end{rem}
\begin{rem}
In order to provide a better presentation, the proof of \Cref{thm:main} is divided into three parts, which can be found in Sections~\ref{sec:lower bound},~\ref{sec:upper bound} and~\ref{sec:alternative representation of alpha}. In proving the lower bound direction in \eqref{eqn:thm:exact soft-cov exponent}, see\footnote{Also see \cite{hayashi2013tight}, which studies the privacy amplification problem. As an application to the wiretap channel, \cite{hayashi2013tight} argues the lower bound in \eqref{eqn:thm:exact soft-cov exponent_ALTERNATIVE} without showing the equivalence in \eqref{eqn:thm:exact soft-cov exponent}. As a comparison to the method suggested in \cite{hayashi2013tight}, note that our proof in \Cref{sec:lower bound} is far simpler to follow.} \Cref{sec:lower bound}, the key steps are the use of the type method and an upper bound on the absolute mean deviation of a binomial distribution in terms of its mean and standard deviation. To prove the upper bound direction, on the other hand, the biggest problem turns out to be dealing with the weakly dependent binomial random variables, see \Cref{sec:upper bound}. To solve this weak dependence puzzle, first, the codebook size $M$ is treated as if it were a Poisson distributed random variable with mean $\mu_n = \exp(nR) $. This surplus assumption on the codebook size grants the desired independence property and provides the gateway to prove the pseudo-upper bound in the case when $M$ is Poisson distributed. Then, to prove the upper bound to the original problem where $M$ is deterministically equals to $\lceil \exp(nR) \rceil $, the extra Poisson assumption is removed by conditioning on $M=\lceil \exp(nR) \rceil $ and the result provided by \Cref{lem:TV_concentration-McDiarmid-deterministic} is enjoyed. As for the proof of the dual representation of the exact soft-covering exponent in \eqref{eqn:thm:exact soft-cov exponent_ALTERNATIVE}, see \Cref{sec:alternative representation of alpha}, the main tools are provided by \Cref{lem:rel ent and func minimization} and several corollaries that follow, all of which are contained in \Cref{apdx:alternative representation and comparisons}. 
\end{rem}
\begin{rem}  %% NEW REMARK WITHOUT THE PROOF OF THM-2
While presenting the proof of \Cref{thm:main}, much effort has been made so that it is possible to capture that of \Cref{thm:cc:main} from the existing proof in \Cref{sec:lower bound,sec:upper bound}. Still, there are certain key differences between aforementioned two proofs, which is why neither theorem is a corollary of the other. One example to these key differences is that, in the case of \Cref{thm:cc:main}, in applying the type method, one needs to keep in mind that $\cX$-marginal of the joint types is fixed to be $P_{\bar X}$, whereas this is not the case in the proof of \Cref{thm:main}. Another key difference is that, in the case of \Cref{thm:cc:main}, the codewords of the random constant-composition codebook $\scrD_M^n$ are distributed according to the non-product distribution $R_{\breve X^n}$, while the codewords of the random (i.i.d.) codebook $\scrC_M^n$ are distributed according to the product distribution $P_{X^n}$. Luckily, using a minimalist approach, it is possible to emphasize the similarities in the techniques used. To do so, while proving \Cref{thm:main} in \Cref{sec:lower bound,sec:upper bound}, several remarks have been made to convince the reader in regard to \Cref{thm:cc:main} without having them read through its entire proof. Since the  presented material is more than enough to recover the proof of \Cref{thm:cc:main}, its full proof is omitted. However, note that, the proof for the equivalence of the primal and dual forms of the exact constant-composition soft-covering exponent, namely \eqref{eqn:thm:cc:exponent_ALTERNATIVE}, can be found in \Cref{sec:alternative representation of aleph}. 
\end{rem}
\begin{rem}
	The result of \Cref{thm:main} can alternatively be interpreted as the \emph{exact random coding exponent for resolvability}. Note, however, that we are \emph{not} claiming to have found ``the" exact resolvability exponent. Finding the exact resolvability exponent is a harder problem as it requires the search over all sequences of codes. Here, we restrict ourselves to random codebooks, as are typically used in achievability proofs (e.g. wiretap channels) where soft covering may be only one of several objectives.  This choice of focus has a side benefit of finding the exact exponent. 
\end{rem}
\begin{rem}
	As is evident from the upper bound in \eqref{eqn:main upper bound} in \Cref{sec:upper bound}, $\alpha(R,P_X,P_{Y|X})$ is the best possible soft-covering exponent in the random (i.i.d.) codebook case.\footnote{A similar statement is true for $\aleph(R,P_{\bar X},P_{Y|X})$ as well.} \Cref{sec:comparison IID Random Codebook,sec:cc:comparison Constant Composition} confirm that $\alpha(R,P_X,P_{Y|X})$ and $\aleph(R,P_{\bar X},P_{Y|X})$ provide an upper bound to the previously known lower bounds\footnote{These lower bounds can be found in (or deduced from) \cite[Theorem 6]{hayashi2006general}, \cite[Lemma VII.9]{cuff2013distributed}, \cite[Theorem 4]{parizi2017exact}, \cite[Theorem 10]{hayashi2011universally}, and \cite[Eq. (177)]{hayashi-matsumoto-2016}.} on the soft-covering exponent in their respective cases.
\end{rem}
%%
%\begin{rem}
%Parizi \emph{et al.} \cite[Theorem 4]{parizi2017exact} provide the soft-covering exponent (both i.i.d. and constant-composition cases) when the relative entropy rather than total variation is used as the measure of distinctness. While Pinsker's and Jensen's inequalities immediately imply that the half of the exponents in \cite{parizi2017exact} are lower bounds on the soft-covering exponents in \Cref{thm:main,thm:cc:main}, as shown in \Cref{sec:comparison parizi et al.(i),sec:cc:comparison Constant Composition}, $\alpha(R,P_X,P_{Y|X})$ and $\aleph(R,P_X,P_{Y|X})$ are greater than the half of their respective counterparts in \cite[Theorem 4]{parizi2017exact}. 
%\end{rem}
%%
\begin{rem} 
From the proofs provided, it is possible to deduce the following finite block-length results, see \Cref{thm:finite block-length exact scl,thm:cc:finite block-length cons. comp. exact scl} in \Cref{apdx:finite block-length results}:
\begin{align}
&\alpha_n(R, P_X, P_{Y|X}) - \kappa_n \nonumber \\
&\quad \le -\frac{1}{n} \log \bbE \left[\left\| \induced - P_{Y^n} \right\|^{}_1\right] \label{eqn:fin-blklnth-alpha_low}  \\
&\quad \le \alpha_n(R, P_X, P_{Y|X}) + \upsilon_n  \text{,} \label{eqn:fin-blklnth-alpha_upp} 
\end{align}
and
\begin{align}
&\aleph_n(R, P_{\bar X}, P_{Y|X}) - \breve \eta_n \nonumber \\
&\quad \le -\frac{1}{n}\log \bbE\left[\left\| \ccinduced - R_{\breve Y^n} \right\|^{}_1\right] \label{eqn:fin-blklnth-aleph_low} \\
&\quad \le \aleph_n(R, P_{\bar X}, P_{Y|X}) + \breve \upsilon_n \text{,} \label{eqn:fin-blklnth-aleph_upp}
\end{align}
where 
  \begin{align}
&\alpha_n(R, P_X, P_{Y|X}) \nonumber \\  
&\quad = \min_{Q_{\bar X\bar Y}\in \cP_n(\cX\times\cY)} \bigg\{ D(Q_{\bar X\bar Y}\|P_{XY}) \nonumber \\
&\qquad \qquad \qquad + \frac12 \left[R-D(Q_{\bar X\bar Y}\|P_XQ_{\bar Y})\right]_{+}\bigg\} \text{,} \label{eqn:def:alpha_n for finite block-length} 
\end{align}
and
\begin{align}
&\aleph_n(R, P_{\bar X}, P_{Y|X}) \nonumber \\
&\quad = \min_{Q_{\bar Y|\bar X}\in  \cP_n(\cY|P_{\bar X})}  \bigg\{  D(P_{\bar X}Q_{\bar Y|\bar X}\|P_{\bar XY}) \nonumber \\
&\qquad \qquad \ \  + \frac12[R-D(P_{\bar X}Q_{\bar Y|\bar X}\|P_{\bar X}Q_{\bar Y})]_+ \bigg\}  \text{.}
\end{align}
Among the vanishing constants $\kappa_n$, $\upsilon_n$, $\breve \eta_n$, $\breve \upsilon_n $, the ones in the lower bounds in \eqref{eqn:fin-blklnth-alpha_low} and \eqref{eqn:fin-blklnth-aleph_low}, i.e., $\kappa_n $ and $\breve \eta_n $, depend only on the block-length $n$ and the alphabet sizes $|{\cal X}|$ and $|{\cal Y}|$, while the ones in the upper bounds in \eqref{eqn:fin-blklnth-alpha_upp} and \eqref{eqn:fin-blklnth-aleph_upp}, i.e., $\upsilon_n $ and $\breve \upsilon_n$, additionally depend mildly\footnote{Also see \Cref{rem:upper bound constants in finite blocklenth} in Appendix~\ref{apdx:finite block-length results}.} on $P_X$, $P_{\bar X}$ and $P_{X|Y}$.  The definitions of these vanishing constants, along with the proofs of the pairs \eqref{eqn:fin-blklnth-alpha_low}--\eqref{eqn:fin-blklnth-alpha_upp} and \eqref{eqn:fin-blklnth-aleph_low}--\eqref{eqn:fin-blklnth-aleph_upp}, are contained in \Cref{apdx:finite block-length results}.
\end{rem}
\begin{rem}
In the case when $R\le I(P_X, P_{Y|X})$, $Q_{XY} = P_XP_{Y|X} $ becomes the optimizer in \eqref{eqn:thm:exact soft-cov exponent}, yielding the correct exponent,
\begin{align}
  \alpha(R, P_X, P_{Y|X} ) = 0 \text{,}
\end{align}
for the low-rate codes. 

Similarly, in the random constant-composition codebook setting, when $R\le I(P_{\bar X}, P_{Y|X})$, $Q_{Y|X} = P_{Y|X} $ becomes the optimizer in \eqref{eqn:thm:cc:exponent}, which yields the correct exponent,
\begin{align}
  \aleph(R, P_{\bar X}, P_{Y|X} ) = 0 \text{,}
\end{align}
  for the low-rate codes in this respective setting. 
\end{rem}

\begin{rem}\label{rem:degenerate channel case discontinuity}
In the degenerate channel case, i.e., when channel input and output are independent from each other, we have $\induced = P_{Y^n}$ (in the constant-composition codes setting, $\ccinduced = R_{\breve Y^n} $) and
\begin{align}
	\bbE \left[\left\| \induced - P_{Y^n} \right\|^{}_1\right]  &= 0 \text{,} \\
	\bbE \left[\left\| \ccinduced - R_{\breve Y^n} \right\|^{}_1\right]  &= 0 \text{.}
\end{align} 
In an allegorical spirit, one can say that the exact soft-covering exponents are $\infty$ in this case. Although, it should be noted that \eqref{eqn:thm:exact soft-cov exponent}, \eqref{eqn:thm:exact soft-cov exponent_ALTERNATIVE}, \eqref{eqn:thm:cc:exponent}, and \eqref{eqn:thm:cc:exponent_ALTERNATIVE} do not capture this conclusion. A similar discontinuity occurs in the case when the distinctness measure is relative entropy instead of total variation distance, see \cite[Theorem 4]{parizi2017exact}. In our treatment, the reason for these discontinuities can be observed from \eqref{from:triangle dude} in the upper bound proof.
\end{rem}

%\begin{rem} \label{rem:seperating the minimizations}
%	The optimization in \eqref{eqn:def:alpha(R,P_X,P_{Y|X})} can also be written as
%	\begin{align}
%	&\alpha(R, P_X, P_{Y|X}) = \nonumber \\
%		&\quad \min_{Q_Y} \left\{D(Q_Y\|P_Y) + \min_{Q_{X|Y}}\left\{D(Q_{X|Y}\|P_{X|Y}|Q_Y) + \frac12 \left[R-D(Q_{X|Y}\|P_X|Q_Y)\right]_+\right\}\right\} \text{.}
%	\end{align}
%	As shown in \Cref{apdx:Exponential Vanishing Argument}, without loss of optimality, the inner minimization can be constrained to be over the random transformations $Q_{X|Y}$ satisfying $D(Q_{X|Y}\|P_X|Q_Y) \ge D(P_{X|Y}\|P_X|Q_Y)$.
%\end{rem}

\begin{rem}\label{rem:alpha-mutual information}
  In the optimization in the right side of \eqref{eqn:def:alpha(R,P_X,P_{Y|X})_the dual}, letting $P_X \to P_{Y|X}\to P_Y$, $(X,Y)\sim P_X P_{Y|X} $,  and $(X, \widetilde Y)  \sim P_X P_Y $,
  \begin{align}
    & I^\sfs_\lambda(P_X, P_{Y|X}) \nonumber \\ 
    &= \frac{\lambda}{\lambda-1}\log \bbE\left[\bbE^{\frac1\lambda}\left[\exp((\lambda-1) \,\imath^{}_{X;Y}(X;Y))\big|Y\right]\right] \label{eqn:def:alpha mut info-semih} \\
    &= \frac{\lambda}{\lambda-1}\log \bbE\left[\bbE^{\frac1\lambda}\left[\exp(\lambda\, \imath^{}_{X;Y}(X; \widetilde Y) ) \big| \widetilde Y\right]\right] \label{eqn:def:alph mut info-verdu} \\
    &= \frac{\lambda}{\lambda-1}\log \sum_{y\in \cY} \left(\sum_{x\in \cX} P^{}_X(x) P^\lambda_{Y|X}(y|x)\right)^{\frac1\lambda}
  \end{align}
  is the $\alpha$-mutual information of order $\lambda $ as defined by Sibson \cite{sibson1969information}. Its more general definition, basic properties, relation to the other variations of $\alpha$-mutual information, and connection to Gallager error exponent function \cite[Eq. (5.6.14)]{gallager1968information} are explored in \cite{verdu2015}. 
\end{rem}

\begin{rem}\label{rem:cc:csiszar-alpha-mutual information}
Denoting the R\'enyi divergence (see, e.g., \cite{vanErven2014}) of order $\lambda$ by $D_{\lambda}(P\|Q)$, in the optimization in the right side of \eqref{eqn:cc:def:aleph_the dual}, letting $(X, Y) \sim P_X P_{Y|X}=P_{XY}$,
  \begin{align}
    & I^\sfc_{1+\lambda}(P_X, P_{Y|X}) \nonumber \\ 
    &= \min_{S_Y\in \cP(\cY) } \bbE \left[D_{1+\lambda}(P_{Y|X}(\cdot | X )\| S_Y)\right] \label{eqn:def:cc:csiszar alpha mut info} \\
    &= \min_{S_Y\in \cP(\cY) }\bbE\Big[\log\bbE^{\frac1{\lambda}}[\exp\Big(\lambda\,\imath^{}_{P_{XY}\| P_X S_Y }(X,Y) \Big) \Big| X ]\Big] \label{eqn:def:cc:csz alp mut alternative}
  \end{align}
 is the $\alpha$-mutual information of order $1+\lambda$ as defined by Csisz\'ar \cite{csiszar1995GeneralizedCutoff}. Its basic properties and relation to Sibson's proposal of $\alpha$-mutual information are explored in \cite{verdu2015}.
\end{rem}
\begin{rem}\label{rem:aleph > alpha}
 Given an arbitrary non-degenerate channel $P_{Y|X}\colon \cX\to \cY$, and an $m$-type $P_{\bar X}\in \cP_m(\cX)$ as the input distribution, proving $\alpha \le \aleph $ is simple:
  \begin{align}
    &\alpha(R,P_{\bar X},P_{Y|X}) \nonumber \\
    &\quad = \min_{Q_{Y|X}} \bigg\{\min_{Q_X} \bigg\{ D(Q_XQ_{Y|X} \| P_{\bar XY}) \nonumber \\  
    &\qquad \quad + \frac12 \left[R-D(Q_XQ_{Y|X}\|P_{\bar X}Q_Y) \right]_{+}\bigg\} \bigg\}  \\
    &\quad \le \min_{Q_{Y|X}} \bigg\{ D(P_{\bar X}Q_{Y|X}\|P_{\bar XY}) \nonumber \\ 
    &\qquad \quad + \frac12[R-D(P_{\bar X}Q_{Y|X}\|P_{\bar X}Q_Y)]_+ \bigg\} \label{frm:cc:subopt Q_X} \\
    &\quad = \aleph(R, P_{\bar X}, P_{Y|X}) \label{eqn:rem:alpha<aleph} \text{,} 
  \end{align}
  where \eqref{frm:cc:subopt Q_X} follows from the suboptimal choice of $Q_X = P_{\bar X}$. Though, as the next remark illustrates, this is not the sole order relation between $\alpha$ and $\aleph$. 
\end{rem}
\input{Fig_Alpha-vs-Aleph.tex} %FIGURE ALPHA VS ALEPH
\begin{rem}
Suppose $P_{\bar X}\to P_{Y|X}\to P_Y $, and let $(\bar X, Y)\sim P_{\bar X}P_{Y|X} $, 
%and  $(\widetilde X , Y) \sim P_{\bar X} P_Y $
\begin{align}
    &\aleph(R,P_{\bar X},P_{Y|X}) \nonumber  \\
    &= \max_{\lambda\in [1,2]}\max_{S_Y} \bigg\{ \frac{\lambda-1}{\lambda} R \label{cc:frm:def_ofcsza_lph_mut}  \\ 
    & -  \bbE\Big[\log\bbE^{\frac1{\lambda}} \Big[ \exp\Big( (\lambda -1)\, \imath^{}_{P_{\bar XY}\|P_{\bar X}S_Y}(\bar X, Y)\Big)\Big|\bar X \Big] \Big] \bigg\}  \nonumber \\
    &\ge \max_{\lambda\in [1,2]}\max_{S_Y} \bigg\{ \frac{\lambda-1}{\lambda} R \label{cc:eq:rem:jensen} \\ 
    & -   \log\bbE\Big[\bbE^{\frac1{\lambda}} \Big[\exp\Big( (\lambda -1)\, \imath^{}_{P_{\bar XY}\|P_{\bar X}S_Y}(\bar X, Y)\Big)\Big|\bar X \Big]\Big]\bigg\} \nonumber \\
    &\ge \max_{\lambda\in [1,2]} \bigg\{ \frac{\lambda-1}{\lambda} R \label{cc:eq:rem:subopt P_Y} \\ 
    & -  \log\bbE\Big[\bbE^{\frac1{\lambda}} \Big[\exp\Big( (\lambda -1)\, \imath^{}_{P_{\bar XY}\|P_{\bar X} P_Y}(\bar X, Y)\Big)\Big|\bar X \Big]\Big]\bigg\} \nonumber \\
%    &=  \max_{\lambda\in [1,2]} \left\{\frac{\lambda-1}{\lambda} R -  \log\bbE\left[\bbE^{\frac1{\lambda}} \left[\exp\left( \lambda\, \imath^{}_{P_{\bar XY}\|P_{\bar X} P_Y}(\widetilde X, Y)\right)\middle|\widetilde X \right]\right]\right\} \label{cc:eq:change of meas} \\
    &=  \max_{\lambda\in [1,2]} \left\{\frac{\lambda-1}{\lambda}  \left( R - I^\sfs_\lambda(P_Y, P_{\bar X|Y } )  \right) \right\}  \label{cc:eq:frm:def:alph mut}\\
    &= \alpha(R, P_Y, P_{\bar X|Y}) \label{eqn:rem:reverse alpha<aleph} \text{,}
\end{align}
where \eqref{cc:frm:def_ofcsza_lph_mut} follows from the definition of $I^\sfc_{\lambda}(P_X, P_{Y|X}) $ in \eqref{eqn:def:cc:csz alp mut alternative}; \eqref{cc:eq:rem:jensen} follows from Jensen's inequality; \eqref{cc:eq:rem:subopt P_Y} follows from the suboptimal choice of $S_Y = P_Y$; and finally, in \eqref{cc:eq:frm:def:alph mut} the reverse channel $P_{\bar X|Y} $ is such that $P_Y \to P_{\bar X|Y} \to P_{\bar X}$ and the equality follows from the definition\footnote{Warning: In general, $\alpha$-mutual information is not a symmetric information measure \cite[Example 4]{verdu2015}. Hence, $I^\sfs_\lambda(P_Y,P_{\bar X|Y })\neq I^\sfs_\lambda(P_{\bar X}, P_{Y|X})$.} of $I^\sfs_\lambda(P_Y, P_{\bar X|Y})$, cf. \eqref{eqn:def:alpha mut info-semih}. 

Together with \eqref{eqn:rem:alpha<aleph}, \eqref{eqn:rem:reverse alpha<aleph} implies that
\begin{align}
  &\aleph(R,P_{\bar X},P_{Y|X}) \nonumber \\ 
  &\quad \ge \max\left\{\alpha(R, P_{\bar X}, P_{Y|\bar X}), \alpha(R, P_Y, P_{\bar X|Y})\right\} \text{.} \label{eqn:rem:cc:aleph bigger than max of two alpha}
\end{align}
Note that, even though we can only show $\ge$ above, as the following example illustrates, there are settings for which the inequality in \eqref{eqn:rem:cc:aleph bigger than max of two alpha} is strict.
\end{rem}
\begin{exa}[Binary Symmetric Channel]\label{exa:bsc alpha vs aleph} 
Suppose $\cX = \cY = \{0, 1\}$, and let $P_{Y|X}\colon \cX \to \cY$ be a binary symmetric channel with crossover probability $p = 0.05$ \cite[Section 7.1.4]{cover2012elements}. If $P_{\bar X}(0) = 2/5$, and $R=0.85>I(P_{\bar X}, P_{Y|X}) \approx 0.69   $ \texttt{\emph{bits}}, 
\begin{align}
  \alpha(0.85, P_{\bar X}, P_{Y|X}) &\approx 2.0429 \times 10^{-2} \text{,} \\
  \alpha(0.85, P_{Y}, P_{\bar X|Y}) &\approx 2.0585 \times 10^{-2} \text{,} \\
  \aleph(0.85, P_{\bar X}, P_{Y|X}) &\approx 2.2216 \times 10^{-2} \text{,}
\end{align} 
implying $\alpha \neq \aleph$, in general. 

\Cref{fig:comparison_alpha-aleph} depicts the gap between $\alpha(R,P_{\bar X}, P_{Y|X}) $, $\alpha(R,P_{Y}, P_{\bar X|Y})$, and $\aleph(R, P_{\bar X}, P_{Y|X})$. While \Cref{fig:comparison:alpha_vs_aleph_BSC} illustrates \Cref{exa:bsc alpha vs aleph} for different values of rate $R$, \Cref{fig:comparison:alpha_vs_aleph_BZC} illustrates the case for Binary Z-Channel \cite[Problem 7.8]{cover2012elements}  with the same input distribution $P_{\bar X}$ and the same error probability $p=0.05 $. 
\end{exa}
\begin{rem}
  If $P_{\bar X}\in \cP_m(\cX) $ is such that 
  \begin{align}
    \lim_{m\to \infty} P_{\bar X} = P_X \text{,}
  \end{align}
 for some $P_X \in \cP(\cX)$, assuming $R>I(P_{\bar X}, P_{Y|X})>0$ for all $m\in \bbN$, being a linear function of $P_{\bar X}$, it is straightforward to see that $\aleph(R, P_{\bar X}, P_{Y|X})$ is sequentially continuous in $P_{\bar X}$. That is, 
  \begin{align}
    \lim_{m\to \infty} \aleph(R, P_{\bar X}, P_{Y|X}) = \aleph(R, P_X, P_{Y|X})\text{.}
  \end{align}
\end{rem}
\begin{rem}
  Regarding the computation of the exact soft-covering exponents $\alpha $ and $\aleph$, the dual forms in \eqref{eqn:def:alpha(R,P_X,P_{Y|X})_the dual} and \eqref{eqn:cc:def:aleph_the dual} are far easier to calculate then their primal counterparts in \eqref{eqn:def:alpha(R,P_X,P_{Y|X})} and \eqref{eqn:cc:def:aleph}. This is because, in calculating the former pair, the optimizations are carried over spaces of dimensions\footnote{Observe that the calculation of $I^\sfc_\lambda(P_{\bar X}, P_{Y|X})$ is an optimization over a space of dimension $|\cY|-1 $, see Remark~\ref{rem:cc:csiszar-alpha-mutual information}.} $1$, and $|\cY|$, respectively, whereas in calculating the latter pair the optimizations are carried over spaces of dimensions $|\cX||\cY|-1 $ and $|\cX|(|\cY|-1) $, respectively. 
\end{rem}
\begin{rem}\label{rem:Taylor expansion of alpha mut information}
Taylor expansion of $I^\sfs_\lambda(P_X, P_{Y|X})$ around $\lambda =1$ yields
  \begin{align}
    &I^\sfs_\lambda(P_X, P_{Y|X}) = I(P_X, P_{Y|X}) \nonumber \\ 
    &\ \  + \frac12 \Var\left[\imath^{}_{X;Y}(X;Y)\right](\lambda -1) + \rmO((\lambda -1)^2) \text{,} 
  \end{align}
  where $(X,Y)\sim P_X P_{Y|X}$, and $\Var[\imath^{}_{X;Y}(X;Y)] $ denotes the variance\footnote{If $P_X$ is a capacity-achieving distribution, then $\Var[\imath^{}_{X;Y}(X;Y)]$ is a property of the channel known as the \emph{channel dispersion} \cite{PolyanskiyPoorVerdu10}. In our treatment, since it is not required that $P_X$ is capacity achieving, inspired by the name of its sibling \emph{varentropy} \cite{kontoyiannisVerdu13OptimalLossless}, we coin the term \emph{mutual varentropy} for $\Var[\imath^{}_{X;Y}(X;Y)]$.} of $\imath^{}_{X;Y}(X;Y)$. Hence, when $R = I(P_X, P_{Y|X}) +\epsilon $ for some small\footnote{When $R= I(P_X, P_{Y|X}) +\epsilon $, since $I^\sfs_\lambda(P_X, P_{Y|X}) $ is non-decreasing in $\lambda $ \cite[Theorem 4]{HoVerdu15}, the maximum in \eqref{forfootnote:maxima here} is achieved at a $\lambda$ value that is close to 1.} $\epsilon$, 
  \begin{align}
    &\alpha(R,P_X,P_{Y|X}) \nonumber \\ 
    &= \max_{\lambda\in [1,2] } \left\{ \frac{\lambda-1}{\lambda} \left(R - I^\sfs_\lambda (P_X, P_{Y|X}) \right)  \right\}  \label{forfootnote:maxima here}  \\
    &\approx \max_{\lambda\in [1,2]} \left\{\frac{\lambda-1}{\lambda}\left(\epsilon -  \frac{\lambda -1}{2} \Var\left[\imath^{}_{X;Y}(X;Y)\right]\right)\right\} \label{approx:maxima in dual} \\
    &\approx \frac{\epsilon^2}2 \Var^{-1}\left[\imath^{}_{X;Y}(X;Y)\right] \label{frm:apprx maximizer}  \\
    &= \frac12(R-I(P_X, P_{Y|X}))^2\Var^{-1}\left[\imath^{}_{X;Y}(X;Y)\right]  
      \text{,}
  \end{align}
 where the maximum in the right side of \eqref{approx:maxima in dual} is achieved when $\lambda =\left(1 + 2\epsilon \Var^{-1}[\imath^{}_{X;Y}(X;Y)]\right)^{1/2} $. For the sake of simplicity, supposing $\lambda = 1$ in the denominator of the right hand side of \eqref{approx:maxima in dual}, the approximate maximizer becomes $\lambda \approx 1+\epsilon \Var^{-1}[\imath^{}_{X;Y}(X;Y)] $ and \eqref{frm:apprx maximizer} follows. 
\end{rem}
\begin{rem}
 In a similar spirit to \Cref{rem:Taylor expansion of alpha mut information}, Taylor expansion of $\bbE[\log\bbE [\exp ( (\lambda -1)\, \imath^{}_{P_{\bar XY}\|P_{\bar X}S_Y}(\bar X, Y))|\bar X ]]$ around $\lambda = 1 $ yields
  \begin{align}
    &\bbE[\log\bbE [\exp ( (\lambda -1)\, \imath^{}_{P_{\bar XY}\|P_{\bar X}S_Y}(\bar X, Y))|\bar X ]] \nonumber \\
    &= (\lambda -1)D(P_{\bar XY}\|P_{\bar X}S_Y)  \\ 
    & + \frac12(\lambda-1 )^2 \Var\left[\imath^{}_{P_{\bar XY}\|P_{\bar X}S_Y}(\bar X, Y))\right] +\rmO((\lambda -1)^3) \text{,} \nonumber
  \end{align}
  where $(\bar X, Y) \sim P_{\bar X}P_{Y|X} $. Therefore, whenever $R = I(P_{\bar X},P_{Y|X}) +\epsilon $ for some small $\epsilon $, 
  \begin{align}
    &\aleph(R, P_{\bar X},P_{Y|X}) \nonumber \\
    &\approx \max_{\lambda \in [1,2]}\max_{S_Y} \bigg\{\frac{\lambda-1}{\lambda} \bigg( R-D(P_{\bar XY}\|P_{\bar X}S_Y) \\ 
    &\qquad \qquad \quad  - \frac{\lambda-1}{2} \Var\left[\imath^{}_{P_{\bar XY}\|P_{\bar X}S_Y}(\bar X, Y))\right]\bigg)   \bigg\} \nonumber \\
    &\approx \max_{\lambda\in [1,2]} \left\{\frac{\lambda-1}{\lambda}\left(\epsilon -  \frac{\lambda -1}{2} \Var\left[\imath^{}_{\bar X;Y}(\bar X;Y)\right]\right)\right\}  \\
    &\approx \frac{\epsilon^2}2 \Var^{-1}\left[\imath^{}_{\bar X;Y}(\bar X;Y)\right]  \\
    &\approx \alpha(R, P_{\bar X},P_{Y|X})
      \text{,}
  \end{align}
  which can also be observed in \Cref{fig:comparison_alpha-aleph}. 
\end{rem}
\begin{rem}\label{rem:less than R/2}
Since $Q_{XY} = P_{XY} $ and $Q_{Y|X} = P_{Y|X}$ are suboptimal choices, it is easy to see that
	\begin{align}
&\alpha(R, P_X, P_{Y|X}) \nonumber \\  
&\qquad = \min_{Q_{XY}} \bigg\{D(Q_{XY}\|P_{XY}) \nonumber \\ 
&\qquad \qquad \qquad + \frac12 \left[R-D(Q_{XY}\|P_XQ_Y) \right]_{+}\bigg\} \\ 
&\qquad \le \frac 12\left[R - I(P_X, P_{Y|X}) \right]_{+} \\
&\qquad < \frac R2 \text{,} \label{chan is no dubious}
	\end{align}
and
\begin{align}
&\aleph(R, P_{\bar X}, P_{Y|X}) \nonumber \\
 &\qquad  =   \min_{Q_{Y|X}} \bigg\{ D(P_{\bar X}Q_{Y|X}\|P_{\bar XY}) \nonumber \\ 
 &\qquad \qquad \quad + \frac12[R-D(P_{\bar X}Q_{Y|X}\|P_{\bar X}Q_Y)]_+ \bigg\} \\
 &\qquad \le \frac12 \left[R- I(P_{\bar X},P_{Y|X})\right]_+ \\ 
 &\qquad < \frac R2 \label{frm:cc:non-deg chan} \text{,}
\end{align}
where \eqref{chan is no dubious} and \eqref{frm:cc:non-deg chan} follow because the channel $P_{Y|X}$ is assumed to be non-degenerate. The same observation can be made from the dual forms of $\alpha(R, P_X, P_{Y|X}) $ and $\aleph(R, P_X, P_{Y|X}) $ in \eqref{eqn:def:alpha(R,P_X,P_{Y|X})_the dual} and \eqref{eqn:cc:def:aleph_the dual}, respectively.
\end{rem}
In what follows, \Cref{sec:lower bound,sec:upper bound} prove the lower and upper bound directions in \eqref{eqn:thm:exact soft-cov exponent}, respectively. \Cref{sec:alternative representations} proves the equivalence of the primal and dual forms of the exact soft-covering exponents, see \eqref{eqn:thm:exact soft-cov exponent_ALTERNATIVE} and \eqref{eqn:thm:cc:exponent_ALTERNATIVE}, finally \Cref{sec:comparisons} is devoted to the comparison of the previously known lower bounds on the soft-covering exponents $\alpha $ and $\aleph$.
\section{Proof of the Lower Bound in Theorem~\ref{thm:main}}\label{sec:lower bound}
This section establishes
\begin{align}
	&\liminf_{n\to \infty}-\frac{1}{n} \log \bbE \left[\left\| \induced - P_{Y^n} \right\|^{}_1\right] \nonumber \\ 
	&\qquad \qquad \ge \alpha(R, P_X, P_{Y|X})  \text{.}
\end{align} 
Indeed, using the finite block-length analysis, we shall prove the following stronger claim (see \Cref{thm:finite block-length exact scl} in \Cref{apdx:finite block-length results}):
\begin{align}
&-\frac{1}{n} \log \bbE \left[\left\| \induced - P_{Y^n} \right\|^{}_1\right] \nonumber \\
&\qquad \qquad  \ge \alpha_n(R, P_X, P_{Y|X}) - \kappa_n \text{,}
\end{align}
where $\alpha_n$ is as defined in \eqref{eqn:def:alpha_n for finite block-length} and the vanishing constant $\kappa_n $ depends only on the block-length $n$ and the alphabet sizes $|\cX|$ and $|\cY|$. 

Suppose that $P_{X^n}$ is the i.i.d. input distribution to the memoryless channel $P_{Y^n|X^n}$ generating the i.i.d. output distribution $P_{Y^n}$, i.e., suppose $P_{X^n}\to P_{Y^n|X^n}\to P_{Y^n}$. Inspired by \cite{parizi2017exact}, given $y^n\in \cY^n$, let 
\begin{align}
& L_{\scrC_M^n}(y^n) \nonumber  \\
&\quad =  \begin{cases}
	 \ds\frac{\induced(y^n)}{P_{Y^n}(y^n)}	& \text{if } P_{Y^n}(y^n)>0 \text{,} \\
	 1 & \text{otherwise.}
	\end{cases} \label{eqn:def:L_scrC}   \\
&\quad = \begin{cases}
	\ds \frac{1}{M} \sum_{j=1}^{M} \frac{P_{Y^n|X^n}(y^n|X_j^n)}{P_{Y^n}(y^n)} & \text{if }P_{Y^n}(y^n)>0 \text{,}\\
	1 & \text{otherwise.} 
	\end{cases}
\end{align}
Observe that $L_{\scrC_M^n}(y^n)$ is a random variable as it depends on the random codebook $\scrC_M^n$, and it is easy to see that
\begin{align}
	\bbE[L_{\scrC_M^n}(y^n)] = 1 \text{.}
\end{align}
Suppose $y^n\in \cY^n$, and let $Q_{\bar X|\bar Y}$ denote the conditional type of $x^n\in \cX^n$ given $y^n$ so that the joint type $Q_{\bar X\bar Y}$ of the sequence $(x^n, y^n)$ satisfies
\begin{align}
	Q_{\bar X \bar Y}(a,b) = Q_{\bar X | \bar Y} (a|b)Q_{\bar Y}(b) \text{,} \label{eqn:joint type in terms of cond. type}
\end{align}
where $Q_{\bar Y}$ denotes the type of $y^n$. Note that $y^n\in \cT_{Q_{\bar Y}}^n$ and $Q_{\bar X|\bar Y} \in \cP_n(\cX|Q_{\bar Y})$ together induce a joint type $Q_{\bar X\bar Y}$ via the relation in \eqref{eqn:joint type in terms of cond. type}. 

Assume $P_{Y^n}(y^n)>0$, since $P_{Y^n|X^n}(y^n|x^n)$ and $P_{Y^n}(y^n)$ depend on $(x^n, y^n)$ only through its joint type, using the type enumeration method \cite{merhav2010statistical, merhav2014exact}, one can write
\begin{align}
	&L_{\scrC_M^n}(y^n) \nonumber \\ 
	&\quad = \frac{1}{M} \sum_{Q_{\bar X|\bar Y} \in \cP_n(\cX|Q_{\bar Y})} N_{Q_{\bar X|\bar Y}}(y^n) l_{Q_{\bar X| \bar Y}}(y^n) \text{,} \label{eqn:def:L}
\end{align}
where
\begin{align}
	l_{Q_{\bar X| \bar Y}}(y^n)= \frac{P_{Y^n|X^n}(y^n| x_{Q_{\bar X|\bar Y}}^n)}{P_{Y^n}(y^n)}  \label{eqn:def:l_{Q_{barX|barY}}}
\end{align}
for some $x_{Q_{\bar X|\bar Y}}^n\in \cT^n_{Q_{\bar X|\bar Y}}(y^n)$, and the random variable
\begin{align}
	& N_{Q_{\bar X|\bar Y}}(y^n) \nonumber \\
	&\qquad =\left|\left\{X^n \in \scrC_M^n \colon X^n \in \cT^n_{Q_{\bar X|\bar Y}}(y^n)\right\}\right| \label{eqn:def:N_Q_{X|Y}}  \\
	&\qquad = \sum_{X^n \in \scrC_M^n} 1\left\{ X^n \in \cT^n_{Q_{\bar X|\bar Y}}(y^n)\right\}
\end{align}
denotes the number of random codewords in $\scrC_M^n $ which have conditional type $Q_{\bar X|\bar Y}$ given $y^n$. Since $\scrC_M^n $ contains $M$ independent codewords, it follows that $N_{Q_{\bar X|\bar Y}}(y^n)$ is a binomial random variable with cluster size $M$ and success probability 
\begin{align}
	p_{Q_{\bar X|\bar Y}}(y^n)=\bbP\left[X^n \in \cT^n_{Q_{\bar X|\bar Y}}(y^n)\right] \label{eqn:def:p_{Q_{barX|barY}}} \text{.}
\end{align}
For the remainder of this paper, it is crucial to note that both $l_{Q_{\bar X| \bar Y}}(y^n)$ and $p_{Q_{\bar X|\bar Y}}(y^n)$ depend on $y^n$ only through its type. 

Given $y^n\in \cT_{Q_{\bar Y}}^n$ and $Q_{\bar X|\bar Y} \in \cP_n(\cX|Q_{\bar Y}) $, define
\begin{align}
Z_{Q_{\bar X\bar Y}} &= \frac{1}{M} N_{Q_{\bar X|\bar Y}}(y^n) l_{Q_{\bar X| \bar Y}}(y^n) \text{,} \label{eqn:def:Z_Q_{barX|barY}} \\
\mfrY(M, Q_{\bar X\bar Y}) &= \min\left\{2p^{}_{Q_{\bar X|\bar Y}}(y^n), M^{-\frac 12} p_{Q_{\bar X|\bar Y}}^{\frac 12}(y^n)\right\} \text{,}  \label{eqn:def:mfrakY} 
\end{align}
and observe that 
\begin{align}
	&\bbE \left[\left\| \induced - P_{Y^n}\right\|^{}_1\right] \nonumber \\ 
	&= \sum_{y^n \in \cY^n } P_{Y^n}(y^n) \bbE \left[\left| L_{\scrC_M^n}(y^n) -1\right|\right] \label{frm:def of L}  \\
	&= \sum_{y^n\in \cY^n } P_{Y^n}(y^n)  \bbE \left[\left|\sum_{Q_{\bar X|\bar Y}} Z_{Q_{\bar X\bar Y}} -  \bbE[Z_{Q_{\bar X\bar Y}}]\right|\right] \label{from:def Z} \\
	&\le  \sum_{y^n \in \cY^n } P_{Y^n}(y^n) \sum_{Q_{\bar X|\bar Y}} \bbE \left[\left|Z_{Q_{\bar X\bar Y}} -  \bbE[Z_{Q_{\bar X\bar Y}}]\right|\right] \label{from:sum of abs vs abs of sum}\\
	&\le \sum_{y^n \in \cY^n } \sum_{Q_{\bar X|\bar Y}} \exp(-n\bbE[\imath^{}_{P_{Y|X}}(\bar Y|\bar X)]) \mfrY(M, Q_{\bar X\bar Y}) 
	\label{from:both are bigger} \\
%	&\quad=  \sum_{Q_{\bar Y}}  \sum_{y^n \in \cT^n_{Q_{\bar Y}}}  \sum_{Q_{\bar X|\bar Y}} \exp(-n\bbE[\imath^{}_{P_{Y|X}}(\bar Y|\bar X)]) \mfrY(M, Q_{\bar X\bar Y})   \\
	& =\sum_{Q_{\bar X\bar Y}} \left|\cT^n_{Q_{\bar Y}}\right|\exp(-n\bbE[\imath^{}_{P_{Y|X}}(\bar Y|\bar X)])  \mfrY(M, Q_{\bar X\bar Y}) \label{frm:type partitioning}   \\
	&\le |\cP_n(\cX\times\cY)|\times \label{inner max depends on outer max} \\ 
	& \max_{Q_{\bar X\bar Y}} \Big\{ \left|\cT^n_{Q_{\bar Y}}\right|   \exp(-n\bbE[\imath^{}_{P_{Y|X}}(\bar Y|\bar X)]) \mfrY(M, Q_{\bar X\bar Y})   \Big\}  \text{,}  \nonumber
\end{align}
where \eqref{frm:def of L} follows from the definition of $L_{\scrC_M^n}(y^n)$ in \eqref{eqn:def:L_scrC}; in \eqref{from:def Z} the inner summation is over the set of conditional types given $y^n \in \cT^n_{Q_{\bar Y}}$, namely $\cP_n(\cX|Q_{\bar Y})$, the equality follows from \eqref{eqn:def:L} and the definition of $Z_{Q_{\bar X\bar Y}}$ in \eqref{eqn:def:Z_Q_{barX|barY}}; \eqref{from:sum of abs vs abs of sum} follows from the triangle inequality; in \eqref{from:both are bigger} $(\bar X, \bar Y)\sim Q_{\bar X|\bar Y}Q_{\bar Y} =Q_{\bar X\bar Y}$, and the inequality is due to \Cref{lem:absolute mean deviation upper bound for Z_{Q_{barX|barY}}} in Appendix~\ref{apdx:main}; in \eqref{frm:type partitioning} the summation is over the set of joint types, $\cP_n(\cX\times\cY) $, while the equality follows from the type class partitioning of $\cY^n$,
\begin{align}
  \cY^n = \bigsqcup_{Q_{\bar Y}\in \cP_n(\cY) } \cT^n_{Q_{\bar Y}} \text{,}
\end{align}
and because\footnote{Also see Remarks~\ref{rem:cP_n(cX|Q_{bar Y} depends only on the type Q_{bar Y}}~and~\ref{rem:joint types vs right coset of conditional types}.} the summand depends on $y^n$ only through its type. Denoting
\begin{align}
	\cP_\infty(\cX\times \cY) = \bigcup_{n \in \bbN} \cP_n(\cX\times \cY) \text{,}
\end{align}
it follows from \eqref{inner max depends on outer max} that 
\begin{align}
	&\liminf_{n\to \infty} -\frac1n \log \bbE \left[\left\|\induced - P_{Y^n}\right\|^{}_1\right] \nonumber\\ 
	&\qquad \ge \inf_{Q_{\bar X \bar Y} \in \cP_{\infty}(\cX \times \cY)} \bigg\{ D(Q_{\bar X \bar Y } \| P^{}_{XY}) \label{from:csiszar and bucnh of lemmas}  \\ 
	&\qquad \qquad \qquad  + \frac12 \left[R- D(Q_{\bar X\bar Y}\|P^{}_XQ_{\bar Y})  \right]_+ \bigg\}  \nonumber \\
	&\qquad = \min_{Q_{XY} \in \cP(\cX \times \cY)} \bigg\{ D(Q_{XY} \| P_{XY}) \label{from:apdx:optimizations over types} \\ 
	&\qquad \qquad \qquad + \frac12 \left[R- D(Q_{XY}\|P_XQ_Y)  \right]_+ \bigg\}  \nonumber \\
	&\qquad = \alpha(R, P_X, P_{Y|X})  \text{,}
\end{align}
where in \eqref{from:csiszar and bucnh of lemmas} we use the fact that the size of the set $|\cP_n(\cX\times\cY)| $ grows polynomially in $n$, see \cite[Lemma 2.2]{csiszar2011information}, and \Cref{lem:infimum over all types} in Appendix~\ref{apdx:asymptotic exponents}; and finally \eqref{from:apdx:optimizations over types} follows from \Cref{lem:optimization over types in the limit} in Appendix~\ref{apdx:optimizations over types in the limit}. 
\hfill \qedsymbol 
\begin{rem}\label{rem:cc:replacements lower bound}
In the constant-composition case,\footnote{See \Cref{def:set of conditional types with fixed marginal} for the definition of the set of conditional types with fixed marginals, i.e., $\condisset$.}
 \begin{align}
  &\breve{L}_{\scrD_M^n}(y^n)  \nonumber\\
  &\quad = \begin{cases}
	 \ds\frac{\ccinduced(y^n)}{R^{}_{\breve Y^n}(y^n)}	& \text{if } R^{}_{\breve Y^n}(y^n)>0 \text{,} \\
	 1 & \text{otherwise.}
	\end{cases}  \\ 
  &\quad = \frac{1}{M} \sum_{Q_{\bar X | \bar Y} \in \condisset} \breve{N}_{Q_{\bar X|\bar Y}}(y^n) \breve{l}_{Q_{\bar X| \bar Y}}(y^n) \text{,} \label{eqn:def:alt:breve{L}_{scrD_M}}
 \end{align}
with
\begin{align}
  \breve{l}_{Q_{\bar X| \bar Y}}(y^n) &= \frac{P_{Y^n|X^n}(y^n|x^n_{Q_{\bar X|\bar Y}} )}{R^{}_{\breve Y^n}(y^n)} \text{,} \label{for:cc:remark_begin} \\ 
	\breve{N}_{Q_{\bar X|\bar Y}}(y^n)&=\left|\left\{\breve X^n \in \scrD_M^n \colon \breve X^n \in \cT^n_{Q_{\bar X|\bar Y}}(y^n)\right\}\right| \text{,}
\end{align}
and\footnote{In \eqref{for:cc:remark_end}, since the $\cX$-marginal of the joint types is fixed to be $P_{\bar X} $, $P_{\bar X}Q_{\bar Y|\bar X} = Q_{\bar X|\bar Y} Q_{\bar Y}$ where $Q_{\bar Y} $ is the type of $y^n$.}
\begin{align}
\breve{p}_{Q_{\bar X|\bar Y}}(y^n) &= \bbP\left[\breve X^n \in \cT^n_{Q_{\bar X|\bar Y}}(y^n)\right] \text{,} \\
  \breve{Z}_{Q_{\bar X\bar Y}} &= \frac{1}{M} \breve{N}_{Q_{\bar X|\bar Y}}(y^n) \breve{l}_{Q_{\bar X| \bar Y}}(y^n) \text{,} \label{eqn:def:cc:breve{Z}_{Q_{bar X|bar Y}}(y^n) AFTER REVISION} \\
 \breve{\mfrY}(M, P_{\bar X}Q_{\bar Y|\bar X})  &= \label{for:cc:remark_end} \\ 
 & \hspace{-1.5em}  \min\left\{2\breve{p}_{Q_{\bar X|\bar Y}}(y^n), M^{-\frac 12} \breve{p}^{\frac12}_{Q_{\bar X|\bar Y}}(y^n) \right\} \text{.} \nonumber 
\end{align} 
The steps \eqref{frm:def of L}--\eqref{from:apdx:optimizations over types} remain almost identical except one needs to keep in mind that $\cX$-marginal of the joint types $Q_{\bar X\bar Y} $ is fixed to be $P_{\bar X}$ and replace\footnote{See \Cref{def:set of joint types with fixed marginals} for the definition of the set of joint types with fixed $\cX$-marginal $P_{\bar X}$, i.e., $\cP_n(\cX\times \cY; P_{\bar X}\times \cdot)$.}
\begin{align*}
P_{X^n} &\leftarrow R_{\breve X^n} \text{,} \\
\induced &\leftarrow \ccinduced \text{,} \\
P_{Y^n}  &\leftarrow R_{\breve Y^n}  \text{,} \\
L_{\scrC_M^n} &\leftarrow \breve{L}_{\scrD_M^n} \text{,} \\
\cP_n(\cX|Q_{\bar Y}) &\leftarrow \cP_n(\cX|Q_{\bar Y};P_{\bar X}) \text{,} \\
\cP_n(\cX\times\cY) &\leftarrow \cP_n(\cX\times \cY; P_{\bar X}\times \cdot) \text{,}\\
\cP_\infty(\cX\times \cY) &\leftarrow  \bigcup_{n\in m\bbN}  \cP_{n}(\cY|P_{\bar X}) \text{,}\\
\text{Remarks~\ref{rem:cP_n(cX|Q_{bar Y} depends only on the type Q_{bar Y}}~and~\ref{rem:joint types vs right coset of conditional types}} &\leftarrow \text{Remarks~\ref{rem:cc:cP_n(cX|Q_{bar Y};P_{bar X}) depends on y^n through its type}~and~\ref{rem:cc:coset notation for set of conditional types}} \text{,}\\
\text{\Cref{lem:infimum over all types,lem:optimization over types in the limit}} &\leftarrow \text{\Cref{lem:cc:key first limit,lem:cc:optimization over conditional types in the limit}} \text{,} 
\end{align*}
together with proper replacement of the terms defined in \eqref{for:cc:remark_begin}--\eqref{for:cc:remark_end}.
\end{rem}
\begin{rem}
	It should be noted that the key step of the lower bound proof is the bound in \eqref{from:both are bigger}. In that step,  the mean and the standard deviation of each of the random variables $Z_{Q_{\bar X\bar Y}}$ are directly used as the upper bound for each conditional type $Q_{\bar X|\bar Y}\in\cP_n(\cX|Q_{\bar Y})$. In previous soft-covering exponent analysis \cite{hayashi2006general, cuff2013distributed}, the set of the conditional types $\cP_n(\cX|Q_{\bar Y}) $ is first partitioned into two sets containing the so-called typical and atypical conditional types according to a threshold on $l_{Q_{\bar X|\bar Y}}(y^n)$. Then, the standard deviation bound is applied on the typical set whereas the mean bound is applied on the atypical one. Although this ``partition by joint probability first, bound later" technique is also espoused in the exact exponent analysis of the relative entropy variant of the soft-covering lemma \cite{parizi2017exact}, it turns out to be a suboptimal method for the total variation distance.  
\end{rem}

\begin{rem}
	Thanks to the analysis on the absolute mean deviation of binomial distribution provided in \cite[Theorem 1]{berend2013sharp}, the mean and standard deviation bound applied in \Cref{lem:absolute mean deviation upper bound for Z_{Q_{barX|barY}}} can be shown to be tight. 
\end{rem}
\section{Proof of the Upper Bound in Theorem~\ref{thm:main}}\label{sec:upper bound}
This section establishes
\begin{align}
	&\limsup_{n\to \infty}-\frac{1}{n} \log \bbE \left[\left\| \induced - P_{Y^n} \right\|^{}_1\right] \nonumber \\ 
	 &\qquad \qquad  \le \alpha(R, P_X, P_{Y|X})  \text{.} \label{eqn:main upper bound}
\end{align} 
Indeed, using the finite block-length analysis, we shall prove the following stronger claim (see \Cref{thm:finite block-length exact scl} in \Cref{apdx:finite block-length results}):
\begin{align}
&-\frac{1}{n} \log \bbE \left[\left\| \induced - P_{Y^n} \right\|^{}_1\right] \nonumber \\
&\qquad \qquad \le  \alpha_n(R, P_X, P_{Y|X}) + \upsilon_n \text{,}
\end{align}
where $\alpha_n$ is as defined in \eqref{eqn:def:alpha_n for finite block-length} and the vanishing constant $\upsilon_n $ depends on the block-length $n$, the alphabet sizes $|\cX|$ and $|\cY|$, and the joint distribution $P_XP_{Y|X}$. 

The biggest obstacle in showing \eqref{eqn:main upper bound} is the mutual dependences of the the random variables\footnote{One quick way to see these mutual dependences is that the sum of $N_{Q_{\bar X|\bar Y}}(y^n)$ over all conditional types $Q_{\bar X | \bar Y} \in \cP_n(\cX|Q_{\bar Y})$ is equal to $M$.} $N_{Q_{\bar X|\bar Y}}(y^n)$, as defined in \eqref{eqn:def:N_Q_{X|Y}}. Note that, given two distinct conditional types (given $y^n\in \cY^n$), say $Q_{\bar X|\bar Y}$ and $R_{\bar X|\bar Y}$, the random variables $N_{Q_{\bar X|\bar Y}}(y^n)$ and $N_{R_{\bar X|\bar Y}}(y^n)$ are not independent from each other. Fortunately, their dependence can be shown to be negligible. Indeed, instead of assuming that the number of codewords $M$ in the codebook $\scrC_M^n$ is a deterministic number $\lceil \exp(nR) \rceil$, if one assumes that it is Poisson distributed with mean $\mu_n = \exp(nR)$, then $N_{Q_{\bar X|\bar Y}}(y^n)$ becomes a \emph{Poisson splitting} of the codewords in $\scrC_M^n$. In that case, given two distinct conditional types $Q_{\bar X|\bar Y}$ and $R_{\bar X|\bar Y}$, the random variables $N_{Q_{\bar X|\bar Y}}(y^n)$ and $N_{R_{\bar X|\bar Y}}(y^n)$ correspond to two distinct Poisson splits and they become independent from one another. This turns out to be the gateway in proving the pseudo-upper bound in the case when $M$ is Poisson distributed. However, to prove the upper bound for the actual statement in \Cref{thm:main}, the auxiliary assumption that the codebook $\scrC_M^n$ contains a random number of codewords needs to be eliminated, which can be done with the help of \Cref{lem:TV_concentration-McDiarmid-deterministic}. As already mentioned in \Cref{rem:foreshadowing the upper bound proof}, it is possible to prove a result similar to \Cref{lem:TV_concentration-McDiarmid-deterministic} with the assumption that $M$ is Poisson distributed, see \Cref{lem:TV-concentration_random-poisson} in Appendix~\ref{apdx:main}. This result can be utilized to show that it is immaterial whether $M$ is Poisson distributed or $M=\lceil \exp(nR) \rceil $ that \eqref{eqn:main upper bound} holds.  

To provide a more transparent presentation, the upper bound proof is divided into three subsections: \Cref{sub:poissonize} introduces the auxiliary assumption that the codebook size $M$ is Poisson distributed with mean $\mu_n = \exp(nR)$, \Cref{sub:upper bound when poissonized} provides the pseudo-upper bound proof under the assumption that $M$ is Poisson distributed, and finally, \Cref{sub:depoissonize} shows that, removing the auxiliary assumption by conditioning on $M=\lceil \mu_n \rceil $, one still cannot do better than $\alpha(R, P_X, P_{Y|X})$.

\subsection{Poissonization}\label{sub:poissonize}

Suppose, for the moment, that $M$ is Poisson distributed with mean $\mu_n = \exp(nR)$. In that case, using the established notation so far, for each $y^n \in \cT^n_{Q_{\bar Y}}$ and each $Q_{\bar X|\bar Y} \in \cP_n(\cX|Q_{\bar Y})$, the random variable
\begin{align}
	N_{Q_{\bar X|\bar Y}}(y^n) &= \sum_{X^n \in \scrC_M^n} 1\left\{ X^n \in \cT^n_{Q_{\bar X|\bar Y}}(y^n)\right\} \label{lbl:for:remark}
\end{align}
is a Poisson splitting of $M$ with mean 
\begin{align}
\mu_n \, p_{Q_{\bar X|\bar Y}}(y^n) = \exp(nR) \bbP\left[X^n \in \cT^n_{Q_{\bar X|\bar Y}}(y^n)\right] \text{.}
\end{align}
Moreover, as the random variables $N_{Q_{\bar X|\bar Y}}(y^n)$ and $N_{R_{\bar X|\bar Y}}(y^n)$ correspond to different bins defined by different conditional types $Q_{\bar X|\bar Y}$ and $R_{\bar X|\bar Y} \in \cP_n(\cX|Q_{\bar Y})$, they are independent from each other.
 
 Choose $\delta \in (0, 1)$, and note that for any $y^n\in \cY^n $ an application of \Cref{lem:Paul's saving Lemma} in Appendix~\ref{apdx:main} with 
 \begin{align*}
 W	&\leftarrow M\left|\induced(y^n) - P_{Y^n}(y^n)\right| \text{,} \\
 X&\leftarrow M \text{,} \\
 c &\leftarrow (1+\delta)\mu_n \text{,}
 \end{align*} 
yields 
\begin{align}
	&(1+\delta)\mu_n\bbE\left[\left|\induced(y^n) - P_{Y^n}(y^n)\right|\right] \nonumber \\
	&\qquad \ge \bbE\left[M \left|\induced(y^n) - P_{Y^n}(y^n) \right|\right] \label{eqn:messy step1} \\
	&\qquad \qquad \qquad  - \bbE[M 1\{M > (1+\delta)\mu_n \} ]  \text{.} \nonumber
\end{align}
On one hand, regarding the first term in the right side of \eqref{eqn:messy step1}, the triangle inequality implies 
\begin{align}
& \bbE\left[M \left|\induced(y^n) - P_{Y^n}(y^n) \right|\right] \nonumber \\
	 &\qquad \ge \bbE\left[\left|M \induced(y^n)- \mu_n P_{Y^n}(y^n) \right|\right] \label{from:triangle dude} \\ 
	 &\qquad \qquad	  - \bbE[|M-\mu_n| P_{Y^n}(y^n)] \nonumber  \\
	 &\qquad \ge \bbE\left[\left|M \induced(y^n) - \mu_n P_{Y^n}(y^n) \right|\right] \label{frm:upp bd on abs mean dev of poiss} \\ 
	 &\qquad \qquad - \sqrt{\mu_n} P_{Y^n}(y^n) \text{,} \nonumber
\end{align}
where \eqref{frm:upp bd on abs mean dev of poiss} follows from Jensen's inequality: 
\begin{align}
\bbE^2[|M-\mu_n|] &\le \bbE\left[|M-\mu_n|^2\right] \\
&= \mu_n \text{.}
\end{align}
On the other hand, regarding the second term in the right side of \eqref{eqn:messy step1}, 
 \begin{align}
 	\bbE[M 1\{M > (1+\delta)\mu_n \} ] \le \mu_n a^{\mu_n}_{\delta - \frac1{\mu_n}} \text{,} \label{frm:poiss awy frm mean}
 \end{align}
 which\footnote{The bound in \eqref{frm:poiss awy frm mean} is valid only when $\delta > \frac1{\mu_n} $. Even though the choice of $\delta\in (0,1)$ does not depend on $\mu_n =\exp(nR)$, the applicability of \Cref{lem:poisson expectation away from mean} is guaranteed for large enough $n$.} is a consequence of \Cref{lem:poisson expectation away from mean} in Appendix~\ref{apdx:main}. Note that, in the right side of \eqref{frm:poiss awy frm mean}, $a_\epsilon$ is a constant that satisfies $a_\epsilon <1 $ for all $\epsilon\in (0,1)$, which is explicitly defined in \eqref{def:a_epsilon}.

  Assembling \eqref{eqn:messy step1}, \eqref{frm:upp bd on abs mean dev of poiss} and \eqref{frm:poiss awy frm mean},
 \begin{align}
 	&(1+\delta)\bbE\left[\left|\induced(y^n) - P_{Y^n}(y^n)\right|\right] \nonumber \\ 
 	&\qquad  \ge \frac{1}{\mu_n}\bbE\left[\left|M \induced(y^n) - \mu_n P_{Y^n}(y^n) \right|\right] \label{eqn:messy sect son} \\
 	&\qquad \qquad - \frac{P_{Y^n}(y^n)}{\sqrt{\mu_n}} - a^{\mu_n}_{\delta - \frac1{\mu_n}} \nonumber \text{.} 
 \end{align}
 The first term in the right side of \eqref{eqn:messy sect son} is the term of main interest whose in-depth analysis is provided in the next subsection. 
\begin{rem}\label{rem:cc:replacements in poissonization}
   To get the counter-part of \eqref{eqn:messy sect son} in the random constant-composition codebook case, using the quantities defined in \Cref{rem:cc:replacements lower bound}, all one needs to do throughout \eqref{lbl:for:remark}--\eqref{eqn:messy sect son} is to replace\footnote{See \Cref{def:set of conditional types with fixed marginal} for the definition of the set of conditional types with fixed marginals, i.e., $\condisset$.} 
\begin{align*}
N_{Q_{\bar X|\bar Y}}(y^n) &\leftarrow \breve{N}_{Q_{\bar X|\bar Y}}(y^n) \text{,}  \\
p_{Q_{\bar X|\bar Y}}(y^n) &\leftarrow \breve{p}_{Q_{\bar X|\bar Y}}(y^n) \text{,} \\
\cP_n(\cX|Q_{\bar Y}) &\leftarrow \cP_n(\cX|Q_{\bar Y};P_{\bar X}) \text{,} \\
P_{X^n} &\leftarrow R_{\breve X^n} \text{,} \\
\induced &\leftarrow \ccinduced \text{,} \\
P_{Y^n} &\leftarrow R_{\breve Y^n} \text{.}
\end{align*} 
\end{rem}
 
 \subsection{Pseudo-Upper Bound Proof Assuming $M$ is Poisson Distributed}\label{sub:upper bound when poissonized}
 
 Capitalizing on the result of the previous subsection, 
 \begin{align}
   &(1+\delta)\bbE \left[\left\| \induced - P_{Y^n}\right\|^{}_1\right] \nonumber \\
   &\qquad =\sum_{y^n \in \cY^n} (1+\delta) \bbE\left[\left|\induced(y^n) - P_{Y^n}(y^n) \right|\right] \\
   &\qquad \ge  \sum_{y^n \in \cY^n} \frac 1{\mu_n}\bbE\left[\left|M \induced(y^n) - \mu_n P_{Y^n}(y^n) \right|\right] \label{eqn:sumterm} \\
   &\qquad \qquad  - \frac{1}{\sqrt{\mu_n}} - |\cY|^n a_{\delta-\frac 1{\mu_n}}^{\mu_n}  \text{.}  \nonumber 
 \end{align}
This section focuses on the summation in the right side of \eqref{eqn:sumterm} and shows that its exponent is $\alpha(R, P_X, P_{Y|X})$. As will be seen, the remaining terms in the right side of \eqref{eqn:sumterm} are residual terms whose exponents are greater than\footnote{In the sense that they vanish with a faster rate with $n$.} $\alpha(R, P_X, P_{Y|X}) $, and therefore, they do not contribute to the overall exponential decay rate of $\bbE [\|\induced - P_{Y^n}\|_1]$. 
  
 To this end, invoking the lemmas provided in Appendix~\ref{apdx:main},
 \begin{align}
 	&\sum_{y^n \in \cY^n} \frac 1{\mu_n}\bbE\left[\left|M \induced(y^n) - \mu_n P_{Y^n}(y^n) \right|\right] \nonumber \\
 	 &= \sum_{y^n \in \cY^n } \frac{P_{Y^n}(y^n)}{\mu_n} \bbE\left[\left| M L_{\scrC_M^n}(y^n) -\mu_n \right|\right] \label{frm:def:of:L agn in upp bd} \\
 	 &= \sum_{y^n \in \cY^n } \frac{P_{Y^n}(y^n)}{\mu_n} \bbE\bigger[ \bigger| \sum_{ Q_{\bar X|\bar Y}\in  \cP_n(\cX|Q_{\bar Y})}l_{Q_{\bar X|\bar Y}}(y^n) \nonumber  \\ 
 	 &\qquad \times  \left(N_{Q_{\bar X|\bar Y}}(y^n) - \bbE[N_{Q_{\bar X|\bar Y}}(y^n)]\right) \bigger| \bigger] \label{frm:mean of M*L} \\
 	 &\ge \sum_{y^n \in \cY^n }  \max_{Q_{\bar X| \bar Y}\in \cP_n(\cX|Q_{\bar Y})}\Bigg\{ \frac{P_{Y^n|X^n}(y^n|x^n_{Q_{\bar X|\bar Y}})}{\mu_n} \nonumber \\ 
 	 &\qquad \times \bbE\left[\left| N_{Q_{\bar X|\bar Y}}(y^n) - \bbE\left[ N_{Q_{\bar X|\bar Y}}(y^n) \right]\right|\right]\Bigg\} \label{frm:sum of indep rvs greater than max} \\
 	 &\ge \frac14  \sum_{Q_{\bar Y} \in \cP_n(\cY)} \sum_{y^n \in \cT^n_{Q_{\bar Y}}} \max_{Q_{\bar X|\bar Y}\in \cP_n(\cX|Q_{\bar Y})}\Big\{  \nonumber \\ 
 	 &\qquad \ \ \ \exp(-n\bbE[\imath^{}_{P_{Y|X}}(\bar Y|\bar X)])  \mfrY(\mu_n, Q_{\bar X\bar Y}) \Big\}    \label{frm:low bnd on abs mean dev for poiss} \\	
 	 &= \frac14 \sum_{Q_{\bar Y} \in\cP_n(\cY)} \max_{Q_{\bar X|\bar Y} \in\cP_n(\cX|Q_{\bar Y})} \Big\{ \left|\cT^n_{Q_{\bar Y}}\right| \nonumber \\ 
 	 &\qquad \times  \exp(-n\bbE[\imath^{}_{P_{Y|X}}(\bar Y|\bar X)]) \mfrY(\mu_n, Q_{\bar X\bar Y})\Big\} \label{frm:reglar sum into type sum} \\
	 &\ge \frac14 \max_{Q_{\bar X\bar Y} \in \cP_n(\cX\times\cY)} \Big\{\left|\cT^n_{Q_{\bar Y}}\right| \nonumber \\ 
	 &\qquad \times  \exp(-n\bbE[\imath^{}_{P_{Y|X}}(\bar Y|\bar X)]) \mfrY(\mu_n, Q_{\bar X\bar Y}) \Big\}  \text{,} \label{frm:sum of non-neg itms}
 \end{align}
 where \eqref{frm:def:of:L agn in upp bd} follows from the definition of $ L_{\scrC_M^n}(y^n)$ in \eqref{eqn:def:L_scrC}; \eqref{frm:mean of M*L} follows from the type enumeration method, see \eqref{eqn:def:L}, and \Cref{lem:expected value of M*L_scrC}; the key step in \eqref{frm:sum of indep rvs greater than max} follows from \Cref{lem:abs sum exp and max abs exp} and the definition of $l_{Q_{\bar X|\bar Y}}(y^n)$ in \eqref{eqn:def:l_{Q_{barX|barY}}}; in \eqref{frm:low bnd on abs mean dev for poiss} the function $\mfrY(\mu_n, Q_{\bar X\bar Y})$ is as defined in \eqref{eqn:def:mfrakY} and the bound follows from \Cref{lem:lower bd on abs mean dev poiss}; in \eqref{frm:reglar sum into type sum} $(\bar X, \bar Y)\sim Q_{\bar X|\bar Y}Q_{\bar Y}$ and the equality follows because\footnote{Also see \Cref{rem:cP_n(cX|Q_{bar Y} depends only on the type Q_{bar Y}}.} the summand depends on $y^n$ only through its type; and finally, \eqref{frm:sum of non-neg itms} follows because the right side of \eqref{frm:reglar sum into type sum} is a sum of non-negative numbers.\footnote{Also see \Cref{rem:joint types vs right coset of conditional types}.}
 
 Note that
 \begin{align}
 	&\lim_{n\to \infty} -\frac1n\log \max_{Q_{\bar X\bar Y} \in \cP_n(\cX\times\cY)}  \Big\{ \left|\cT^n_{Q_{\bar Y}}\right| \nonumber \\ 
 	&\qquad \times \exp(-n\bbE[\imath^{}_{P_{Y|X}}(\bar Y|\bar X)]) \mfrY(\mu_n, Q_{\bar X\bar Y})\Big\} \nonumber \\ 
	&\qquad = \inf_{Q_{\bar X \bar Y} \in \cP_{\infty}(\cX \times \cY)} \bigg\{ D(Q_{\bar X\bar Y}\|P^{}_{XY}) \label{important limit} \\ 
	&\qquad \qquad \qquad + \frac 12 \left[R-D(Q_{\bar X\bar Y}\|P^{}_XQ_{\bar Y})  \right]_+ \bigg\} \nonumber \\
	&\qquad = \min_{Q_{XY} \in \cP(\cX \times \cY)} \bigg\{ D(Q_{XY} \| P_{XY}  )  \label{from:apdx:optimizations over types in the upper bound} \\ 
	&\qquad \qquad \qquad  + \frac12 \left[R- D(Q_{XY}\|P_XQ_Y)  \right]_+ \bigg\} \nonumber   \\
	&\qquad = \alpha(R, P_X, P_{Y|X})  \text{,}
 \end{align}
 where \eqref{important limit} is thanks to \Cref{lem:infimum over all types} in \Cref{apdx:asymptotic exponents} while \eqref{from:apdx:optimizations over types in the upper bound} follows from \Cref{lem:optimization over types in the limit} in Appendix~\ref{apdx:optimizations over types in the limit}.
 
 On the other hand, going back to \eqref{eqn:sumterm}, the fact that $\mu_n=\exp(nR) $ and $a_\epsilon <1 $ for all $\epsilon \in (0,1)$ implies
 \begin{align}
 	 -\frac1n \log \frac1{\sqrt{\mu_n}} &= \frac R2 \text{,} \label{lim:sqrt mu} \\
 	\lim_{n \to \infty} -\frac1n \log \left(|\cY|^n a_{\delta-\frac 1{\mu_n}}^{\mu_n}\right)  &= \infty \text{.} \label{lim:doub exponential}
 \end{align}
Since the right side of \eqref{from:apdx:optimizations over types in the upper bound} is strictly less than $R/2$, see \Cref{rem:less than R/2}, it follows from \eqref{eqn:sumterm}, and \eqref{frm:sum of non-neg itms}--\eqref{lim:doub exponential} that, when $M$ is a Poisson distributed random variable with mean $\mu_n = \exp(nR) $, 
\begin{align}
	&\limsup_{n\to \infty}-\frac{1}{n} \log \bbE \left[\left\| \induced - P_{Y^n} \right\|^{}_1\right] \nonumber \\ 
	&\qquad \qquad \le \alpha(R, P_X, P_{Y|X})   \text{.} \label{eqn:upper bound when poisson}
\end{align} 
\begin{rem}\label{rem:cc:replacements pseudo-upper}
In the constant-composition case, in addition to the replacements mentioned in \Cref{rem:cc:replacements in poissonization}, replace\footnote{See \Cref{def:set of joint types with fixed marginals} for the definition of the set of joint types with fixed $\cX$-marginal $P_{\bar X}$, i.e., $\cP_n(\cX\times \cY; P_{\bar X}\times \cdot)$.}
\begin{align*}
L_{\scrC_M^n}   &\leftarrow \breve{L}_{\scrD_M^n} \text{,} \\
l_{Q_{\bar X|\bar Y}}(y^n) &\leftarrow \breve{l}_{Q_{\bar X|\bar Y}}(y^n) \text{,} \\
\mfrY(\mu_n, Q_{\bar X\bar Y}) &\leftarrow \breve{\mfrY}(\mu_n, P_{\bar X}Q_{\bar Y |\bar X}) \text{,} \\
\cP_n(\cX\times\cY) &\leftarrow \cP_n(\cX\times \cY; P_{\bar X}\times \cdot) \text{,}\\
\cP_\infty(\cX\times \cY) &\leftarrow  \bigcup_{n\in m\bbN}  \cP_{n}(\cY|P_{\bar X}) \text{,}\\
\text{Remarks~\ref{rem:cP_n(cX|Q_{bar Y} depends only on the type Q_{bar Y}}~and~\ref{rem:joint types vs right coset of conditional types}} &\leftarrow \text{Remarks~\ref{rem:cc:cP_n(cX|Q_{bar Y};P_{bar X}) depends on y^n through its type}~and~\ref{rem:cc:coset notation for set of conditional types}} \text{,}\\
\text{\Cref{lem:infimum over all types,lem:optimization over types in the limit}} &\leftarrow  \text{\Cref{lem:cc:key first limit,lem:cc:optimization over conditional types in the limit}} \text{,}
\end{align*}  
and keep in mind that the $\cX$-marginal of the joint types $Q_{\bar X\bar Y} $ is fixed to be $P_{\bar X}$. 
\end{rem}
\begin{rem}
 In order for the key step in \eqref{frm:sum of indep rvs greater than max} to be valid, independence among $N_{Q_{\bar X|\bar Y}}(y^n) $ is a must. This is the reason why poissonization was applied. 
\end{rem}
\subsection{Depoissonization}\label{sub:depoissonize}
To prove the upper bound in \Cref{thm:main}, it remains to show that the result established in \eqref{eqn:upper bound when poisson} still holds when $M = \lceil \exp(nR) \rceil$. To this end, once again utilizing the fact that $\alpha(R, P_X, P_{Y|X}) < R/2 $, choose $r \in (\alpha,  R/2) $, let $\epsilon_n = \exp(-nr) $, define the random variable\footnote{Randomness is because of the random codebook $\scrC_m^n$.}
\begin{align}
	T_n(m) =\left\| P_{Y^n|\scrC_m^n} - P_{Y^n} \right\|^{}_1 \text{,}
\end{align}
and consider the following three events:
\begin{align}
	\cA_{n} &= \{|\bbE[T_n(\lceil \mu_n \rceil)] - \bbE [T_n(M)]| < 2\epsilon_n\} \text{,} \\
	\cB_{n} &= \{|T_n(\lceil \mu_n \rceil) - \bbE[T_n(\lceil \mu_n \rceil)] | < \epsilon_n \} \text{,}\\
	\cC_{n} &= \{|T_n(\lceil \mu_n \rceil) - \bbE [T_n(M)]  | < \epsilon_n \} \text{,}
\end{align}
where $T_n(\lceil \mu_n \rceil)$ denotes the case when the codebook is assumed to have a deterministic number of codewords and $T_n(M)$ denotes the case when the codebook is assumed to have a random (Poisson) number of codewords. 

Observe that
\begin{align}
	\bbP[\cA_n]  &\ge \bbP[\cB_n \cap \cC_n] \label{frm:tri mri} \\
	&\ge 1- \bbP[\cB^c_n] - \bbP[\cC^c_n] \label{frm:unionization} \\  
	&\ge 1- \left(2+16\lceil \mu_n \rceil^\frac12\right) \exp_\rme \left(-\frac{\mu_n \epsilon_n^2}{2 + \epsilon_n^2 }\right) \text{,} \label{frm:lemma lemm lem} 
\end{align}
where \eqref{frm:tri mri} is because $\cA_n \supset \cB_n \cap \cC_n$; \eqref{frm:unionization} is the union bound; and \eqref{frm:lemma lemm lem} follows from \Cref{lem:TV_concentration-McDiarmid-deterministic,lem:TV concentrants to poisson-TV mean} in \Cref{sec:main} and Appendix~\ref{apdx:main}, respectively. Thanks to the choice of $\epsilon_n$, for large enough $n$, the right side of \eqref{frm:lemma lemm lem} is strictly greater than $0$. Moreover, since $\cA_n $ is a deterministic event, $\bbP[\cA_n ] >0$ implies that $ \bbP[\cA_n]=1$. That is, for large enough $n$, and $r\in (\alpha, R/2)$,
\begin{align}
 \bbE[T_n(\lceil \mu_n \rceil)] > \bbE [T_n(M)] - 2\exp(-nr) \text{.} \label{eqn:fixed grt thn pois -exp(-nr)}
\end{align}
Hence, it follows that
\begin{align}
	&\limsup_{n \to \infty}  -\frac 1n \log \bbE[T_n(\lceil \mu_n \rceil)] \nonumber \\
	&\quad \le \limsup_{n \to \infty} -\frac 1n \log( \bbE [T_n(M)] - 2\exp(-nr)) \\
	&\quad \le \alpha(R, P_X, P_{Y|X}) \text{,} \label{frm:upper bound with poiss}
\end{align}
where \eqref{frm:upper bound with poiss} is due to \eqref{eqn:upper bound when poisson}. 
\hfill \qedsymbol 
\begin{rem}\label{rem:cc:replacements depoissonization}
In addition to the replacements mentioned in Remarks~\ref{rem:cc:replacements in poissonization}~and~\ref{rem:cc:replacements pseudo-upper}, replacing
\begin{align*}
 \text{\Cref{lem:TV_concentration-McDiarmid-deterministic}}  &\leftarrow \text{\Cref{lem:cc:concentration of TV distance btw ccinduced and ccoutput}}
\end{align*}
recovers the proof in the constant-composition case.
\end{rem}
\section{Proof of the Dual Representations}\label{sec:alternative representations}
This section provides proofs for \eqref{eqn:thm:exact soft-cov exponent_ALTERNATIVE} and \eqref{eqn:thm:cc:exponent_ALTERNATIVE}, which are alternative representations of the exact soft-covering exponents in the random i.i.d. codebook and random constant-composition codebook cases, respectively. 
\subsection{Proof of the Dual Representation of $\alpha$}\label{sec:alternative representation of alpha}
\begin{prop} \label{prop:alternative representation}
	Given $P_X\to P_{Y|X} \to P_Y $, and $R> I(P_X, P_{Y|X})$ 
	\begin{align}
	& \min_{Q_{XY}} \bigg\{ D(Q_{XY} \| P_{XY} )  + \frac12 [R-D(Q_{XY}\|P_XQ_Y)]_{+}\bigg\} \nonumber \\
	&= \max_{\lambda \in [1,2]} \left\{\frac{\lambda-1}{\lambda} \left(R- I^\sfs_{\lambda}\left(P_X, P_{Y|X}\right)\right)\right\}  \text{.}
\end{align}
\end{prop} 
\begin{proof}
Note that
\begin{align}
 &\min_{Q_{XY}}\bigg\{D(Q_{XY}\|P_{XY}) + \frac12[R-D(Q_{XY}\|P_XQ_Y)]_+\bigg\} \nonumber \\ 
	 & = \min_{Q_Y}\min_{Q_{X|Y}} \max_{\lambda \in [0,1]} \bigg\{  D(Q_Y\|P_Y ) \label{eq:required for parizi vs hayashi} \\
	 &\qquad \qquad + D(Q_{X|Y}\|P_{X|Y}|Q_Y) \nonumber \\ 
	 &\qquad \qquad \qquad \qquad + \frac\lambda2\left(R-D(Q_{X|Y}\|P_X|Q_Y)\right)\bigg\} \nonumber   \\
	 &= \min_{Q_Y}\min_{Q_{X|Y}} \max_{\lambda \in [0,1]} \bigg\{ D(Q_Y\|P_Y)\label{tilde rv} \\
	 &\qquad \qquad  + \frac{2-\lambda}{2}D(Q_{X|Y} \| P_{X|Y}|Q_Y) \nonumber \\ 
	 &\qquad \qquad \qquad \qquad + \frac\lambda2 \left(R-\bbE[\imath^{}_{X;Y}(\widetilde X; \widetilde Y) ] \right) \bigg\} \nonumber \\
	 & = \min_{Q_Y}\max_{\lambda \in [0,1]} \min_{Q_{X|Y}}  \bigg\{D(Q_Y\|P_Y)\label{eqn:changed minmax to maxmin} \\
	 &\qquad \qquad + \frac{2-\lambda}{2}D(Q_{X|Y} \| P_{X|Y}|Q_Y) \nonumber \\ 
	 &\qquad \qquad \qquad \qquad + \frac\lambda2\left(R-\bbE[\imath^{}_{X;Y}(\widetilde X;\widetilde Y)]\right)\bigg\}  \nonumber \\
	 & = \min_{Q_Y}\max_{\lambda \in [0,1]}   \bigg\{ D(Q_Y\|P_Y) +\frac\lambda2 R  \\ 
	 & \qquad \qquad + \min_{Q_{X|Y}}\bigg\{\frac{2-\lambda}{2}D(Q_{X|Y}\|P_{X|Y}|Q_Y) \nonumber \\ 
	 & \qquad \qquad \qquad \qquad - \frac{\lambda}{2} \bbE[\imath^{}_{X;Y}(\widetilde X; \widetilde Y)]\bigg\}\bigg\} \nonumber \\
	 & = \min_{Q_Y}\max_{\lambda \in [0,1]}   \bigg\{D(Q_Y\|P_Y) +\frac\lambda2 R \label{from:cor2} \\ 
	 & -\frac{2-\lambda}{2} \bbE \left[\log\bbE\left[\exp\left(\frac{\lambda}{2-\lambda}\, \imath^{}_{X;Y}(\widehat X;\widetilde Y)\right)\middle|\widetilde Y\right]\right]\bigg\} \nonumber \\
	 & = \max_{\lambda \in [0,1]}  \min_{Q_Y} \bigg\{ D(Q_Y\|P_Y) +\frac\lambda2 R \label{from:minmax to maximin again} \\ 
	 & -\frac{2-\lambda}{2}\bbE\left[\log \bbE\left[\exp\left(\frac{\lambda}{2-\lambda}\, \imath^{}_{X;Y}(\widehat X ;\widetilde Y)\right)\middle|\widetilde Y\right]   \right]\bigg\} \nonumber \\
	 & = \max_{\lambda \in [0,1]} \bigg\{\frac\lambda2R+\min_{Q_Y} \bigg\{D(Q_Y\|P_Y) \\ 
	 & -\frac{2-\lambda}{2} \bbE \left[\log\bbE\left[\exp\left(\frac{\lambda}{2-\lambda}\, \imath^{}_{X;Y}(\widehat X ;\widetilde Y)\right)\middle| \widetilde Y\right]\right]\bigg\}\bigg\} \nonumber \\
	 & = \max_{\lambda \in [0,1]}   \bigg\{ \frac\lambda2 R \label{from:lem1} \\ 
	 & - \log \bbE \left[\bbE^\frac{2-\lambda}{2}\left[\exp \left(\frac{\lambda}{2-\lambda}\, \imath^{}_{X;Y}(X;Y) \right )\middle|Y\right ]\right]\bigg\} \nonumber \\
	 & = \max_{\lambda \in [0,1]}   \bigg\{\frac\lambda2 \left( R- I^\sfs_{\frac{2}{2-\lambda}}(P_X, P_{Y|X})  \right)  \bigg\} \text{,} \label{frm:change of meas for alph mutual}
\end{align}
where in \eqref{tilde rv} $(\widetilde X, \widetilde Y)\sim Q_{X|Y}Q_Y$ and the fact that
\begin{align}
	&D(Q_{X|Y}\|P_X|Q_Y) \nonumber \\  
	&\qquad = D(Q_{X|Y}\|P_{X|Y}|Q_Y) + \bbE[\imath^{}_{X;Y} (\widetilde X; \widetilde Y)] 
\end{align}
is used; in \eqref{eqn:changed minmax to maxmin} there is no duality gap in changing the minimax to maximin because the optimized quantity is convex in $Q_{X|Y}$ and linear in $\lambda$; in \eqref{from:cor2} $(\widehat X, \widetilde Y) \sim P_{X|Y}Q_Y$ and \Cref{cor:conditional minima} in \Cref{apdx:alternative representation and comparisons} is used; in \eqref{from:minmax to maximin again}, once again, there is no duailty gap in changing minimax to maximin because the optimized quantity is convex in $Q_Y$ while this time it is concave in $\lambda$ because the minimum of a collection of linear functions is concave; \eqref{from:lem1} is an application of \Cref{lem:rel ent and func minimization} in \Cref{apdx:alternative representation and comparisons} such that
\begin{align}
	&f(y) = \nonumber \\ 
	&\quad \frac{2-\lambda}{2} \log  \bbE \left[ \exp \left( \frac{\lambda}{2-\lambda}\, \imath^{}_{X;Y}(X;Y) \right) \middle|Y=y \right] \text{,}  \label{in:fy} 
\end{align}
with the random transformation from $\cY$ to $\cX$ in \eqref{in:fy} is fixed to be $P_{X|Y}$; and finally \eqref{frm:change of meas for alph mutual} follows from the definition of Sibson's proposal of $\alpha$-mutual information in \eqref{eqn:def:alpha mut info-semih}. 
\end{proof}
\subsection{Proof of the Dual Representation of $\aleph$} \label{sec:alternative representation of aleph}
\begin{prop}
  Given $P_{\bar X}\to P_{Y|X}\to P_Y $, $P_{\bar X} \to Q_{Y|X}\to Q_Y $, and $R>I(P_{\bar X}, P_{Y|X}) $
  \begin{align}
    &\min_{Q_{Y|X}} \bigg\{D(P_{\bar X}Q_{Y|X}\|P_{\bar XY}) \nonumber \\ 
    &\qquad \qquad \quad + \frac12\left[R-D(P_{\bar X}Q_{Y|X}\|P_{\bar X}Q_Y)\right]_+ \bigg\} \nonumber \\
%   &\qquad =  \max_{\lambda\in [1,2]}\max_{S_Y\in \cP(\cY)} \left\{\frac{\lambda-1}{\lambda} R -  \bbE\left[\log\bbE^{\frac1{\lambda}} \left[\exp\left( (\lambda -1)\, \imath^{}_{P_{\bar XY}\|P_{\bar X}S_Y}(\bar X, Y)\right)\middle|\bar X \right]\right]\right\}  \\
& =  \max_{\lambda\in [1,2]} \left\{ \frac{\lambda-1}{\lambda} \left(R - I^\sfc_\lambda (P_{\bar X}, P_{Y|X})\right)\right\} \text{.}
  \end{align}
\end{prop}

\begin{proof} Observe that
 \begin{align}
  &\min_{Q_{Y|X}} \bigg\{ D(P_{\bar X}Q_{Y|X}\|P_{\bar XY}) \nonumber \\ 
  &\qquad \qquad \quad + \frac12\left[R-D(P_{\bar X}Q_{Y|X}\|P_{\bar X}Q_Y)\right]_+ \bigg\} \nonumber \\ 
&= \min_{Q_{Y|X}}\max_{\lambda\in [0,1]}\bigg\{ D(P_{\bar X}Q_{Y|X}\|P_{\bar XY}) \nonumber \\ 
&\qquad \qquad \quad  + \frac\lambda2(R-D(P_{\bar X}Q_{Y|X}\|P_{\bar X}Q_Y))\bigg\} \label{frm:cc:optimized func here}  \\
&= \min_{Q_{Y|X}}\max_{\lambda\in [0,1]} \bigg\{ \frac\lambda2R + \left( 1-\frac\lambda2 \right) D(P_{\bar X}Q_{Y|X}\|P_{\bar XY}) \nonumber \\ 
&\qquad \qquad \  - \frac\lambda2 H(Q_{Y}) + \frac\lambda2  \bbE\left[ \imath^{}_{P_{Y|X}} (\widetilde Y|\bar X )\right]   \bigg\} \label{frm:cc:algebra}   \\
&= \max_{\lambda\in [0,1]} \min_{Q_{Y|X}} \bigg\{ \frac\lambda2R + \left( 1-\frac\lambda2 \right) D(P_{\bar X}Q_{Y|X}\|P_{\bar XY}) \nonumber \\ 
&\qquad \qquad \ - \frac\lambda2 H(Q_{Y}) + \frac\lambda2  \bbE\left[ \imath^{}_{P_{Y|X}} (\widetilde Y|\bar X )\right]   \bigg\} \label{frm:cc:no duality gap}  \\
&= \max_{\lambda\in [0,1]} \min_{Q_{Y|X}} \max_{S_Y}\bigg\{\frac\lambda2R \nonumber \\ 
&\qquad  + \left( 1-\frac\lambda2 \right) D(P_{\bar X}Q_{Y|X}\|P_{\bar XY}) \nonumber \\
&\qquad \quad  - \frac\lambda2 \bbE\left[\imath^{}_{S_Y}(\widetilde Y)\right] + \frac\lambda2  \bbE\left[\imath^{}_{P_{Y|X}} (\widetilde Y|\bar X )  \right] \bigg\} \label{frm:cc:var rep of entropy}  \\
& = \max_{\lambda\in [0,1]}\max_{S_Y}  \bigg\{ \frac\lambda2R \nonumber \\ 
&\quad + \min_{Q_{Y|X}}\bigg\{ \left( 1-\frac\lambda2 \right) D(P_{\bar X}Q_{Y|X}\|P_{\bar XY}) \nonumber \\
&\qquad \,  - \frac\lambda2 \bbE\left[\imath^{}_{S_Y}(\widetilde Y)\right]   + \frac\lambda2  \bbE\left[ \imath^{}_{P_{Y|X}} (\widetilde Y|\bar X )\right] \bigg\}\bigg\} \label{frm:cc:no duality gap2} \\
&= \max_{\lambda\in [0,1]}\max_{S_Y}  \bigg\{ \frac\lambda2  R \nonumber \\ 
&\qquad \qquad \quad \  - \frac\lambda2 \bbE\left[D_{\frac{2}{2-\lambda}}(P_{Y|X}(\cdot|\bar X ) \| S_Y)  \right]  \bigg\} \label{frm:cc:corollary for aleph dual}  \\
&=  \max_{\lambda\in [0,1]} \left\{ \frac \lambda2 \left( R - I^\sfc_{\frac{2}{2-\lambda}}(P_{\bar X}, P_{Y|X} ) \right) \right\}   \text{,} \label{frm:def:csiszar alpha mut inf}
\end{align}
where in \eqref{frm:cc:algebra} $( \widetilde Y,\bar X )\sim Q_{Y|X}P_{\bar X}$ and the fact that 
\begin{align}
  &D(P_{\bar X}Q_{Y|X}\|P_{\bar X}Q_Y) \nonumber \\
  &=  D(P_{\bar X}Q_{Y|X}\|P_{\bar XY}) + H(Q_Y) -  \bbE\left[ \imath^{}_{P_{Y|X}} (\widetilde Y|\bar X )\right]
\end{align}
is used; in \eqref{frm:cc:no duality gap} there is no duality gap as the optimized quantity is linear in $\lambda $ and convex in $Q_{Y|X}$; \eqref{frm:cc:var rep of entropy} follows from the variational representation of entropy: 
\begin{align}
  H(Q_{Y}) = \min_{S_Y} \bbE\left[\imath^{}_{S_Y}(\widetilde Y )\right] \text{;}
\end{align}
in \eqref{frm:cc:no duality gap2} there is no duality gap as the optimized quantity is convex in $Q_{Y|X} $ and concave in $S_Y$; in \eqref{frm:cc:corollary for aleph dual} $(\bar X, Y) \sim P_{\bar X}P_{Y|X} = P_{\bar XY} $ and we use \Cref{cor:for aleph dual}; finally \eqref{frm:def:csiszar alpha mut inf} follows from the definition of Csisz\'ar's proposal of $\alpha $-mutual information in \eqref{eqn:def:cc:csiszar alpha mut info}. 
\end{proof}

\section{Comparisons with the Known Lower Bounds on the Soft-Covering Exponent}\label{sec:comparisons}
This section compares the exact soft-covering exponents in \Cref{thm:main,thm:cc:main} to their previously known lower bounds. In particular, \Cref{sec:comparison IID Random Codebook} provides comparisons of $\alpha(R, P_X,P_{Y|X})$ with the exponents that can be found in \cite[Lemma VII.9]{cuff2013distributed} and \cite[Theorem 6]{hayashi2006general}, and with the half of the relative entropy variant of the soft-covering exponent that can be found in \cite[Theorem 4(i)]{parizi2017exact} while \Cref{sec:cc:comparison Constant Composition} compares $\aleph(R, P_{\bar X}, P_{Y|X}) $ with the half of the relative entropy variant of the soft-covering exponents that can be found in \cite[Theorem 4(ii)]{parizi2017exact}, \cite[Theorem 10]{hayashi2011universally}, and  \cite[Eq. (177)]{hayashi-matsumoto-2016}.
\subsection{Comparisons in the Random i.i.d. Codebook Case}\label{sec:comparison IID Random Codebook}
Prior to our result in \Cref{thm:main}, the best known-to-date lower bound on the soft-covering exponent was provided in \cite[Lemma VII.9]{cuff2013distributed} which was shown to be
\begin{align}
	&\beta(R, P_X, P_{Y|X}) \nonumber \\ 
	& = \max_{\lambda \ge 0 } \max_{\lambda' \le 1} \bigg\{  \frac{\lambda}{2\lambda+1-\lambda'} \Big(R \label{eqn:def:beta1} \\ 
	&\qquad \qquad    - (1-\lambda')D_{1+\lambda}(P_{XY}\| P_X P_Y ) \nonumber \\ 
	&\qquad \qquad \qquad \qquad \quad  -\lambda' \widetilde D_{1+\lambda'} (P_{XY} \| P_X P_Y ) \Big)\bigg\}  \text{,} \nonumber
\end{align}
where, supposing $(X,Y)\sim P_X P_{Y|X}$,
\begin{align}
	D_{1+\lambda}(P_{XY}\| P_X P_Y ) &= \frac{1}{\lambda } \log \bbE \left[ \exp( \lambda\, \imath^{}_{X;Y}(X;Y) ) \right]  \label{eqn:de:renyi_mutual_info} 
\end{align}
is the R\'enyi divergence (see, e.g., \cite{vanErven2014}) of order $1+\lambda$ between the joint and product distributions, and 
\begin{align}
	&\widetilde D_{1+\lambda'} (P_{XY} \| P_X P_Y) \nonumber \\ 
	 &\quad \ = \frac{2}{\lambda'} \log \bbE\left[\bbE^\frac12\left[\exp\left(\lambda' \, \imath^{}_{X;Y}(X;Y)\right )\middle| Y \right]\right] \text{.}  \label{eqn:def:almost_renyi_mutual_info}
\end{align}
Using the results provided in \Cref{apdx:alternative representation and comparisons}, \Cref{prop:comparison with cuff distributed channel synthesis} proves the fact that $\alpha(R, P_X, P_{Y|X})$ captures the exponential decay rate in soft-covering lemma better than $\beta(R, P_X, P_{Y|X})$. 
\begin{prop}\label{prop:comparison with cuff distributed channel synthesis}
	 Suppose $P_X \to P_{Y|X} \to P_Y $, and $R>I(P_X, P_{Y|X})>0$. Then
	\begin{align}
		\alpha(R, P_X, P_{Y|X}) \ge \beta(R, P_X, P_{Y|X}) \text{,}
	\end{align}
	where $\alpha(R, P_X, P_{Y|X})$ and $ \beta(R, P_X, P_{Y|X})$ are as defined in \eqref{eqn:def:alpha(R,P_X,P_{Y|X})} and \eqref{eqn:def:beta1}, respectively. 
\end{prop}
\begin{proof}
	Let $(X,Y)\sim P_XP_{Y|X}  $, $(\widetilde X, \widetilde Y)\sim Q_{X|Y}Q_Y $, and $(\widehat X, \widehat Y)\sim S_{X|Y}S_Y$. It follows that
	\begin{align}
&\beta(R, P_X, P_{Y|X}) \nonumber \\ 
&= \max_{\substack{ \lambda \ge 0\\ \lambda' \le 1 }} \bigg\{ \frac{\lambda}{2\lambda+1-\lambda'} \Big(R \\
& \qquad \qquad  -(1-\lambda')D_{1+\lambda}(P_{XY}\| P_X P_Y ) \nonumber \\
& \qquad \qquad  \qquad \qquad \qquad -\lambda' \widetilde D_{1+\lambda'} (P_{XY} \| P_X P_Y ) \Big)\bigg\} \nonumber \\
&= \max_{\substack{ \lambda \ge 0\\ \lambda' \le 1 }} \bigg\{ \frac{\lambda}{2\lambda+1-\lambda'} \bigg(R \label{eqn:use lemmas} \\
&\  +\frac{1-\lambda'}{\lambda} \min_{Q_{XY}} \left\{  D(Q_{XY} \| P_{XY}) - \lambda \bbE \big[ \imath^{}_{X;Y}(\widetilde X; \widetilde Y)\big]\right\}   \nonumber \\ 
&\qquad + \min_{S_{XY} } \Big\{ 2 D(S_{XY} \| P_{XY} ) -  D(S_{XY}\|P_{X|Y}S_Y) \nonumber \\ 
&\qquad \qquad \qquad \qquad \qquad \qquad   - \lambda' \bbE \big[ \imath^{}_{X;Y}(\widehat X; \widehat Y ) \big]  \Big\}   \bigg) \bigg\} \nonumber \\
& \le  \max_{\substack{ \lambda \ge 0\\ \lambda' \le 1 }} \min_{Q_{XY}} \bigg\{ D(Q_{XY}\|P_{XY}) \label{from:sub opt. chch} \\ 
&\qquad \qquad \quad \, + \frac{\lambda}{2\lambda +1 - \lambda' } \left( R- D(Q_{XY}\|P_XQ_Y)\right)  \bigg\} \nonumber \\
& \le  \min_{Q_{XY}} \max_{\substack{ \lambda \ge 0\\ \lambda' \le 1 }} \bigg\{ D(Q_{XY}\|P_{XY}) \label{from:minmax > maxmin} \\ 
&\qquad \qquad \quad \, + \frac{\lambda}{2\lambda +1 - \lambda' } \left( R- D(Q_{XY}\|P_XQ_Y)  \right)  \bigg\} \nonumber    \\
&= \min_{Q_{XY}} \bigg\{ D(Q_{XY}\|P_{XY})  \label{from:take maxima} \\ 
&\qquad \qquad \qquad \qquad \quad \,  +  \frac12 \left[ R- D(Q_{XY}\|P_XQ_Y)  \right]_+  \bigg\}  \nonumber \\
	&= \alpha(R, P_X, P_{Y|X}) \text{,} 
	\end{align}
where \eqref{eqn:use lemmas} uses \Cref{cor:lemma for renyi mutual info,cor:lemma for almost renyi mutual info} in \Cref{apdx:alternative representation and comparisons}; \eqref{from:sub opt. chch} constrains the two minimizations by assuming that their minimizers are equivalent and uses the fact that
\begin{align}
&D(Q_{XY}\|P_{X|Y}Q_Y) + \bbE[\imath^{}_{X;Y}(\widetilde X;\widetilde Y)] \nonumber \\ 
&\qquad = D(Q_{XY}\|P_XQ_Y)  \text{;} \label{use:parizi comparison}
\end{align}
 \eqref{from:minmax > maxmin} is due to the duality gap; and finally \eqref{from:take maxima} follows because $\frac{\lambda a}{2\lambda +1 -\lambda'} $ is monotone decreasing or increasing in\footnote{Same observation holds if one focuses on $\lambda$ instead of $\lambda'$.} $\lambda'  $ depending on whether $a<0 $ or $a>0$. 
\end{proof}
Prior to Cuff's exponent in \cite[Lemma VII.9]{cuff2013distributed}, Hayashi \cite[Theorem 6]{hayashi2006general} argues that 
\begin{align}
		& \liminf_{n\to \infty} -\frac{1}{n}\log \bbE \left[\left\| \induced - P_{Y^n}\right\|^{}_1\right] \nonumber \\ 
		&\qquad \qquad \ge \gamma(R, P_X, P_{Y|X}) \text{,}
	\end{align}
where 	
\begin{align}
	&\gamma(R, P_X, P_{Y|X}) \nonumber \\ 
	&\qquad  = \max_{\lambda \in [0,1] } \left\{\frac{\lambda}{1+\lambda} \left( R- D_{1+\lambda}(P_{XY}\| P_{X}P_{Y} ) \right)\right\} \text{.} \label{eqn:def:gamma}
\end{align}
As shown in \cite{cuff2013distributed}, thanks to Jensen's inequality, noting that
\begin{align}
	\widetilde D_{1+\lambda} (P_{XY} \| P_X P_Y) \le  D_{1+\lambda} (P_{XY} \| P_X P_Y) \text{,}
\end{align}
and altering the maximization domain in the right side of \eqref{eqn:def:beta1} by restricting $\lambda'=\lambda$ yields
\begin{align}
	\beta(R, P_X, P_{Y|X}) \ge \gamma(R, P_X, P_{Y|X}) \text{.} \label{cuff > hayashi}
\end{align}
Together with \Cref{prop:comparison with cuff distributed channel synthesis}, \eqref{cuff > hayashi} implies
\begin{align}
	\alpha(R, P_X, P_{Y|X}) \ge \gamma(R, P_X, P_{Y|X}) \label{kerim hayashi} \text{.}
\end{align} 

As a further comparison, Parizi \emph{et al.} \cite[Theorem 4(i)]{parizi2017exact} show that\footnote{Previously, Hayashi argues the lower bound in \eqref{eq:par-tel-mer and hay-mats} without showing the primal form of $\zeta$ in \eqref{eqn:def:zeta_def}, see \cite{hayashi2011exponential}.}
\begin{align}
	&\lim_{n\to \infty } -\frac1n \log \bbE\left[D\left(\induced \middle\|P_{Y^n}\right)\right] \nonumber \\ 
	&\qquad \qquad = \zeta(R, P_X, P_{Y|X}) \text{,} \label{eq:par-tel-mer and hay-mats}
\end{align}
where
\begin{align}
&\zeta(R, P_X, P_{Y|X}) \nonumber \\ 
&= \min _{Q_{XY}} \left\{ D(Q_{XY} \| P_{XY} )+\left[R - \bbE\big[\imath^{}_{X;Y}(\widetilde X;\widetilde Y)\big] \right]_+ \right\} \label{eqn:def:zeta_def} \\
&= \max_{\lambda\in [0,1]} \lambda \left( R - D_{1+\lambda}(P_{XY} \| P_X P_Y ) \right) \label{eqn:def:zeta_ALTERNATIVE} 
\end{align}
with $(\widetilde X, \widetilde Y) \sim Q_{X|Y}Q_Y $. Using Pinsker's \cite[Problem 3.18]{csiszar2011information} and Jensen's inequalities 
\begin{align}
	&\bbE\big[D\big(\induced\big\|P_{Y^n}\big)\big] \nonumber \\ 
	&\qquad \ge \frac{\log \rme}{2} \bbE\left[\left\|\induced-P_{Y^n}\right\|_1^2\right] \\
	&\qquad \ge \frac{\log \rme}{2} \bbE^2\left[\left\|\induced-P_{Y^n}\right\|_1\right] \text{,}
\end{align}
and one can easily see the following lower bound on the soft-covering exponent
\begin{align}
	&\liminf_{n \to \infty} -\frac 1n\log \bbE \left[\left\| \induced - P_{Y^n} \right\|^{}_1\right] \nonumber \\ 
	&\qquad \qquad  \ge \frac12\zeta(R, P_X, P_{Y|X}) \text{.}
\end{align}
From the the definition of $\gamma(R, P_X, P_{Y|X})$ in \eqref{eqn:def:gamma} and the dual form of $\zeta(R, P_X, P_{Y|X}) $ in \eqref{eqn:def:zeta_ALTERNATIVE}, it is immediate that 
\begin{align}
	\gamma(R, P_X, P_{Y|X}) \ge \frac12\zeta(R, P_X, P_{Y|X}) \label{half parizi is even less} \text{,}
\end{align}
which, together with the bound in \eqref{kerim hayashi}, implies
\begin{align}
	\alpha(R, P_X, P_{Y|X}) \ge \frac12\zeta(R, P_X, P_{Y|X}) \text{.}
\end{align}

Following example illustrates the fact that, in general, there is a strictly positive gap between the above compared exponents.
\begin{exa}[Binary Symmetric Channel]\label{exa:bsc alpha vs beta}
Consider the setting in \Cref{exa:bsc alpha vs aleph}, where $P_{Y|X}\colon \cX\to\cY$ is a binary symmetric channel with crossover probability $p=0.05$, and $P_X(0) = 2/5$. If $R = 0.85 > I(P_X, P_{Y|X}) \approx 0.69$ \texttt{\emph{bits}},
  \begin{align}
  \alpha(0.85, P_X, P_{Y|X})   &\approx 2.0429 \times 10^{-2} \text{,} \\
   \beta(0.85, P_X, P_{Y|X})   &\approx 2.0331 \times 10^{-2} \text{,} \\
   \gamma (0.85, P_X, P_{Y|X}) &\approx 2.0116  \times 10^{-2} \text{,} \\
  0.5\times \zeta(0.85, P_X, P_{Y|X}) &\approx 1.3767 \times 10^{-2} \text{,}
  \end{align}
implying $\alpha > \beta > \gamma > \frac12 \zeta$, in general. 

On the next page, \Cref{fig:comparison} shows the computed $\alpha $, $\beta$, $\gamma $ and $\frac12 \zeta$ values for various rates $R$. Note that, although \Cref{fig:comparison:bigpic} shows that $\alpha $, $\beta$, and $\gamma $, are almost equal to one another for a range of $R$ values, there exists a small but strictly positive gap between them, see, e.g., \Cref{fig:comparison.zoom}. 
\end{exa}
%%%%
%%%%
\input{Fig_Comparison-of-Exponents} %FIGURE COMPARISON OF EXPONENTS
%%%%%
%%%%%
\subsection{Comparisons in the Random Constant-Composition Codebook Case}\label{sec:cc:comparison Constant Composition}
When the constant-composition coding ensemble $\scrD_M^n$ is used instead of the i.i.d. coding ensemble $\scrC_M^n$, Parizi \emph{et al.} \cite[Theorem 4(ii)]{parizi2017exact} show that, 
\begin{align}
	&\lim_{n\to \infty } -\frac1n \log \bbE\left[D\left(\ccinduced \middle\|R_{\breve Y^n}\right)\right] \nonumber \\  
	&\qquad \qquad  = \beth(R, P_{\bar X}, P_{Y|X}) \text{,}
\end{align}
such that
\begin{align}
  &\beth(R, P_{\bar X}, P_{Y|X}) \nonumber \\ 
  &\quad  =  \min_{Q_{Y|X} \in \cP(\cY|\cX)} \big\{ D(P_{\bar X}Q_{Y|X}\|P_{\bar XY}) \label{eqn:cc:def:beth} \\ 
  & \qquad \qquad \qquad \qquad  + [R - G(Q_{Y|X}\|P_{Y|X}|P_{\bar X}) ]_+  \big\} \text{,}  \nonumber
\end{align}
where $P_{\bar XY} = P_{\bar X}P_{Y|X} $, and for $P_{\bar X}\to Q_{Y|X}\to Q_Y $, assuming $(\bar X, \widetilde Y)\sim P_{\bar X} Q_{Y|X} $,
\begin{align}
&G(Q_{Y|X}\|P_{Y|X}|P_{\bar X})= H(Q_Y)- \bbE\big[ \imath^{}_{P_{Y|X}} (\widetilde Y|\bar X)\big] \nonumber \\ 
&\quad \qquad + \min_{\substack{R_{Y|X} \colon \\ P_{\bar X} \to R_{Y|X} \to Q_Y }}D(P_{\bar X}R_{Y|X}\|P_{\bar XY}) \text{.} \label{cf:G func}  
\end{align}
Once again, using Pinsker's \cite[Problem 3.18]{csiszar2011information} and Jensen's inequalities 
\begin{align}
	&\bbE\Big[D\Big(\ccinduced \Big\|R_{\breve Y^n}\Big)\Big] \nonumber \\ 
	&\qquad \ge \frac{\log \rme}{2} \bbE\bigg[\Big\|\ccinduced-R_{\breve Y^n}\Big\|_1^2\bigg]   \\
	&\qquad \ge \frac{\log \rme}{2} \bbE^2\Big[\Big\|\ccinduced-R_{\breve Y^n}\Big\|_1\Big] \text{,} 
\end{align}
one can easily see the following lower bound on the soft-covering exponent in the constant-composition case:
\begin{align}
	&\liminf_{n \to \infty} -\frac 1n\log \bbE \left[\left\| \ccinduced - R_{\breve Y^n} \right\|^{}_1\right] \nonumber \\ 
	&\qquad \qquad \ge \frac12\beth(R, P_X, P_{Y|X}) \text{.}
\end{align}
Since $\aleph(R, P_{\bar X}, P_{Y|X}) $ is the exact soft-covering exponent in the constant-composition case, it is expected that $\aleph \ge \frac12\beth $. This result is formally established by \Cref{prop:aleph vs beth}.  
\begin{prop}\label{prop:aleph vs beth}
  Given an $m$-type $P_{\bar X}\in \cP_m(\cX)$ suppose $P_{\bar X}\to P_{Y|X} \to P_Y$, and $R>I(P_{\bar X}, P_{Y|X})>0$. Then,
  \begin{align}
    \aleph(R,P_{\bar X}, P_{Y|X}) \ge \frac12\beth(R,P_{\bar X}, P_{Y|X}) \text{,}
  \end{align}
  where $\aleph(R,P_{\bar X}, P_{Y|X})$  and $\beth(R,P_{\bar X}, P_{Y|X})$ are as defined in \eqref{eqn:cc:def:aleph} and \eqref{eqn:cc:def:beth}, respectively. 
\end{prop}
\begin{proof}
Assume $(\bar X, \widetilde Y)\sim P_{\bar X}Q_{Y|X} $. Realizing (cf. \eqref{cf:G func})
\begin{align}
 G(Q_{Y|X}\|P_{Y|X}|P_{\bar X}) \ge  H(Q_Y) - \bbE\big[ \imath^{}_{P_{Y|X}} (\widetilde Y|\bar X)\big] \text{,}
\end{align}
note that
\begin{align}
  &\aleph(R, P_{\bar X}, P_{Y|X}) \nonumber \\ 
  &= \min_{Q_{Y|X} \in \cP(\cY|\cX)} \bigg\{ D(P_{\bar X}Q_{Y|X}\|P_{\bar XY}) \\ 
  &\qquad \qquad \qquad \quad \ \, + \frac12[R-D(P_{\bar X}Q_{Y|X}\|P_{\bar X}Q_Y)]_+ \bigg\} \nonumber \\
  &\ge  \min_{Q_{Y|X} \in \cP(\cY|\cX)} \bigg\{ \frac12 D(P_{\bar X}Q_{Y|X}\|P_{\bar XY})  \label{frm:max(0,f-a)>max(0,f)-a}  \\ 
  &\qquad \qquad \quad \,  + \frac12   \big[R + \bbE\big[ \imath^{}_{P_{Y|X}} (\widetilde Y|\bar X)\big] - H(Q_Y)\big]_+ \bigg \}  \nonumber \\
   &\ge  \min_{Q_{Y|X} \in \cP(\cY|\cX)} \bigg\{\frac12 D(P_{\bar X}Q_{Y|X}\|P_{\bar XY}) \label{frm:G def in par tel} \\ 
   &\qquad \qquad \qquad \qquad   + \frac12[R - G(Q_{Y|X}\|P_{Y|X}|P_X) ]_+  \bigg\}  \nonumber \\
   &= \frac12 \beth (R, P_{\bar X}, P_{Y|X} )
  \end{align}
where \eqref{frm:max(0,f-a)>max(0,f)-a} follows from the facts that
\begin{align}
 & D(P_{\bar X}Q_{Y|X}\|P_{\bar X}Q_Y) \nonumber \\ 
 & =  D(P_{\bar X}Q_{Y|X}\|P_{\bar XY}) + H(Q_Y) - \bbE\big[ \imath^{}_{P_{Y|X}} (\widetilde Y|\bar X)\big] \label{eq:req for G<mut inf}
\end{align}
and $[f-h]_+ \ge [f]_+-h $ for any non-negative $h$.
\end{proof}

Apart from the exponent shown in \cite[Theorem 4(ii)]{parizi2017exact}, Hayashi and Matsumoto \cite[Theorem 10]{hayashi2011universally} discuss that 
\begin{align}
	&\liminf_{n\to \infty } -\frac1n \log \bbE\left[D\left(\ccinduced \middle\|R_{\breve Y^n}\right)\right] \nonumber \\ 
	&\qquad \qquad \ge \daleth (R, P_{\bar X}, P_{Y|X}) \text{,}
\end{align}
where
\begin{align}
  &\daleth (R, P_{\bar X}, P_{Y|X}) \nonumber \\ 
  &\qquad \quad   = \max_{\lambda\in [0,1]} \left\{\lambda \left(R -I^\sfs_{\frac{1}{1-\lambda}}(P_{\bar X}, P_{Y|X})\right)\right\} \text{.} \label{eqn:cc:def:dalethCC}
\end{align}
Using \eqref{eq:required for parizi vs hayashi}--\eqref{frm:change of meas for alph mutual} and the fact that (cf. \eqref{cf:G func} and \eqref{eq:req for G<mut inf})
\begin{align}
  G(Q_{Y|X}\|P_{Y|X}|P_{\bar X}) \le D(P_{\bar X}Q_{Y|X}\|P_{\bar X}Q_Y) \text{,}
\end{align}
it is easy to establish\footnote{Also see \cite[Appendix C]{parizi2017exact} for a different (and more complex) proof of \eqref{bc:G<mut inf}.}
\begin{align}
 &\daleth (R, P_{\bar X}, P_{Y|X}) \nonumber \\ 
 &=\max_{\lambda \in [0,1]}   \left\{\lambda\left( R- I^\sfs_{\frac{1}{1-\lambda}}(P_X, P_{Y|X}) \right)\right\} \\
  & = \min_{Q_{XY}} \big\{ D(Q_{XY}\|P_{\bar XY}) \\ 
  &\qquad \qquad \qquad \qquad \quad  + \left[R-D(Q_{XY}\|P_{\bar X}Q_Y)\right]_+\big\} \nonumber \\
  &\le \min_{Q_{Y|X}} \Big\{ D(P_{\bar X}Q_{Y|X}\|P_{\bar XY}) \label{bc:suboptimal Q_X=P_X} \\ 
  &\qquad \qquad \qquad \quad + \left[R-D(P_{\bar X}Q_{Y|X}\|P_{\bar X}Q_Y)\right]_+\Big\} \nonumber \\
  &\le  \beth(R,P_{\bar X},P_{Y|X}) \text{,} \label{bc:G<mut inf}
\end{align}
where \eqref{bc:suboptimal Q_X=P_X} follows from the suboptimal choice $Q_X=P_{\bar X}$. Together with \Cref{prop:aleph vs beth}, \eqref{bc:G<mut inf} readily implies that 
\begin{align}
   \aleph(R,P_{\bar X}, P_{Y|X}) \ge \frac12\daleth(R,P_{\bar X}, P_{Y|X}) \text{.} \label{cc:proven here aga}
\end{align}
Furthermore, in a different paper, Hayashi and Matsumoto \cite[Eq. (177)]{hayashi-matsumoto-2016} also argue that 
\begin{align}
  	&\liminf_{n\to \infty } -\frac1n \log \bbE\left[D\left(\ccinduced \middle\|R_{\breve Y^n}\right)\right] \nonumber \\ 
  	&\qquad \qquad \ge \gimel (R, P_{\bar X}, P_{Y|X}) \text{,}
\end{align}
where 
\begin{align}
  &\gimel(R, P_{\bar X}, P_{Y|X}) \nonumber \\ 
  & = \min_{Q_{Y|X} } \Big\{ D(P_{\bar X}Q_{Y|X}\|P_{\bar XY}) \\ 
  &\qquad \qquad \qquad \quad  + \left[R-D(P_{\bar X}Q_{Y|X}\|P_{\bar X}Q_Y)\right]_+ \Big\} \text{.} \nonumber 
\end{align}
Though, since\footnote{Also see \eqref{bc:suboptimal Q_X=P_X}--\eqref{bc:G<mut inf} together with \Cref{prop:aleph vs beth}.} $D(P_{\bar X}Q_{Y|X}\|P_{\bar XY}) \ge 0$, it is trivial to see that $\aleph\ge \frac12 \gimel $ in this case. 

Similar to its counterpart in \Cref{exa:bsc alpha vs beta}, the following example illustrates the fact that, in general, there is a strictly positive gap between $\aleph$, $\beth$, $\daleth$, and $\gimel$.
\begin{exa}[Binary Symmetric Channel] Consider the setting in Examples~\ref{exa:bsc alpha vs aleph}~and~\ref{exa:bsc alpha vs beta} , where $P_{Y|X}\colon \cX\to \cY$ is a binary symmetric channel with crossover probability $p=0.05$, and $P_{\bar X}(0) = 2/5$. If $R=0.85>I(P_{\bar X}, P_{Y|X})\approx 0.69 $ \texttt{\emph{bits}}
\begin{align}
  \aleph(0.85,P_{\bar X},P_{Y|X}) &\approx 2.21595 \times 10^{-2}\text{,} \\
\frac12\beth(0.85,P_{\bar X},P_{Y|X})  &\approx 1.60663 \times 10^{-2} \text{,} \\
\frac12\gimel(0.85,P_{\bar X},P_{Y|X}) &\approx 1.10797\times 10^{-2} \text{,} \\
\frac12\daleth(0.85,P_{\bar X},P_{Y|X}) &\approx 1.02143 \times 10^{-2} \text{,} 
\end{align}
implying $\aleph>\frac12\beth>\frac12\gimel>\frac12\daleth $, in general. 

\Cref{fig:cc:aleph vs all other cc} illustrates the computed $\aleph $, $\frac12\beth$, $\frac12\gimel $, and $\frac12 \daleth $ values for various rates $R$. 
\input{Fig_CC_Comparison-0f-Exponents}
\end{exa}
\begin{appendices}\crefalias{section}{apdx}
\section{Proofs of Lemmas~\ref{lem:TV_concentration-McDiarmid-deterministic}~and~\ref{lem:cc:concentration of TV distance btw ccinduced and ccoutput}}\label{apdx:proof of TV concentration McDiarmid}
This section provides the proofs of \Cref{lem:TV_concentration-McDiarmid-deterministic,lem:cc:concentration of TV distance btw ccinduced and ccoutput} that are presented in \Cref{sec:main}. The simple proof of \Cref{lem:TV_concentration-McDiarmid-deterministic}, which can be found in \cite[Theorem 31]{lcv2017egamma} and \cite[Lemma 2]{tahmasbi2017second}, is repeated in the first part of this appendix whereas the proof of \Cref{lem:cc:concentration of TV distance btw ccinduced and ccoutput}, which follows the footsteps of that of \Cref{lem:TV_concentration-McDiarmid-deterministic}, is contained in the second part.
\subsection{Proof of Lemma~\ref{lem:TV_concentration-McDiarmid-deterministic}}\label{apdxsub:TV concentration McDiarmid}
Define the variation of a function $f\colon \cX^m \to \bbR $ at coordinate $i$ as
\begin{align}
	 &d_i(f(x^m))  \nonumber \\ 
	&\quad \,  =\sup_{z,z'} \big|f(x_1, \ldots, x_{i-1}, z ,x_{i+1}, \ldots, x_{m}) \label{def:dic_der} \\
	&\qquad \qquad \quad -f(x_1, \ldots, x_{i-1}, z' ,x_{i+1}, \ldots, x_{m})\big| \text{,}  \nonumber
\end{align}
and observe that
\begin{align}
	\left\|\induced - P_{Y^n}\right\|^{}_1 = f\left(X_1^n, \ldots, X_M^n \right) \text{,}
\end{align}
where for the given discrete memoryless channel, $P_{Y^n|X^n}$, the function $f \colon (\cX^n)^M \to \bbR $ is defined as
\begin{align}
	&f\left(X_1^n, \ldots, X_M^n \right) \nonumber  \\ 
	&\quad = \sum_{y^n\in \cY^n} \left|\frac{1}{M}\sum_{j=1}^{M} P_{Y^n|X^n}(y^n|X_j^n) - P_{Y^n}(y^n) \right| \text{.}
\end{align}
Since for any $i\in \{1, \ldots, M\}$
\begin{align}
&\sum_{y^n} \left|\frac{1}{M}\sum_{j \neq i} P_{Y^n|X^n}(y^n|X_j^n) - P_{Y^n}(y^n)\right|
 - \frac{1}{M} \nonumber \\  
 &\le \sum_{y^n} \left|\frac{1}{M}\sum_{j=1}^{M} P_{Y^n|X^n}(y^n|X_j^n) - P_{Y^n}(y^n)\right| \\
 &\le \sum_{y^n} \left|\frac{1}{M}\sum_{j \neq i} P_{Y^n|X^n}(y^n|X_j^n) - P_{Y^n}(y^n)\right|
 + \frac{1}{M} \text{,}
\end{align}
it follows that, for any $i\in \{1, \ldots, M\}$, 
\begin{align}
	d_i\left(\left\|\induced - P_{Y^n}\right\|^{}_1\right) \le \frac{2}{M} \text{.}
\end{align}
Finally, the desired result follows from the McDiarmid's inequality, see, e.g., \cite[Theorem 2.2.3]{raginsky2013concentration}.
\hfill \qedsymbol 
\subsection{Proof of Lemma~\ref{lem:cc:concentration of TV distance btw ccinduced and ccoutput}}\label{apdxsub:cc:TV_mcdiarmid}
  Following the footsteps of the proof of \Cref{lem:TV_concentration-McDiarmid-deterministic}, observe that 
\begin{align}
	d_i\left(\left\|\ccinduced - R_{\breve Y^n}\right\|^{}_1\right) \le \frac{2}{M} \text{.}
\end{align}
Hence, once again, McDiarmid's inequality \cite[Theorem 2.2.3]{raginsky2013concentration} yields the desired result.
\hfill \qedsymbol 
\section{Preliminary Lemmas for the Proofs of Theorems~\ref{thm:main}~and~\ref{thm:cc:main}}\label{apdx:main}
This section provides several non-asymptotic results that are used in the proof of \Cref{thm:main}.  

\begin{lemma}\label{lem:absolute mean deviation upper bound for Z_{Q_{barX|barY}}}
	Given $y^n \in \cT^n_{Q_{\bar Y}}$, and $Q_{\bar X|\bar Y}\in \cP_n(\cX|Q_{\bar Y})$, let $Z_{Q_{\bar X\bar Y}}$ be the random variable as defined in \eqref{eqn:def:Z_Q_{barX|barY}}, and let $P_{Y^n}$ be the i.i.d. output distribution. Then,
	\begin{align}
      & P_{Y^n}(y^n)\bbE \left[\left|Z_{Q_{\bar X\bar Y}} -  \bbE[Z_{Q_{\bar X\bar Y}}]\right|\right] \nonumber \\ 
      &\qquad \le P_{Y^n|X^n}(y^n|x^n_{Q_{\bar X|\bar Y}}) \mfrY(M, Q_{\bar X\bar Y}) \label{eqn:lem:to make it sim} \\
      &\qquad =  \exp(-n\bbE[\imath^{}_{P_{Y|X}}(\bar Y|\bar X)]) \mfrY(M, Q_{\bar X\bar Y}) \text{,} \label{eqn:lem:to make it sim 2-after revision}
	\end{align}
	 where in \eqref{eqn:lem:to make it sim} $x^n_{Q_{\bar X|\bar Y}}$ represents an element from the conditional type class $\cT^n_{Q_{\bar X|\bar Y}}(y^n) $, and $\mfrY(M, Q_{\bar X \bar Y})$ is as defined in \eqref{eqn:def:mfrakY}; in \eqref{eqn:lem:to make it sim 2-after revision} $(\bar X, \bar Y)\sim Q_{\bar X |\bar Y} Q_{\bar Y} = Q_{\bar X \bar Y} $. 
\end{lemma}

\begin{proof}
	Thanks to the triangle inequality and the fact that $Z_{Q_{\bar X\bar Y}} \ge 0 $ almost surely, 
	\begin{align}
		&\bbE \left[\left|Z_{Q_{\bar X\bar Y}} -  \bbE[Z_{Q_{\bar X\bar Y}}]\right|\right] \nonumber \\ 
		&\qquad \le 2\bbE[Z_{Q_{\bar X\bar Y}}] \\ 
		&\qquad = 2 l_{Q_{\bar X|\bar Y}}(y^n) p_{Q_{\bar X|\bar Y}}(y^n) \text{.} \label{russ_comb1}
	\end{align}
	On the other hand, by Jensen's inequality,
	\begin{align}
		&\bbE \left[\left|Z_{Q_{\bar X\bar Y}} -  \bbE[Z_{Q_{\bar X\bar Y}}]\right|\right] \nonumber \\ 
		&\qquad \le \bbE^{\frac 12} \left[\Big|Z_{Q_{\bar X\bar Y}} -  \bbE[Z_{Q_{\bar X\bar Y}}]\Big|^2\right] \\ 
		&\qquad = l_{Q_{\bar X|\bar Y}}(y^n) M^{-\frac 12}p_{Q_{\bar X|\bar Y}}^{\frac 12}(y^n) (1-p_{Q_{\bar X|\bar Y}}(y^n))^{\frac 12} \label{the variance term} \\ 
		&\qquad \le l_{Q_{\bar X|\bar Y}}(y^n) M^{-\frac 12}p_{Q_{\bar X|\bar Y}}^{\frac 12}(y^n) \text{.} \label{russ_comb2}
	\end{align}
Combining \eqref{russ_comb1} and \eqref{russ_comb2} together with the fact that
\begin{align}
  l_{Q_{\bar X|\bar Y}}(y^n) = \frac{P_{Y^n|X^n}(y^n|x^n_{Q_{\bar X|\bar Y}})}{P_{Y^n}(y^n)}
\end{align}	
for some $x^n_{Q_{\bar X|\bar Y}} \in \cT^n_{Q_{\bar X|\bar Y}}(y^n)$ yields \eqref{eqn:lem:to make it sim}. 
\end{proof}
\begin{lemma} \label{lem:Paul's saving Lemma}
Let $W$ and $X$ be non-negative random variables such that $W \le X $ almost surely. Then, for any $c\in (0, \infty)$, 
\begin{align}
 \bbE\left[\frac WX\right] \ge \frac 1c \bbE[W] - \frac 1c \bbE[X 1\{X>c\}] \text{.}
\end{align}
\end{lemma}
\begin{proof} Since both $W$ and $X$ are non-negative,
\begin{align}
	 \bbE\left[\frac WX\right] &\ge  \bbE\left[\frac WX 1\{X \le c \} \right] \\
	 &\ge \frac 1c \bbE[W1\{X \le c \} ] \\
	 &= \frac 1c \bbE[W ] - \frac 1c \bbE[W 1\{X > c \}] \\
	 &\ge  \frac 1c \bbE[W] - \frac 1c \bbE[X 1\{X>c\}] \text{,} \label{from:almost sure bigger}
\end{align}
	where \eqref{from:almost sure bigger} is due to the fact that $W \le X $ almost surely. 
\end{proof}
\begin{lemma}\label{lem:poisson expectation away from mean}
Suppose that $M$ is a Poisson distributed random variable with mean $\mu > 1$. Assuming $\delta \in (\frac 1\mu, 1)$
	\begin{align}
		\bbE[M1\{M>(1+ \delta)\mu\}] \le \mu a_{\delta - \frac1\mu}^\mu  \text{,}
	\end{align}  
	where 
	\begin{align}
		a_\epsilon &= \frac{\rme^\epsilon}{(1+\epsilon)^{1+\epsilon}} \label{def:a_epsilon}  
	\end{align}
is a constant which is strictly less than $1$ for all $\epsilon \in (0,1)$. 
\end{lemma}

\begin{proof}
	Note that 
	\begin{align}
		&\bbE[M1\{M>(1+ \delta)\mu\}]  \nonumber \\ 
		&\qquad  = \mu \bbP[M>(1+ \delta)\mu -1] \label{from:poisson ozel} \\
		&\qquad \le \mu a_{\delta-\frac 1\mu}^\mu \text{,} \label{from:thm 5.4 in mitzenmacher and upfal}
	\end{align}
	where \eqref{from:poisson ozel} holds because $M$ is Poisson distributed; and \eqref{from:thm 5.4 in mitzenmacher and upfal} follows from \cite[Theorem 5.4]{mitzenmacher2017probability}. 
	
To see $a_\epsilon <1$ for any $\epsilon \in (0,1)$, observe that $a_0= 1$ and $a_\epsilon$ is strictly monotone decreasing in $\epsilon \in (0,1) $ as
		\begin{align}
		\frac{\rmd\log_\rme a_\epsilon}{\rmd \epsilon} &= \log_\rme \frac{1}{1+\epsilon} \\
		&< 0 \text{.}
	\end{align}
\end{proof}

\begin{lemma} \label{lem:expected value of M*L_scrC}
Suppose that $M$ is a Poisson distributed random variable with mean $\mu$. Given $y^n \in \cY^n$, 
\begin{align}
	\bbE[M L_{\scrC_M^n}(y^n)] &=  \mu \text{.}
\end{align}
In particular, if $P_{Y^n}(y^n)>0$, 
\begin{align}
	&\bbE[M L_{\scrC_M^n}(y^n)] \nonumber \\ 
	&\quad \ \, = \bbE\left[ \sum_{j=1}^{M} \frac{P_{Y^n|X^n}(y^n|X_j^n)}{P_{Y^n}(y^n)}\right] \\ 
	&\quad \ \, = \sum_{ Q_{\bar X|\bar Y}\in  \cP_n(\cX|Q_{\bar Y})}l_{Q_{\bar X|\bar Y}}(y^n) \bbE[N_{Q_{\bar X|\bar Y}}(y^n)]  \label{eqn:lem:linearity and type version}\\ 
	&\quad \ \, = \mu \text{.}
\end{align}
\end{lemma}

\begin{proof}
If $P_{Y^n}(y^n) = 0$, then $L_{\scrC_M^n}(y^n) = 1$, and 
\begin{align}
	\bbE[M L_{\scrC_M^n}(y^n)] &= \bbE[M] \\
	&= \mu	\text{.}
\end{align}
Suppose $P_{Y^n}(y^n) > 0$, then by definition of $L_{\scrC_M^n}(y^n)$, 
	\begin{align}
&\bbE[M L_{\scrC_M^n}(y^n)] \nonumber \\ 
&\qquad \quad \ \ = \bbE\left[ \sum_{j=1}^{M} \frac{P_{Y^n|X^n}(y^n|X_j^n)}{P_{Y^n}(y^n)}\right] \\ 
&\qquad \quad \ \  =\bbE\left[\bbE\left[ \sum_{j=1}^{M} \frac{P_{Y^n|X^n}(y^n|X_j^n)}{P_{Y^n}(y^n)} \middle| M \right] \right] \label{frm:twr it up} \\
&\qquad \quad \ \ = \bbE[M] \\
&\qquad \quad \ \ =\mu \text{,}
\end{align}
where \eqref{frm:twr it up} follows from the tower property of expectation.

Note that \eqref{eqn:lem:linearity and type version} is due to the linearity of expectation and the fact that both $P_{Y^n|X^n}(y^n|x^n)$ and $P_{Y^n}(y^n)$ depend on $(x^n, y^n)$ through its joint type, see \eqref{eqn:def:L} and the discussion therein. 
\end{proof}

\begin{lemma}\label{lem:abs sum exp and max abs exp}
Suppose that $X_1, \ldots, X_m $ are mutually independent zero-mean random variables, then 
\begin{align}
\bbE\left[\left|\sum_{i=1}^m X_i  \right|\right]	 \ge \max_{i\in \{1, \ldots, m \}} \bbE[|X_i|] \text{.}
\end{align}	
\end{lemma}
\begin{proof}
	Without loss of generality assume 
	\begin{align}
		\bbE[|X_1|] = \max_{i\in \{1, \ldots, m \}} \bbE[|X_i|] \text{,}
	\end{align}
	and note that
	\begin{align}
		\bbE\left[\left|\sum_{i=1}^m X_i  \right|\right]	  &= \bbE \left[\bbE \left[\left|X_1+\sum_{i=2}^m X_i  \right| \bigg| X_1 \right]\right] \label{from:tower}	\\ 
		&\ge \bbE \left[\left|X_1+\bbE \left[\sum_{i=2}^m X_i  \right]\right|\right]	 \label{from:conv of abs} \\ 
		&=\bbE[|X_1|] \text{,} \label{from:zero mean}
	\end{align}
	where \eqref{from:tower} follows from the tower property of expectation; \eqref{from:conv of abs} follows from modulus inequality and the independence of $X_1$ from $X_i$ for $i\neq 1$; lastly \eqref{from:zero mean} follows as the random variables are all zero-mean. 
\end{proof}

\begin{lemma} \label{lem:lower bd on abs mean dev poiss}
Let $N$ be a Poisson distributed random variable with mean $\xi >0$, then\footnote{The inequality in \eqref{eqn:lem:one of the key} is the lower bound counterpart of `upper bounding the absolute mean deviation of binomial random variable by either twice its mean or its standard deviation' that can be seen in the proof of \Cref{lem:absolute mean deviation upper bound for Z_{Q_{barX|barY}}}.}
	\begin{align}
		\bbE[|N-\xi|]   &\ge  \frac{1}{4} \min\left\{2\xi ,\xi^{\frac12} \right\} \label{eqn:lem:one of the key} \text{.}
	\end{align}
\end{lemma}
\begin{proof}
As can be seen in \cite{crow1958meandev}, one can show that
\begin{align}
	\bbE[|N-\xi|]&= \frac{\xi^{\lfloor \xi\rfloor+1}}{\lfloor \xi\rfloor !} 2\rme^{-\xi} \text{.}
\end{align}
To see \eqref{eqn:lem:one of the key}, observe that $\xi\in (0,1]$ implies 
\begin{align}
	\frac{\xi^{\lfloor \xi\rfloor+1}}{\lfloor \xi\rfloor !} 2\rme^{-\xi} &=  2\xi\,\rme^{-\xi} \\
	&\ge \frac{1}{2}\xi \label{aseleasele}  \text{.}
\end{align}	
On the other hand, when $\xi \in (1, \infty)$, by Robbins' sharpening of Stirling's approximation \cite{robbins1955remark},
\begin{align}
 \lfloor \xi \rfloor ! \leq \lfloor \xi \rfloor^{\lfloor \xi \rfloor}\rme^{-{\lfloor \xi \rfloor}+\frac{1}{12{\lfloor \xi \rfloor}}}\sqrt{2\pi {\lfloor \xi \rfloor}} \text{.} \label{eqn:stirling_approx_for_factorial}
\end{align}
Denoting $\tau = \xi- \lfloor \xi \rfloor$, thanks to \eqref{eqn:stirling_approx_for_factorial},
	\begin{align}
		\frac{\xi^{\lfloor \xi\rfloor+1}}{\lfloor \xi\rfloor !} 2\rme^{-\xi} &\ge \frac{2\xi\, \rme^{-\tau-\frac{1}{12\lfloor \xi \rfloor}} }{\sqrt{2\pi\lfloor \xi \rfloor }} \left(1+\frac{\tau}{\lfloor \xi \rfloor}\right)^{\lfloor \xi \rfloor} \\
		&>\frac{2\xi^{\frac12}}{(2\pi)^{\frac12}} \rme^{-\frac{13}{12}}  \label{from:delta_less1} \\
		&> \frac{1}{4}\xi^\frac12 \text{,} 
\end{align} 
where \eqref{from:delta_less1} follows as $0\le \tau<1$, and $1\le \lfloor \xi \rfloor \le \xi$. Combining \eqref{aseleasele} and \eqref{from:delta_less1} yields \eqref{eqn:lem:one of the key}.
\end{proof}
\begin{lemma}\label{lem:poisson mean probability}
Let $M$ be a Poisson distributed random variable with mean $\mu \ge 1$, then
\begin{align}
	\bbP[M= \lceil \mu \rceil ] >  \frac{1}{8 \lceil \mu \rceil^\frac12} \text{.}
\end{align}
\end{lemma}
\begin{proof}
Let $\tau = \lceil\mu \rceil  - \mu $, using Stirling approximation as in \eqref{eqn:stirling_approx_for_factorial}, 
	\begin{align}
		\bbP[M= \lceil \mu \rceil] &= \frac{\mu^{\lceil\mu\rceil}}{\lceil \mu \rceil!} \rme^{-\mu} \\
		&\ge \frac{\rme^{\tau - \frac{1}{12\lceil\mu\rceil}}}{\sqrt{2\pi\lceil\mu\rceil}}  \left(1- \frac{\tau}{\lceil \mu \rceil}\right)^{\lceil\mu \rceil} \\
		&>  \frac{1}{8 \lceil \mu \rceil^\frac12 } \text{,} \label{frm:poiss smpl bd} 
	\end{align}
where \eqref{frm:poiss smpl bd} follows from the facts that  $\log_\rme (1-x) \ge -x - \frac{x^2}{1-x}$ for $x\in [0,1)$, $\tau <1 $, and $ \mu \ge 1 $. 
\end{proof}
\begin{lemma}\label{lem:TV-concentration_random-poisson}
	Let $M$ be a Poisson distributed random variable with mean $\mu$,
\begin{align}
 &\bbP\left[\left|\big \|\induced - P_{Y^n}\big\|^{}_1 - \bbE\big[\big\|\induced - P_{Y^n}\big\|^{}_1\big]\right|\ge t\right] \nonumber \\  
 &\qquad \qquad \le 2 \exp_\rme\left(-\mu \left(1-\rme^{-t^2/2}\right)  \right)  \label{eqn:lem:TV concentration when M poisson} \\
 &\qquad \qquad \le 	2 \exp_\rme\left(-\frac{\mu t^2}{2+t^2}\right) \text{.} \label{eqn:cool coincidence}
 \end{align}
\end{lemma}
\begin{proof}
For the sake of notational convenience, let
\begin{align}
	T_n(M) &= \left\|\induced - P_{Y^n}\right\|^{}_1 \text{,} \\
	V_n(M) &= T_n(M) - \bbE[T_n(M)]   \text{.}
\end{align}
Conditioned on $M=m$, by \Cref{lem:TV_concentration-McDiarmid-deterministic},
	\begin{align}
	\bbP\left[ |V_n(M) | \ge t | M=m \right] \le 2 \exp_\rme\left(-\frac{mt^2}2\right) \text{.}
	\end{align}
	Hence, by the total probability law,
	\begin{align}
		\bbP\left[|V_n(M)| \ge t\right ] &\le 2\bbE\left[\exp_\rme\left(-\frac{Mt^2}2\right)\right] \\ 
		&=2 \exp_\rme\left(-\mu \left(1-\rme^{-t^2/2}\right)  \right) \text{.}
	\end{align}
	To see \eqref{eqn:cool coincidence}, simply note that $\rme^{-x} \le \frac{1}{1+x} $. 
\end{proof}
\begin{lemma} \label{lem:TV concentrants to poisson-TV mean}
	Let $M$ be a Poisson distributed random variable with mean $\mu$,
\begin{align}
	&\bbP[| T_n(\lceil \mu \rceil ) - \bbE[T_n(M) ]| \ge t ] \nonumber \\ 
	&\qquad \qquad  \le 16 \lceil \mu \rceil^\frac12  \exp_\rme\left(-\frac{\mu t^2}{2+t^2}\right) \text{,}
\end{align}
where $T_n(m) = \|P_{Y^n|\scrC_m^n} - P_{Y^n} \|^{}_1$.  
\end{lemma}
\begin{proof}
Let $\widetilde M$ be an independent copy of $M$, and observe that
	\begin{align}
&\bbP[| T_n(\lceil \mu \rceil ) - \bbE[T_n(M) ]| \ge t ] \nonumber \\  
&\quad \ \ \, = \frac{\bbP[|T_n(\widetilde M) - \bbE[T_n(M)]| \ge t, \widetilde M = \lceil \mu \rceil]}{\bbP[\widetilde M = \lceil \mu \rceil]}\\
&\quad \ \ \, =\frac{\bbP[|T_n(\widetilde M) - \bbE[T_n(\widetilde M)]| \ge t, \widetilde M = \lceil \mu \rceil]}{\bbP[\widetilde M = \lceil \mu \rceil]}\\ 
&\quad \ \ \, \le \frac{\bbP[|T_n(\widetilde M) - \bbE[T_n(\widetilde M)]| \ge t]}{\bbP[\widetilde M= \lceil \mu \rceil]}  \text{,} 
	\end{align}
the	result is immediate from \Cref{lem:poisson mean probability,lem:TV-concentration_random-poisson}.
\end{proof}
\section{Asymptotic Exponents of the Key Quantities}\label{apdx:asymptotic exponents}
This section provides the asymptotic\footnote{Non-asymptotic exponents are given wherever possible which are then used in proving the finite block-length results contained in \Cref{apdx:finite block-length results}.} exponents of the several key quantities that play a central role in the proofs of \Cref{thm:main,thm:cc:main}. 
\subsection{Exponents Used in the Proof of Theorem~\ref{thm:main}}\label{apdxsub:asymptotic exponents}
\begin{lemma}\label{lem:lem for p}
Fix $y^n\in \cY^n$, and let $Q_{\bar Y}\in \cP_n(\cY)$ denote its type. For any $Q_{\bar X|\bar Y} \in \cP_n(\cX|Q_{\bar Y})$
	\begin{align}
			 &p_{Q_{\bar X|\bar Y}}(y^n) \nonumber \\ 
			 &\qquad =  \bbP\left[X^n \in \cT^n_{Q_{\bar X|\bar Y}}(y^n)\right ] \label{eqn:lem:lem for p1} \\ 
			&\qquad = \exp \left(-n\bbE[\imath^{}_{P_X}(\bar X) ] \right) \left|\cT^n_{Q_{\bar X|\bar Y}}(y^n)\right| \text{,} \label{eqn:lem:lem for p2}
	\end{align}
	where $p_{Q_{\bar X|\bar Y}}(y^n)$ is defined in \eqref{eqn:def:p_{Q_{barX|barY}}}, $\{X_i\}_{i=1}^n$ are i.i.d. according to $P_X$, and $\bar X \sim Q_{\bar X}$ with $Q_{\bar X}$ denoting the $\cX$-marginal of the joint $n$-type $Q_{\bar X|\bar Y}Q_{\bar Y}$.
\end{lemma}
\begin{proof}
Note that
	\begin{align}
		&\bbP\left[X^n \in \cT^n_{Q_{\bar X|\bar Y}}(y^n)\right]  \nonumber \\
		&\qquad \qquad = \sum_{x^n \in \cT^n_{Q_{\bar X|\bar Y}}(y^n)} P_{X^n}(x^n)  \\ 
%		&= \sum_{x^n \in \cT^n_{Q_{\bar X|\bar Y}}(y^n)} \prod_{i=1}^n P_{X_i}(x_i) \\
		&\qquad \qquad = \sum_{x^n \in \cT^n_{Q_{\bar X|\bar Y}}(y^n)} \prod_{a\in \cX } P_X^{nQ_{\bar X}(a) }(a) \label{from:def:Na} \\ 
		&\qquad \qquad =  \exp \left(-n\bbE[\imath^{}_{P_X}(\bar X) ] \right) \left|\cT^n_{Q_{\bar X|\bar Y}}(y^n)\right| \text{,} 
	\end{align}
where in \eqref{from:def:Na} $nQ_{\bar X}(a) \in \{0, 1, \ldots, n \} $ denotes the number of times that $a\in \cX$ appears in $ \{ x_i  \}_{i=1}^n$. 
\end{proof}
\begin{lemma} \label{lem:exponent_for_uniform_convergence}
Let $\mfrY(M, Q_{\bar X\bar Y})$ be as defined in \eqref{eqn:def:mfrakY}. Assuming\footnote{For the ease of presentation, the fact that $M$ is an integer is ignored. A more careful analysis with $M = \lceil \exp(nR) \rceil$ results in $\kappa_n = \frac{|\cX||\cY|}{n}\log (n+1) + \frac 1n \log(2\sqrt{2})$ as $\exp(nR) \le \lceil \exp(nR) \rceil \le 2\exp(nR)$.} $M = \exp(nR)$, for any $Q_{\bar X\bar Y} \in \cP_n(\cX\times\cY)$
  \begin{align}
    & D(Q_{\bar X\bar Y}\|P^{}_X Q_{\bar Y})+ \frac 12 \left[R - D(Q_{\bar X\bar Y}\|P^{}_X Q_{\bar Y})\right]_+  \nonumber \\ 
    &\qquad \le -\frac1n \log \left(\frac12\mfrY(M, Q_{\bar X\bar Y}) \right) \label{eqnlem:lower_bd}  \\
    &\qquad \le  D(Q_{\bar X\bar Y}\|P^{}_X Q_{\bar Y})\label{eqnlem:upper_bd} \\ 
    &\qquad \qquad \qquad  + \frac12 \left[R- D(Q_{\bar X\bar Y}\|P^{}_X Q_{\bar Y})  \right]_+  + \kappa_n  \text{,} \nonumber
  \end{align}
  where $Q_{\bar Y}$ is the $\cY$-marginal of $Q_{\bar X\bar Y}$, and
  \begin{align}
    [f]_+ &= \max\{0, f\} \text{,} \\
 \kappa_n &= \frac{|\cX||\cY|}{n}\log (n+1) + \frac 1n \log2 \text{.} \label{eqn:def:kappa}
  \end{align} 
\end{lemma}
\begin{proof}
Noting that 
\begin{align}
  \bbE[\imath^{}_{P_X}(\bar X) ] - H(\bar X |\bar Y) =  D(Q_{\bar X\bar Y } \| P^{}_X Q_{\bar Y}) \text{,}
\end{align}
where $(\bar X, \bar Y) \sim Q_{\bar X\bar Y} = Q_{\bar X| \bar Y} Q_{\bar Y} $, \eqref{eqnlem:lower_bd} is a direct consequence of \Cref{lem:lem for p} and the upper bound in \cite[Lemma 2.5]{csiszar2011information}. To see \eqref{eqnlem:upper_bd}, observing
\begin{align}
  &\frac12\mfrY(M, Q_{\bar X\bar Y})  \nonumber \\ 
  &\quad \ \,  = p^{}_{Q_{\bar X|\bar Y}}(y^n) \min\left\{1 , \frac12 M^{-\frac 12} p_{Q_{\bar X|\bar Y}}^{ -\frac 12}(y^n)\right\} \text{,}
\end{align}
and applying \cite[Lemma 2.5]{csiszar2011information} and \Cref{lem:lem for p} suffices. 
\end{proof}
\begin{lemma}\label{lem:infimum over all types}
Given a joint $n$-type $Q_{\bar X \bar Y} \in \cP_n(\cX\times \cY)$, suppose that $(\bar X, \bar Y)\sim Q_{\bar X\bar Y} $. Let $\mfrY(M, Q_{\bar X\bar Y})$ be as defined in \eqref{eqn:def:mfrakY}, then
  \begin{align}
    &\lim_{n\to \infty} -\frac 1n \log \max_{Q_{\bar X\bar Y}} \bigg\{\frac12 \left|\cT^n_{Q_{\bar Y}}\right|   \nonumber \\ 
    &\qquad \quad \times \exp(-n\bbE[\imath^{}_{P_{Y|X}}(\bar Y|\bar X)]) \mfrY(M, Q_{\bar X\bar Y})   \bigg\} \nonumber \\
    &\qquad = \inf_{Q_{\bar X\bar Y}\in \cP_\infty(\cX\times \cY)} \bigg\{ D(Q_{\bar X\bar Y } \| P^{}_{XY}) \\ 
    &\qquad \qquad \qquad \ \   + \frac 12 \left[R - D(Q_{\bar X\bar Y} \| P^{}_X Q_{\bar Y})\right]_+ \bigg\} \text{,} \nonumber
  \end{align}
  where 
  \begin{align}
    \cP_\infty(\cX\times \cY) &= \bigcup_{n=1}^\infty \cP_n(\cX\times \cY) \text{,}  \\
    [f]_+&= \max\{0, f\} \text{.}
  \end{align}
\end{lemma}
\begin{proof}
Using \cite[Lemma 2.3]{csiszar2011information}, \Cref{lem:exponent_for_uniform_convergence}, and the fact that
  \begin{align}
     D(Q_{\bar X\bar Y}\|P^{}_X Q_{\bar Y})- & H(Q_{\bar Y})+\bbE[\imath^{}_{P_{Y|X}}(\bar Y|\bar X)] 
     \nonumber\\
      & = D(Q_{\bar X\bar Y}\| P^{}_{XY}) \text{,} 
  \end{align}
it follows that, for any fixed $n$,
\begin{align}
   & \min_{Q_{\bar X\bar Y}\in \cP_n(\cX\times \cY)} \bigg\{ D(Q_{\bar X\bar Y}\|P^{}_{XY}) \nonumber \\ 
   &\qquad \qquad \qquad \qquad  + \frac12 \left[R-D(Q_{\bar X\bar Y} \| P^{}_X Q_{\bar Y})\right]_+ \bigg\} \nonumber \\
   &\quad \le -\frac 1n \log  \max_{Q_{\bar X\bar Y}\in \cP_n(\cX\times \cY)} \bigg\{ \frac12 \left|\cT^n_{Q_{\bar Y}}\right|  \label{for:fb:lower}  \\ 
   &\qquad \qquad \ \times \exp(-n\bbE[\imath^{}_{P_{Y|X}}(\bar Y|\bar X)]) \mfrY(M, Q_{\bar X\bar Y}) \bigg\} \nonumber \\
   &\quad \le  \min_{Q_{\bar X\bar Y}\in \cP_n(\cX\times \cY)} \bigg\{ D(Q_{\bar X\bar Y } \| P^{}_{XY}) \label{for:fb:upper} \\ 
   &\qquad \qquad \qquad \qquad + \frac 12 \left[R - D(Q_{\bar X\bar Y} \| P^{}_X Q_{\bar Y})\right]_+ \bigg\} \nonumber \\ 
   &\qquad \qquad \qquad \qquad \qquad \qquad + \kappa_n  + \frac{|\cY|}{n} \log (n+1) \text{,} \nonumber
\end{align}
where $\kappa_n $ is as defined in \eqref{eqn:def:kappa}. Taking $n\to \infty$ yields the desired result as $\kappa_n \to 0$. 
\end{proof}
\subsection{Exponents Used in the Proof Theorem~\ref{thm:cc:main}}\label{apdxsub:cc:exponents}
This section contains some additional asymptotic and non-asymptotic results that are needed in proving \Cref{thm:cc:main} (in \Cref{sec:main}) and \Cref{thm:cc:finite block-length cons. comp. exact scl} (in Appendix~\ref{apdx:finite block-length results}).
\begin{lemma}\label{lem:cc:the magic lemma}
Suppose $y^n \in \cT^n_{Q_{\bar Y}}$, and $Q_{\bar X|\bar Y}\in \cP_n(\cX|Q_{\bar Y}; P_{\bar X}) $. Then,
\begin{align}
	  &\sum_{x^n \in \cT^n_{P_{\bar X}}}1\left\{x^n \in \cT^n_{Q_{\bar X| \bar Y}}(y^n)\right\} \nonumber \\ 
	   &\qquad = \sum_{x^n \in \cT^n_{P_{\bar X}}} 1\left\{(x^n, y^n)\in \cT^n_{Q_{\bar X \bar Y}}\right\}  \label{eqn:lem:type stuff1} \\
	   &\qquad = \frac{\left|\cT^n_{Q_{\bar X \bar Y}}\right|}{\left|\cT^n_{Q_{\bar Y}}\right|} \text{,} \label{eqn:lem:type stuff2}
\end{align}
where $Q_{\bar X \bar Y} = Q_{\bar X|\bar Y}Q_{\bar Y}= P_{\bar X}Q_{\bar Y|\bar X} $ for some conditional type $Q_{\bar Y|\bar X} $ given $x^n \in \cT^n_{P_{\bar X}}$.  
\end{lemma}
\begin{proof}
It is easy to get \eqref{eqn:lem:type stuff1}:
\begin{align}
& \sum_{x^n \in \cT^n_{P_{\bar X}}}1\left\{x^n \in \cT^n_{Q_{\bar X| \bar Y}}(y^n)\right\} \nonumber \\ 
&=  \sum_{x^n\in\cT^n_{P_{\bar X}}}1\left\{x^n \in \cT^n_{Q_{\bar X| \bar Y}}(y^n)\right\} 1\left\{y^n \in \cT^n_{Q_{\bar Y}}\right\}   \\ 
&= \sum_{x^n \in \cT^n_{P_{\bar X}}} 1\left\{(x^n, y^n)\in \cT^n_{Q_{\bar X \bar Y}}\right\} \text{.} 
\end{align}
To establish \eqref{eqn:lem:type stuff2}, observe that
\begin{align}
  &\left|\cT^n_{Q_{\bar X \bar Y}}\right| \nonumber \\
  &\quad  = \sum_{(a^n,b^n) \in \cX^n\times \cY^n} 1\left\{(a^n,b^n)\in \cT^n_{Q_{\bar X\bar Y}}\right\} \\
%  &= \sum_{(a^n,b^n) \in \cX^n\times \cY^n} 1\left\{(a^n,b^n)\in \cT^n_{Q_{\bar X\bar Y}}\right\} 1\left\{a^n\in \cT^n_{P_{\bar X}}\right\}1\left\{b^n\in \cT^n_{Q_{\bar Y}}\right\} \label{frm:cc:fixed marginals} \\
  &\quad  = \sum_{b^n \in \cT^n_{Q_{\bar Y}}}  \sum_{a^n\in \cT^n_{P_{\bar X}}} 1\left\{(a^n, b^n)\in \cT^n_{Q_{\bar X\bar Y}}\right\} \label{xxx:summand here} \\
  &\quad = \sum_{b^n \in \cT^n_{Q_{\bar Y}}}  \sum_{a^n\in \cT^n_{P_{\bar X}}} 1\left\{(a^n, y^n)\in \cT^n_{Q_{\bar X\bar Y}}\right\}  \label{frm:cc:dep on type of y^n}  \\
  &\quad = \left|\cT^n_{Q_{\bar Y}}\right| \sum_{x^n \in \cT^n_{P_{\bar X}}} 1\left\{(x^n, y^n)\in \cT^n_{Q_{\bar X \bar Y}}\right\} \text{,} 
\end{align}
where \eqref{xxx:summand here} follows because $\cX$- and $\cY$-marginals of $Q_{\bar X\bar Y}$ are fixed to be $P_{\bar X}$ and $Q_{\bar Y}$; \eqref{frm:cc:dep on type of y^n} follows because $y^n\in \cT^n_{Q_{\bar Y}}$  and $\sum_{a^n\in \cT^n_{P_{\bar X}}} 1\left\{(a^n, b^n)\in \cT^n_{Q_{\bar X\bar Y}}\right\} $ depends on $b^n$ only through its type $Q_{\bar Y}$.
\end{proof}
\begin{lemma}\label{lem:cc:exponent of p}
	Suppose $y^n\in \cT^n_{Q_{\bar Y}} $, and $Q_{\bar X|\bar Y}\in \cP_n(\cX|Q_{\bar Y}; P_{\bar X}) $. Then,
	\begin{align}
		\breve{p}_{Q_{\bar X|\bar Y}}(y^n) &= \frac{\left|\cT^n_{Q_{\bar X \bar Y}}\right|}{\left|\cT^n_{P_{\bar X}}\right|\left|\cT^n_{Q_{\bar Y}}\right|} \text{,}
	\end{align}
	where $Q_{\bar X \bar Y} = Q_{\bar X|\bar Y}Q_{\bar Y}= P_{\bar X}Q_{\bar Y|\bar X} $ for some conditional type $Q_{\bar Y|\bar X} $ given $x^n \in \cT^n_{P_{\bar X}}$.  
\end{lemma}
\begin{proof}
\begin{align}
		&\breve{p}_{Q_{\bar X|\bar Y}}(y^n) \nonumber \\ 
		&\qquad \quad = \bbP\left[\breve X^n \in \cT^n_{Q_{\bar X|\bar Y}}(y^n)\right] \\
		&\qquad \quad = \frac{1}{\left|\cT^n_{P_{\bar X}}\right|} \sum_{x^n \in \cT^n_{P_{\bar X}}}1\left\{x^n \in \cT^n_{Q_{\bar X| \bar Y}}(y^n)\right\}   \\
		&\qquad \quad = \frac{1}{\left|\cT^n_{P_{\bar X}}\right|} \sum_{x^n\in\cT^n_{P_{\bar X}}} 1\left\{(x^n, y^n)\in\cT^n_{Q_{\bar X \bar Y}}\right \} \label{by:lem1 jiglypuff} \\
		&\qquad \quad = \frac{\left|\cT^n_{Q_{\bar X \bar Y}}\right|}{\left|\cT^n_{P_{\bar X}}\right|\left|\cT^n_{Q_{\bar Y}}\right|} \text{,} \label{by:lem1 iglypuff}
\end{align}
where \eqref{by:lem1 jiglypuff} and \eqref{by:lem1 iglypuff} both follow from \Cref{lem:cc:the magic lemma}.
\end{proof}
\begin{lemma}\label{lem:cc:key first limit} 
  Given an $m$-type $P_{\bar X} \in \cP_m(\cX)$ and a conditional type\footnote{We assume $n\in m\bbN$. In \eqref{eqn:lem:cc:meaning of n and life}, $n\to \infty $ means that $n = km $ and $k\to \infty$.} $Q_{\bar Y|\bar X} \in \cP_{n}(\cY| P_{\bar X})$ given $x^n \in \cT^n_{P_{\bar X}}$, suppose $(\bar X, \bar Y) \sim P_{\bar X}Q_{\bar Y|\bar X}$. Let\footnote{For the ease of presentation, the fact that $M$ is an integer is ignored. A more careful analysis with $M = \lceil \exp(nR) \rceil$ results in $\breve \kappa_n = \frac{|\cX|+2|\cX||\cY|+|\cY|}{2n} \log(n+1) + \frac1n \log(2\sqrt{2})$ as $\exp(nR) \le \lceil \exp(nR) \rceil \le 2\exp(nR) $.} $M= \exp(nR) $, and $\breve{\mfrY}(M, P_{\bar X}Q_{\bar Y|\bar X})$ be as defined in \eqref{for:cc:remark_end}, then
  \begin{align}
    &\lim_{n\to \infty} -\frac1n\log  \max_{Q_{\bar Y |\bar X} \in \cP_n(\cY|P_{\bar X})} \bigg\{ \frac12 \left|\cT^n_{Q_{\bar Y}}\right| \nonumber \\ 
    &\qquad \ \, \times \exp(-n\bbE[\imath^{}_{P_{Y|X}}(\bar Y|\bar X)]) \breve{\mfrY}(M, P_{\bar X}Q_{\bar Y|\bar X})\bigg\} \nonumber  \\
    &\quad = \inf_{Q_{\bar Y|\bar X} \in \cP_\infty(\cY|P_{\bar X})} \bigg\{ D(P_{\bar X}Q_{\bar Y|\bar X}\|P_{\bar XY}) \label{eqn:lem:cc:meaning of n and life} \\ 
    &\qquad \qquad \qquad + \frac12\left[R-D(P_{\bar X}Q_{\bar Y|\bar X}\|P_{\bar X}Q_{\bar Y})\right]_+\bigg\} \text{,} \nonumber 
  \end{align}
  where 
  \begin{align} 
    \cP_\infty(\cY|P_{\bar X}) &= \bigcup_{n\in m\bbN}  \cP_{n}(\cY|P_{\bar X}) \text{,} \\
    P_{\bar XY} &= P_{\bar X}P_{Y|X} \text{,} \\
    [f]_+&= \max\{0, f\} \text{.}
  \end{align}
\end{lemma}
\begin{proof}
 From the definition of $  \breve{\mfrY}(M, P_{\bar X}Q_{\bar Y|\bar X})$ in \eqref{for:cc:remark_end}, \Cref{lem:cc:exponent of p} implies that
  \begin{align}
  &\breve{\mfrY}(M, P_{\bar X}Q_{\bar Y|\bar X}) \nonumber \\ 
  &= \frac{\left|\cT^n_{P_{\bar X}Q_{\bar Y|\bar X}}\right|}{\left|\cT^n_{P_{\bar X}}\right| \left|\cT^n_{Q_{\bar Y}}\right|} \min\left\{2 , M^{-\frac 12} \frac{\left|\cT^n_{P_{\bar X}}\right|^\frac12 \left|\cT^n_{Q_{\bar Y}}\right|^\frac12 }{\left|\cT^n_{P_{\bar X}Q_{\bar Y|\bar X}}\right|^\frac12 }\right\} \text{.}
  \end{align}
Observing
  \begin{align}
    &H(P_{\bar X}) + \bbE[\imath^{}_{P_{Y|X}}(\bar Y|\bar X)]) - H(P_{\bar X}Q_{\bar Y|\bar X}) \nonumber \\
    &\qquad \qquad \qquad \qquad \quad = D(P_{\bar X}Q_{\bar Y|\bar X}\|P_{\bar XY}) \text{,}\\
    &H(P_{\bar X}) + H(Q_{\bar Y}) - H(P_{\bar X}Q_{\bar Y|\bar X}) \nonumber \\ 
    &\qquad \qquad \qquad \qquad \quad = D(P_{\bar X}Q_{\bar Y|\bar X}\|P_{\bar X}Q_{\bar Y} ) \text{,}  
  \end{align}
and using \cite[Lemma 2.3]{csiszar2011information} tailored for the type classes $\cT^n_{P_{\bar X}} $, $\cT^n_{Q_{\bar Y}} $,  and $\cT^n_{P_{\bar X}Q_{\bar Y|\bar X}}$, it follows for any fixed $n$ that 
  \begin{align}
  &\min_{Q_{\bar Y |\bar X} \in \cP_n(\cY|P_{\bar X})} \bigg\{ D(P_{\bar X}Q_{\bar Y|\bar X}\|P_{\bar XY}) \nonumber \\ 
  &\qquad \qquad + \frac12[R-D(P_{\bar X}Q_{\bar Y|\bar X}\|P_{\bar X}Q_{\bar Y})]_+\bigg\} - \breve \iota_n \nonumber  \\
    &\quad \le -\frac1n\log  \max_{Q_{\bar Y |\bar X} \in \cP_n(\cY|P_{\bar X})} \bigg\{ \frac12\left|\cT^n_{Q_{\bar Y}}\right|\label{for:cc:fb:lower} \\
    &\qquad \quad \times \exp(-n\bbE[\imath^{}_{P_{Y|X}}(\bar Y|\bar X)]) \breve{\mfrY}(M, P_{\bar X}Q_{\bar Y|\bar X})\bigg\} \nonumber \\
    &\quad \le \min_{Q_{\bar Y |\bar X} \in \cP_n(\cY|P_{\bar X})} \bigg\{ D(P_{\bar X}Q_{\bar Y|\bar X}\|P_{\bar XY}) \label{for:cc:fb:upper} \\ 
    &\qquad \qquad  + \frac12[R-D(P_{\bar X}Q_{\bar Y|\bar X}\|P_{\bar X}Q_{\bar Y})]_+  \bigg\} + \breve \kappa_n \text{,}  \nonumber
  \end{align}
  where
  \begin{align}
  \breve \iota_n &= \frac{|\cX|(2+|\cY|)}{2n}\log(n+1) - \frac1n \log 2 \text{,} \\
  \breve \kappa_n &= \frac{|\cX|+2|\cX||\cY|+|\cY|}{2n} \log(n+1) + \frac1n \log 2  \text{.}
  \end{align}
 Taking $n\to \infty$ yields the desired result as both $\breve \iota_n \to 0 $, and $\breve \kappa_n \to 0 $.   
\end{proof}
\section{Optimization over Types in the Limit}\label{apdx:optimizations over types in the limit}
\subsection{Optimization over Joint Types in the Limit}\label{apdxsub:optimizations over joint types in the limit}
\begin{lemma}\label{lem:optimization over types in the limit}
Let $\cP_\infty(\cX\times \cY) = \bigcup_{n\in \bbN}\cP_n(\cX \times \cY)$. Then,
\begin{align}
	& \inf_{Q_{\bar X \bar Y} \in \cP_\infty (\cX \times \cY)} \bigg\{ D(Q_{\bar X\bar Y}\|P^{}_{XY}) \nonumber \\
	&\qquad \qquad \qquad \quad + \frac12 \left[R- D(Q_{\bar X\bar Y}\|P^{}_XQ_{\bar Y})\right]_+ \bigg\} \nonumber \\
	&\quad =\min_{Q_{XY} \in \cP(\cX \times \cY)}\bigg\{  D(Q_{XY}\|P_{XY}) \label{minimization over types is no problem in the limit} \\ 
	&\qquad \qquad \qquad \quad  + \frac12 \left[R- D(Q_{XY}\|P_XQ_Y)  \right]_+  \bigg\}  \text{.}  \nonumber
\end{align}
\end{lemma}
\begin{proof}
First of all, since $\cP_n(\cX \times \cY) \subset \cP(\cX \times \cY)$ for all $n\in \bbN$, $\ge$ is trivial in \eqref{minimization over types is no problem in the limit}. To establish $\le$, let $Q^\star_{XY} $ be the minimizer in the right side of \eqref{minimization over types is no problem in the limit}. We may assume that $Q^\star_{XY} \ll P_{XY} $, otherwise $D(Q^\star_{XY} \| P^{}_{XY}) = +\infty $ which contradicts the minimality of $Q^\star_{XY}$. Since for every $Q_{XY} \in \cP(\cX \times \cY) $ either $Q_{XY} \in \cP_\infty(\cX\times \cY)$ or $Q_{XY}$ is a limit point of $\cP_\infty(\cX\times\cY)$, it follows that $\cP_\infty(\cX\times \cY)$ is dense in $ \cP(\cX \times \cY)$. Hence, one can find a sequence of types $\big\{{Q^{\star}_{\bar X\bar Y}}_{[k]} \in \cP_\infty(\cX \times \cY)\big\}_{k\in \bbN} $ such that
\begin{align}
	\lim_{k\to \infty} \left\|Q^\star_{XY} - {Q^{\star}_{\bar X\bar Y}}_{[k]}\right\|^{}_1 &= 0 \label{limiting optimizer} \text{.}
\end{align}
We may assume ${Q^{\star}_{\bar X\bar Y}}_{[k]} \ll  P^{}_{XY} $ as well. Note that, for all $k\in \bbN$,
\begin{align}
&\inf_{Q_{\bar X \bar Y} \in \cP_\infty(\cX \times \cY)} \bigg\{ D(Q_{\bar X \bar Y } \| P^{}_{XY}  ) \nonumber \\ 
& \qquad \qquad \qquad \quad + \frac12 \left[R- D(Q_{\bar X \bar Y } \| P^{}_X Q_{\bar Y} )  \right]_+  \bigg\} \nonumber \\ 
&\qquad  \le D({Q^{\star}_{\bar X\bar Y}}_{[k]} \| P^{}_{XY}  ) \label{ineq:zz} \\ 
&\qquad \qquad \qquad \quad + \frac12 \big[R- D({Q^{\star}_{\bar X\bar Y}}_{[k]} \| P^{}_X {Q^\star_{\bar Y}}_{[k]})  \big]_+  \text{,} \nonumber
\end{align}
where 
\begin{align}
	{Q^\star_{\bar Y}}_{[k]}(y) &= \sum_{x\in \cX} {Q^{\star}_{\bar X\bar Y}}_{[k]}(x,y) \text{.}
\end{align}
Since both $D(Q_{XY} \| P_{XY})$ and $D(Q_{XY}\|P_XQ_Y)$ are convex functions of $Q_{XY}$ on the finite dimensional space $\cP(\cX\times \cY)$, they are both continuous in $Q_{XY} $ throughout the relative interior of $\cP(\cX\times \cY)$, see, e.g., \cite[Section 7.9]{luenberger1997optimization}. Therefore,\footnote{In order for \eqref{strme up} to hold, $Q^{\star}_{XY}$ needs to be in the relative interior of $\cP(\cX\times \cY)$. If $Q^{\star}_{XY}$ is on the boundary of $\cP(\cX\times\cY)$, restricting attention to a smaller simplex suffices.\label{ftnt:continuity_of_relent}}
\begin{align}
	&\inf_{Q_{\bar X \bar Y} \in \cP_\infty(\cX \times \cY)} \bigg\{  D(Q_{\bar X\bar Y}\|P^{}_{XY}) \nonumber \\ 
	&\qquad \qquad \qquad  + \frac12 \left[R- D(Q_{\bar X\bar Y}\|P^{}_XQ_{\bar Y})\right]_+ \bigg\} \nonumber  \\ 
	&\le \lim_{k\to \infty} \bigg\{D({Q^{\star}_{\bar X\bar Y}}_{[k]} \| P^{}_{XY} ) \label{from:ineq in the zz} \\ 
	&\qquad \qquad \qquad  + \frac12 \big[R- D({Q^{\star}_{\bar X\bar Y}}_{[k]} \| P^{}_X {Q^\star_{\bar Y}}_{[k]})  \big]_+\bigg\} \nonumber  \\
	& = D(Q^{\star}_{XY} \| P^{}_{XY}  )  + \frac12 \big[R- D(Q^{\star}_{XY} \| P^{}_XQ^\star_{Y})  \big]_+ \text{,} \label{strme up}
\end{align}
where \eqref{from:ineq in the zz} is due to \eqref{ineq:zz}; and in \eqref{strme up} $Q^\star_{Y}(y) = \sum_{x\in \cX} Q^{\star}_{XY} (x,y) $. 
\end{proof}
\subsection{Optimization over Conditional Types in the Limit}\label{apdxsub:cc:optimizations over conditional types in the limit}
\begin{lemma}\label{lem:cc:optimization over conditional types in the limit}
 Given an $m$-type $P_{\bar X}\in  \cP_m(\cX) $, let $  \cP_\infty( \cY|P_{\bar X}) = \bigcup_{n\in m\bbN} \cP_n(\cY|P_{\bar X})$. Then,
  \begin{align}
    &\inf_{Q_{\bar Y|\bar X} \in \cP_\infty(\cY|P_{\bar X})} \bigg\{ D(P_{\bar X}Q_{\bar Y|\bar X}\|P_{\bar XY}) \nonumber \\ 
    &\qquad \qquad \qquad  + \frac12\left[R-D(P_{\bar X}Q_{\bar Y| \bar X }\|P_{\bar X}Q_{\bar Y})\right]_+\bigg\} \nonumber \\
    &\quad =  \min_{Q_{Y|X} \in \cP(\cY|\cX)}  \bigg\{ D(P_{\bar X}Q_{Y|X}\|P_{\bar XY})\label{eqn:lem:opt over condt in the lim} \\ 
    &\qquad \qquad \qquad  + \frac12\left[R-D(P_{\bar X}Q_{Y|X}\|P_{\bar X}Q_Y)\right]_+\bigg\} \text{,}   \nonumber
 \end{align}
  where $\cP(\cY|\cX)$ denotes the set of all random transformations from $\cX $ to $\cY$, $P_{\bar XY} = P_{\bar X} P_{Y|X}$, and $Q_{Y}$ is such that $P_{\bar X} \to Q_{Y|X} \to Q_{Y} $. 
\end{lemma}
\begin{proof}
Since $\cP_n(\cY|P_{\bar X}) \subset \cP (\cY|\cX) $ for all $n\in m\bbN$, $\ge$ is trivial in \eqref{eqn:lem:opt over condt in the lim}. To establish $\le $, let $Q^\star_{Y|X} $ be the minimizer in the right side of \eqref{eqn:lem:opt over condt in the lim}. We may assume that $P_{\bar X}Q_{Y|X} \ll P_{\bar XY}$, otherwise $D(P_{\bar X}Q_{Y|X}\|P_{\bar XY}) = +\infty$, which contradicts the minimality of $ Q^\star_{Y|X} $. Since for every probability transition matrix $Q_{Y|X} \in \cP (\cY|\cX)$ either $Q_{Y|X} \in \cP_\infty(\cY|P_{\bar X}) $ or $Q_{Y|X}$ is a limit point of $\cP_\infty(\cY|P_{\bar X})$, it follows that $\cP_\infty(\cY|P_{\bar X}) $ is dense in $\cP(\cY|\cX)$. Hence, we can find a sequence of conditional types $\{{Q^\star_{\bar Y|\bar X}}_{[k]} \in \cP_\infty(\cY|P_{\bar X})\}_{k\in \bbN} $ such that 
\begin{align}
 \lim_{k\to \infty} \left\| {Q^\star_{Y|X}} - {{Q^\star_{\bar Y|\bar X}}}_{[k]} \right\|^{}_1 = 0\text{.}
\end{align}
We may assume $P_{\bar X} {Q^\star_{\bar Y|\bar X}}_{[k]} \ll P_{\bar XY} $ as well. Note that, for all $k\in \bbN $,
\begin{align}
  &\inf_{Q_{\bar Y|\bar X} \in \cP_\infty(\cY|P_{\bar X})} \bigg\{ D(P_{\bar X}Q_{\bar Y|\bar X}\|P_{\bar XY}) \nonumber \\ 
  &\qquad \qquad  + \frac12\left[R-D(P_{\bar X}Q_{\bar Y|\bar X}\|P_{\bar X}Q_{\bar Y})\right]_+\bigg\} \nonumber \\
  &\quad \le D\left(P_{\bar X}{Q^\star_{\bar Y|\bar X}}_{[k]}\big\|P_{\bar XY} \right) \label{eqn:svb opt sec} \\ 
  &\qquad \qquad  + \frac12\left[R-D\left(P_{\bar X}{Q^\star_{\bar Y|\bar X}}_{[k]}\big\|P_{\bar X}{Q^\star_{\bar Y}}_{[k]} \right)\right]_+  \text{,}  \nonumber
\end{align}
where 
\begin{align}
  {Q^\star_{\bar Y}}_{[k]} = \sum_{x\in \cX} P_{\bar X}(x) {Q^\star_{\bar Y|\bar X}}_{[k]}(\cdot|x) \text{.}
\end{align}
Since both $D(Q_{Y|X}\|P_{Y|X}|P_X)$ and $D(Q_{Y|X}\|Q_Y|P_X) $ are convex in $Q_{Y|X}$ on the finite dimensional space of discrete distributions, it follows that they are both continuous in $Q_{Y|X}$ throughout the relative interior of $\cP(\cY|\cX)$, see, e.g., \cite[Section 7.9]{luenberger1997optimization}. Therefore,\footnote{In order for \eqref{eqn:in here cc} to hold, $Q^{\star}_{Y|X}$ needs to be in the relative interior of $\cP(\cY|\cX)$. If $Q^{\star}_{Y|X}$ is on the boundary of $\cP(\cY|\cX)$, restricting attention to a smaller simplex suffices.}
\begin{align}
  &\inf_{Q_{\bar Y|\bar X} \in \cP_\infty(\cY|P_{\bar X})} \bigg\{ D(P_{\bar X}Q_{\bar Y|\bar X}\|P_{\bar XY}) \nonumber \\ 
  &\qquad \qquad \ + \frac12\left[R-D(P_{\bar X}Q_{\bar Y|\bar X}\|P_{\bar X}Q_{\bar Y})\right]_+\bigg\} \nonumber \\
  &\quad \le \lim_{k\to \infty} \bigg\{ D\left(P_{\bar X}{Q^\star_{\bar Y|\bar X}}_{[k]}\big\|P_{\bar XY}\right) \label{frm:sub opt sequ abv} \\ 
  &\qquad \qquad \  + \frac12\left[R-D\left(P_{\bar X}{Q^\star_{\bar Y|\bar X}}_{[k]}\big\|P_{\bar X}{Q^\star_{\bar Y}}_{[k]}\right)\right]_+\bigg\} \nonumber  \\
  &\quad = D(P_{\bar X} Q^\star_{Y|X}\|P_{\bar XY}) \label{eqn:in here cc} \\ 
  &\qquad \qquad \   + \frac12\left[R-D(P_{\bar X}Q^\star_{Y|X}\|P_{\bar X}Q^\star_Y)\right]_+ \nonumber \\
  &\quad =  \min_{Q_{Y|X} \in \cP(\cY|\cX)} \bigg\{ D(P_{\bar X}Q_{Y|X}\|P_{\bar XY}) \\ 
  &\qquad \qquad \   + \frac12\left[R-D(P_{\bar X}Q_{Y|X}\|P_{\bar X}Q_Y)\right]_+\bigg\} \text{,} \nonumber 
\end{align}
where \eqref{frm:sub opt sequ abv} follows from \eqref{eqn:svb opt sec}; and in \eqref{eqn:in here cc} $Q^\star_Y = \sum_{x \in \cX} P_{\bar X}(x) Q^\star_{Y|X}(\cdot|x ) $. 
\end{proof} 
\section{Lemmas for the Dual Representation and Exponent Comparisons}\label{apdx:alternative representation and comparisons}
\begin{lemma}\label{lem:rel ent and func minimization} Let $U \sim P$, $V \sim Q$ and assume that $f$ is a real valued function that has no internal dependence on the distribution $Q$,
\begin{align}
	\min_{Q} \left\{ D(Q\| P) - \bbE\left[f(V)  \right]  \right \} = - \log \bbE [\exp(f(U)) ] \text{,}
\end{align}
	and the minimizing distribution $Q^\ast $ satisfies
	\begin{align}
		\imath^{}_{Q^\ast\|P}(x) = f(x)- \log \bbE[ \exp(f(U) )] \text{.} \label{eqn:lem1_minimizer}
	\end{align}
\end{lemma}
\begin{proof}
Thanks to Jensen's inequality
	\begin{align}
		&D(Q\| P) - \bbE\left[f(V) \right] \nonumber \\ 
		&\qquad = \bbE [\imath^{}_{Q\|P}(V) - f(V) ] \\
		&\qquad \ge -\log \bbE [\exp(-\imath^{}_{Q\|P}(V)+f(V) ) ] \label{eqn:lw_bd_eqlty1} \\
		&\qquad =- \log \bbE [\exp(f(U))] \text{,}
	\end{align}
	where the inequality in \eqref{eqn:lw_bd_eqlty1} holds with equality when $\imath^{}_{Q\|P}(x)=f(x)- \log \bbE[ \exp(f(U) )] $.
\end{proof}
\begin{corollary}\label{cor:conditional minima} Suppose $(X,Y)\sim P_{X|Y}P_Y$, $(\widetilde X,\widetilde Y)\sim Q_{X|Y} Q_Y $, and $(\widehat X, \widetilde Y) \sim P_{X|Y}Q_Y $, then for any $\lambda \in \bbR $
	
	\begin{align}
		&\min_{Q_{X|Y}} \left\{D(Q_{X|Y}\|P_{X|Y}|Q_Y) - \lambda\bbE\left[\imath^{}_{X;Y}(\widetilde X;\widetilde Y)\right]\right\} \nonumber \\ 
		&\qquad = -\bbE\left[\log\bbE\left[\exp\left(\lambda\,\imath^{}_{X;Y}(\widehat X;\widetilde Y)\right)\middle|\widetilde Y\right]\right] \text{,} \label{eqn:lem:con minimization}
	\end{align}
	and for a fixed $y\in \cY$, the minimizing conditional distribution $Q^\ast_{X|Y} $ satisfies 
	\begin{align}
		&\imath^{}_{Q^\ast_{X|Y}\|P^{}_{X|Y}}(x|y)= \lambda\, \imath^{}_{X;Y} (x;y)  \\ 
		&\qquad \qquad - \log\bbE\left[\exp\left(\lambda\,\imath^{}_{X;Y}(\widehat X;\widetilde Y)\right)\middle| \widetilde Y = y\right] \text{.} \nonumber
	\end{align}
\end{corollary}
\begin{proof}
	For a fixed $y\in \cY$, an application of \Cref{lem:rel ent and func minimization} with
	\begin{align*}
		P(\cdot) &\leftarrow P_{X|Y}(\cdot |y) \text{,} \\
		Q(\cdot) &\leftarrow Q_{X|Y}(\cdot |y) \text{,} \\
		f(\cdot) &\leftarrow \lambda\, \imath^{}_{X;Y} (\cdot\,;y) 
	\end{align*}
	yields
	\begin{align}
		&\min_{Q_{X|Y}} \bigg\{ D(Q_{X|Y}(\cdot|y)\|P_{X|Y}(\cdot|y)) \nonumber \\ 
		&\qquad \qquad \qquad \qquad \qquad - \lambda\bbE\left[\imath^{}_{X;Y}(\widetilde X; \widetilde Y) \middle| \widetilde Y = y  \right] \bigg\} \nonumber \\ 
		&\quad = - \log \bbE\left[\exp \left(\lambda\, \imath^{}_{X;Y}(\widehat X ; \widetilde Y)\right) \middle| \widetilde Y = y\right] \text{.} \label{eqn:exp both side}
	\end{align}
Taking expectation on both sides of \eqref{eqn:exp both side} with respect to $\widetilde Y \sim Q_Y $ gives \eqref{eqn:lem:con minimization}.
\end{proof}
\begin{corollary}\label{cor:applied lem1 no1} Suppose $(X,Y)\sim P_{X|Y}P_Y $, and $(\widehat X, \widetilde Y) \sim P_{X|Y}Q_Y $, then for any $\lambda \in \bbR $
	\begin{align}
		&\min_{Q_Y} \bigg\{ D(Q_Y \| P_Y) \nonumber \\ 
		&\qquad \qquad   - \frac12 \bbE \left[ \log \bbE \left[ \exp\left(\lambda\, \imath^{}_{X;Y}(\widehat X;\widetilde Y)\right) \middle| \widetilde Y\right]\right] \bigg\} \nonumber \\ 
		&\quad = -\log\bbE\left[{\bbE^\frac12\left[\exp\left(\lambda\,\imath^{}_{X;Y}\left(X;Y\right)\right)\middle|Y\right]}\right] \text{,}
	\end{align}
	and the minimizing distribution $Q^\ast_Y $ satisfies 
	\begin{align}
	&\imath^{}_{Q^{\ast}_Y\| P_Y}(y) \nonumber \\ 
	&\quad =\frac12 \log\bbE\left[ \exp\left( \lambda\, \imath^{}_{X;Y}(X;Y)\right)\middle|Y=y\right] \\
	&\qquad \qquad \quad -\log\bbE \left[\bbE^\frac12 \left[\exp\left(\lambda\, \imath^{}_{X;Y}\left(X;Y\right)\right) \middle|Y\right] \right] \text{.} \nonumber
	\end{align}
\end{corollary}
\begin{proof}
Let $P_{X|Y}$ be the fixed random transformation from $\cY$ to $\cX$. Applying \Cref{lem:rel ent and func minimization} with 
\begin{align*}
P(\cdot) &\leftarrow P_Y(\cdot) \text{,} \\
Q(\cdot) &\leftarrow Q_Y(\cdot) \text{,} \\
f(\cdot) &\leftarrow \frac12\log\bbE\left[\exp\left(\lambda\,\imath^{}_{X;Y}(X;Y)\right)\middle|Y=\cdot \right]
\end{align*}
gives the desired result. 
\end{proof}
\begin{corollary}\label{cor:lemma for renyi mutual info}
	Suppose $\lambda \in [0, \infty) $, $(X,Y)\sim P_{X|Y}P_Y $, and $(\widetilde X, \widetilde Y)\sim Q_{X|Y}Q_Y$, then
	\begin{align}
			&\min_{Q_{XY}} \left\{ D(Q_{XY} \| P_{XY}) -\lambda\bbE\big[\imath^{}_{X;Y}(\widetilde X;\widetilde Y)\big]  \right\} \nonumber \\ 
			&\quad  = -\lambda D_{1+\lambda}(P_{XY}\|P_X P_Y) \text{,}
	\end{align}
	where $D_\alpha(P\|Q)$ denotes the R\'enyi divergence (see, e.g., \cite{vanErven2014}) of order $\alpha$ between $P$ and $Q$, and the minimizing distribution $Q^\ast_{XY}$ satisfies
	\begin{align}
	&\imath^{}_{Q^\ast_{XY}\|P^{}_{XY}}(x,y) \nonumber \\ 
	&\quad =\lambda\,\imath^{}_{X;Y}(x;y) - \lambda D_{1+\lambda}(P_{XY}\|P_X P_Y).
	\end{align}
\end{corollary}
\begin{proof}
	Immediate consequence of \Cref{lem:rel ent and func minimization} with 
	\begin{align*}
		Q &\leftarrow Q_{XY} \text{,} \\
		P &\leftarrow P_{XY} \text{,} \\
		f(x,y) &\leftarrow \lambda\, \imath^{}_{X;Y}(x;y) \text{.}
	\end{align*}
\end{proof}
\begin{corollary} \label{cor:lemma for almost renyi mutual info}
Suppose $\lambda' \in [0, \infty) $, $(X,Y)\sim P_{X|Y}P_Y $, and $(\widetilde X, \widetilde Y)\sim Q_{X|Y}Q_Y$, then
	\begin{align}
		&\min_{Q_{XY}} \bigg\{ D(Q_{XY}\|P_{XY} ) - \frac12 D(Q_{XY}\|P_{X|Y}Q_Y) \nonumber \\
		&\qquad \qquad \qquad \qquad \qquad \qquad  -\frac{\lambda'}{2} \bbE \big[\imath^{}_{X;Y}(\widetilde X;\widetilde Y)\big]\bigg\} \nonumber \\ 
		&\quad = -\frac{\lambda'}{2} \widetilde D_{1+\lambda'}(P_{XY}\|P_X P_Y) \text{,}
	\end{align}
	where $\widetilde D_{1+\lambda'}(P_{XY}\|P_X P_Y)$ is defined in \eqref{eqn:def:almost_renyi_mutual_info}, and the minimizing distribution $Q^\ast_{XY} $ satisfies
	\begin{align}
	&\imath^{}_{Q^\ast_{XY}\|P^{}_{XY}}(x,y) = \lambda\, \imath^{}_{X;Y} (x;y)  \\ 
	&\qquad + \frac12 \log\bbE\left[\exp\left(\lambda\, \imath^{}_{X;Y}(X;Y)\right) \middle|Y=y\right]  \nonumber \\ 
	&\qquad \qquad  - \log\bbE\left[\bbE^\frac12\left[\exp\left(\lambda\, \imath^{}_{X;Y}\left(X;Y\right)\right)\middle|Y \right]\right] \nonumber \\ 
	&\qquad \qquad \quad \, - \log \bbE\left[ \exp\left(\lambda\, \imath^{}_{X;Y}(\widehat X; \widetilde Y)\right)\middle|\widetilde Y = y\right] \text{,} \nonumber
	\end{align}
	with $(\widehat X, \widetilde Y) \sim P_{X|Y}Q_Y $.
\end{corollary}
\begin{proof}
	Observe that 
\begin{align}
&\min_{Q_{XY}} \bigg\{ D(Q_{XY} \| P_{XY} ) - \frac12 D(Q_{XY}\|P_{X|Y}Q_Y) \nonumber \\ 
&\qquad \qquad \qquad \qquad \qquad \qquad -\frac{\lambda'}{2} \bbE \left[ \imath^{}_{X;Y}( \widetilde X; \widetilde Y)\right]\bigg\} \nonumber \\
& = \min_{Q_{XY}} \bigg\{ D(Q_Y\|P_Y)+\frac12 D(Q_{X|Y}\|P_{X|Y}|Q_Y) \\ 
&\qquad \qquad \qquad \qquad \qquad \qquad -\frac{\lambda'}{2} \bbE \left[\imath^{}_{X;Y}(\widetilde X; \widetilde Y)\right]\bigg\} \nonumber \\
& = \min_{Q_{\bar Y}} \bigg\{ D(Q_Y \| P_Y)  \\ 
&\qquad \qquad \qquad + \frac12 \min_{Q_{X|Y}} \Big\{ D(Q_{X|Y}\|P_{X|Y}|Q_Y) \nonumber \\ 
&\qquad \qquad \qquad \qquad \qquad \quad \ \ - \lambda' \bbE \left[\imath^{}_{X;Y}(\widetilde X; \widetilde Y) \right] \Big\} \bigg\} \nonumber \\
& = \min_{Q_Y} \bigg\{ D(Q_Y \| P_Y) \label{from:lem 3} \\ 
&\qquad \qquad \, -\frac12 \bbE\left[\log\bbE\left[\exp\left(\lambda'\, \imath^{}_{X;Y}(\widehat X;\widetilde Y)\right )\middle| \widetilde Y\right]\right]\bigg\} \nonumber  \\
& = -\log\bbE\left[\bbE^\frac12 \left[\exp\left(\lambda'\, \imath^{}_{X;Y}\left(X;Y\right)\right)\middle|Y\right]\right] \label{from:lem 2} \\
& = -\frac{\lambda'}{2} \widetilde D_{1+\lambda'}(P_{XY}\|P_X P_Y) \text{,} \label{from:def:almost_renyi info}
\end{align}
	where \eqref{from:lem 3} is the result of \Cref{cor:conditional minima}; \eqref{from:lem 2} is the result of \Cref{cor:applied lem1 no1}; and \eqref{from:def:almost_renyi info} is the definition of $\widetilde D_{1+\lambda'}(P_{XY}\|P_X P_Y)$.
\end{proof}
\begin{corollary}\label{cor:for aleph dual}
  Suppose $(X,Y)\sim P_X P_{Y|X} $, and $(X, \widetilde Y)\sim P_X Q_{Y|X} $, then for any $\lambda \in \bbR $
  \begin{align}
    &\min_{Q_{Y|X}}\Big\{ D(Q_{Y|X}\|P_{Y|X}|P_X) \nonumber \\ 
    &\qquad \qquad \qquad \qquad \qquad  -  \lambda \bbE\left[\imath^{}_{P_{XY}\|P_XS_Y }(X, \widetilde Y)\right]   \Big\} \nonumber \\ 
    & = -\bbE\left[\log \bbE\left[ \exp\left(\lambda\, \imath^{}_{P_{XY}\|P_XS_Y}(X, Y)\right)\middle|X \right]\right ]  \label{eqn:cc:cor for dual} \\
    & = -\lambda \bbE[ D_{1+\lambda}(P_{Y|X}(\cdot|X)\|S_Y)] \text{,}
  \end{align}
  and for a fixed $x \in \cX$, the minimizing conditional distribution $Q_{Y|X}^* $ satisfies 
  \begin{align}
    &\imath^{}_{Q_{Y|X}^*\| P^{}_{Y|X}}(y|x) = \lambda\, \imath^{}_{P_{XY}\|P_XS_Y}(x,y)  \\ 
    &\qquad \quad - \log\bbE\left[ \exp\left(\lambda\, \imath^{}_{P_{XY}\|P_XS_Y}(X, Y)\right)\middle|X = x\right] \text{.} \nonumber
  \end{align}
\end{corollary}
\begin{proof}
  For a fixed $x\in \cX$, an application of \Cref{lem:rel ent and func minimization} with 
  \begin{align*}
		P(\cdot) &\leftarrow P_{Y|X}(\cdot |x) \text{,} \\
		Q(\cdot) &\leftarrow Q_{Y|X}(\cdot |x) \text{,} \\
		f(\cdot) &\leftarrow \lambda\, \imath^{}_{P_{XY}\|P_XS_Y} (x, \cdot) 
	\end{align*}
	yields
\begin{align}
&\min_{Q_{Y|X}}\Big\{ D(Q_{Y|X}(\cdot|x)\|P_{Y|X}(\cdot|x)) \nonumber \\ 
&\qquad \qquad \qquad  - \lambda \bbE\left[\imath^{}_{P_{XY}\|P_XS_Y }( X, \widetilde Y) \middle| X = x \right] \Big\}\nonumber \\ 
&\quad = -\log \bbE\Big[\exp\Big(\lambda\, \imath^{}_{P_{XY}\|P_XS_Y}(X,Y)\Big)\Big|X=x\Big] \text{.} \label{cc:take expectation} 
\end{align}
Taking expectation on both sides of \eqref{cc:take expectation} with respect to $X \sim P_X $ gives \eqref{eqn:cc:cor for dual}.
\end{proof}
\section{Finite Block-length Results}\label{apdx:finite block-length results}
Using simple algebra, the following finite block-length bounds can be deduced from the analysis provided in \Cref{sec:lower bound,sec:upper bound}.
\begin{theorem}\label{thm:finite block-length exact scl}
Fix $n \in \bbN$. Suppose $P_{X^n}\to P_{Y^n|X^n} \to P_{Y^n}$, where the $n$-shot stationary memoryless channel $P_{Y^n|X^n}$ is non-degenerate, i.e., $P_{Y^n|X^n} \neq P_{Y^n}$. For any $R> I(P_X, P_{Y|X})$, let $M=\exp(nR)$, and denote by $\induced (y^n)$ the induced output distribution when a uniformly chosen codeword from the random codebook $\scrC_M^n$ is transmitted through the channel, see \Cref{def:random codebook,def:induced output distribution}. Then,
	\begin{align}
	&\min_{Q_{\bar X\bar Y}\in \cP_n(\cX\times\cY)} \bigg\{ D(Q_{\bar X\bar Y}\|P_{XY}) \nonumber \\ 
	&\qquad \qquad \qquad + \frac12 \left[R-D(Q_{\bar X\bar Y}\|P_XQ_{\bar Y}) \right]_{+}\bigg\}  - \kappa_n \nonumber  \\
		&\quad\le -\frac{1}{n} \log \bbE \left[\left\| \induced - P_{Y^n} \right\|^{}_1\right] \label{eqn:thm:fb:lower bd} \\ 
		&\quad\le  \min_{Q_{\bar X\bar Y}\in \cP_n(\cX\times\cY)} \bigg\{ D(Q_{\bar X\bar Y}\|P_{XY}) \label{eqn:thm:fb:upper bd} \\  
		&\qquad \qquad \qquad  + \frac12 \left[R-D(Q_{\bar X\bar Y}\|P_XQ_{\bar Y}) \right]_{+}\bigg\} + \upsilon_n \text{,} \nonumber
	\end{align}
	where for $r\in (\alpha(R,P_X,P_{Y|X}), R/2)$ and a fixed $\delta\in (0,1) $ that is greater than $\exp(-nR) $,
	\begin{align}
	\kappa_n  &= \frac{|\cX||\cY|}{n}\log(n+1) + \frac1n\log 2 \text{,} \label{eqn:kappa_n second def}  \\
	\alpha_n  &=  \min_{Q_{\bar X\bar Y}\in \cP_n(\cX\times\cY)} \bigg\{ D(Q_{\bar X\bar Y}\|P_{XY}) \\ 
	&\qquad \qquad \qquad \quad \ +\frac12 \left[R-D(Q_{\bar X\bar Y}\|P_XQ_{\bar Y}) \right]_{+} \bigg\} \text{,} \nonumber  \\
	\rho_n    &= \frac{(|\cX|+1)|\cY|}{n}\log(n+1) +\frac1n\log 4  \text{,} \\
	\mu_n     &= \exp(nR) \text{,} \\ 
	a_\epsilon &= \frac{\rme^\epsilon}{(1+\epsilon)^{1+\epsilon}} \text{,} \\ 
	\phi_n    &= \exp(n(\alpha_n +\rho_n)) \\ 
	&\quad \times \left( \frac{1}{\sqrt{\mu_n}} + |\cY|^n a^{\mu_n}_{\delta-\frac1{\mu_n}} + 2(1+\delta)\exp(-nr) \right) \text{,} \nonumber \\
	\upsilon_n &= \rho_n + \frac{\log\rme}{n} \frac{\phi_n}{1-\phi_n} + \frac1n\log(1+\delta) \text{.} \label{eqn:def:ugly upsilon}
	\end{align}
\end{theorem}
\begin{proof}
  The lower bound, \eqref{eqn:thm:fb:lower bd}, easily follows from \eqref{inner max depends on outer max}, \eqref{for:fb:lower}, and \cite[Lemma 2.2]{csiszar2011information}. To see the upper bound, first assemble \eqref{eqn:sumterm}, \eqref{frm:sum of non-neg itms}, \eqref{eqn:fixed grt thn pois -exp(-nr)} and \eqref{for:fb:upper} to get
  \begin{align}
 &(1+\delta)\bbE \left[\left\| \induced - P_{Y^n} \right\|^{}_1\right] \nonumber \\ 
 &\qquad \ge \exp(-n(\alpha_n +\rho_n))\left(1-\phi_n\right)\text{.} \label{eqn:prf:log thm up}
  \end{align}
The result in \eqref{eqn:thm:fb:upper bd} follows after taking $-\frac1n\log $ both sides and noticing that
\begin{align}
\log(1-x)\ge \frac{-x }{1-x}\log\rme \text{.}
\end{align}
\end{proof}
Similarly, following the proof of \Cref{thm:main} along the path paved by Remarks~\ref{rem:cc:replacements lower bound},~\ref{rem:cc:replacements in poissonization},~\ref{rem:cc:replacements pseudo-upper},~and~\ref{rem:cc:replacements depoissonization}; an imitation of the proof of \Cref{thm:finite block-length exact scl} with the replacements
\begin{align*}
 \text{\eqref{for:fb:lower}} &\leftarrow \text{\eqref{for:cc:fb:lower}}\\
 \text{\eqref{for:fb:upper}} &\leftarrow \text{\eqref{for:cc:fb:upper}}
\end{align*}
yields the finite block-length bounds for the constant-composition case analysis as stated in \Cref{thm:cc:finite block-length cons. comp. exact scl}.
\begin{theorem}\label{thm:cc:finite block-length cons. comp. exact scl}
Let $m$ be a fixed integer and $P_{\bar X} \in \cP_m(\cX)$ be a fixed $m$-type.  Fix $n \in m\bbN$. Suppose that $R_{\breve X^n}$ is a constant-composition distribution based on $P_{\bar X}$ defined as in \eqref{eqn:def:cc distribution}, and let $R_{\breve X^n}\to P^{}_{Y^n|X^n} \to R_{\breve Y^n} $, where the $n$-shot stationary discrete memoryless channel $P_{Y^n|X^n}$ is non-degenerate, i.e., $P^{}_{Y^n|X^n} \neq R_{\breve Y^n}$. For any $R>I(P_{\bar X}, P_{Y|X}) $, let $M=\exp(nR)$, and denote by $\ccinduced $ the constant-composition induced output distribution when a uniformly chosen codeword from the random constant-composition codebook $\scrD^n_M$ is transmitted through the channel, see \Cref{def:random cc codebook,def:cc induced output distribution}. Then,
  \begin{align}
  &\min_{Q_{\bar Y|\bar X} \in \cP(\cY|P_{\bar X})} \bigg\{ D(P_{\bar X}Q_{\bar Y|\bar X}\|P_{\bar XY}) \nonumber \\ 
  &\qquad \qquad \quad + \frac12[R-D(P_{\bar X}Q_{\bar Y|\bar X}\|P_{\bar X}Q_{\bar Y})]_+ \bigg\} -\breve \eta_n \nonumber \\
  &\quad \le  -\frac{1}{n}\log \bbE\left[\left\| \ccinduced - R_{\breve Y^n} \right\|^{}_1\right] \label{eqn:thm:cc:fb:lower bd} \\
  &\quad \le \min_{Q_{\bar Y|\bar X} \in \cP(\cY|P_{\bar X})} \bigg\{ D(P_{\bar X}Q_{\bar Y|\bar X}\|P_{\bar XY}) \label{eqn:thm:cc:fb:upper bd} \\ 
  &\qquad \qquad \quad + \frac12[R-D(P_{\bar X}Q_{\bar Y|\bar X}\|P_{\bar X}Q_{\bar Y})]_+ \bigg\} + \breve \upsilon_n  \text{,}  \nonumber 
  \end{align}
  where $P_{\bar X}  \to Q_{\bar Y|\bar X} \to Q_{\bar Y} $, and for $r\in (\aleph(R,P_{\bar X},P_{Y|X}), R/2)$ and a fixed $\delta\in (0,1) $ that is greater than $\exp(-nR) $,
  \begin{align}
    [f]_+ &= \max\{0,f\} \text{,} \\
    \breve \eta_n  &=  \frac{|\cX|(2+3|\cY|)}{2n}\log(n+1) \text{,} \\
    \aleph_n        &=  \min_{Q_{\bar Y|\bar X}} \bigg\{ D(P_{\bar X}Q_{\bar Y|\bar X}\|P_{\bar XY}) \\ 
    &\qquad \qquad \quad + \frac12[R-D(P_{\bar X}Q_{\bar Y|\bar X}\|P_{\bar X}Q_{\bar Y})]_+  \bigg\}  \text{,} \nonumber \\
    \breve \rho_n   &= \frac{|\cX|+2|\cX||\cY|+|\cY|}{2n} \log(n+1) + \frac1n \log 4  \text{,} \\
    \mu_n           &= \exp(nR) \text{,} \\ 
	a_\epsilon      &= \frac{\rme^\epsilon}{(1+\epsilon)^{1+\epsilon}} \text{,} \\ 
	\breve \phi_n   &= \exp(n(\aleph_n +\breve \rho_n)) \\ 
	& \times \left(\frac{1}{\sqrt{\mu_n}} + |\cY|^n a^{\mu_n}_{\delta-\frac1{\mu_n}} + 2(1+\delta)\exp(-nr) \right) \text{,} \nonumber \\
	\breve \upsilon_n &= \breve \rho_n + \frac{\log\rme}{n} \frac{\breve \phi_n}{1-\breve \phi_n} + \frac1n\log(1+\delta) \text{.}
  \end{align}
\end{theorem}
\begin{rem}\label{rem:upper bound constants in finite blocklenth}
  The $P_XP_{Y|X}$ (respectively, $P_{\bar X}P_{Y|X}$) dependence of the upper bound constant $\upsilon_n$ (respectively, $\breve \upsilon_n$) in \Cref{thm:finite block-length exact scl} (respectively, in \Cref{thm:cc:finite block-length cons. comp. exact scl}) is due to the discontinuity of the exponent in the degenerate channel case, see \Cref{rem:degenerate channel case discontinuity}.
\end{rem}
\end{appendices}
\section*{Acknowledgement}
The authors would like to thank Alex Dytso, Amin Gohari and Jingbo Liu for their valuable comments on an early manuscript. Mani Bastani Parizi is acknowledged for graciously sending his codes for the computation of relative entropy variant of the constant composition exponent which are used in plotting \Cref{fig:cc:aleph vs all other cc}. Further thanks to the developers of Texpad and the TexpadTex engine for turning the typesetting of this paper into an enjoyment.

\bibliographystyle{IEEEtran}
\bibliography{bibliography}

\begin{IEEEbiographynophoto}{Semih Yagli}
	received his Bachelor of Science degree in Electrical and Electronics Engineering in 2013, his Bachelor of Science degree in Mathematics in 2014 both from Middle East Technical University and his Master of Arts degree in Electrical Engineering in 2016 from Princeton University.
	
	Currently, he is pursuing his Ph.D. degree in Electrical Engineering at Princeton University under the supervision of H. Vincent Poor. His research interest include information theory, optimization, statistical modeling, and unsupervised machine learning. 
\end{IEEEbiographynophoto}

\begin{IEEEbiographynophoto}{Paul Cuff (S'08-M'10)}
received the B.S. degree in electrical engineering from Brigham Young University, Provo, UT, in 2004 and the M.S. and Ph.D. degrees in electrical engineering from Stanford University in 2006 and 2009. From 2009 to 2017 he was an assistant professor of electrical engineering at Princeton University.  Since 2017 he has been a member of the general research group at Renaissance Technologies.

As a graduate student, Dr. Cuff was awarded the ISIT 2008 Student Paper Award for his work titled Communication Requirements for Generating Correlated Random Variables and was a recipient of the National Defense Science and Engineering Graduate Fellowship and the Numerical Technologies Fellowship. As faculty he received the NSF Career Award in 2014 and the AFOSR Young Investigator Program Award in 2015.
\end{IEEEbiographynophoto}

\end{document}